\documentclass[12pt,american]{article}
\usepackage[T1]{fontenc}
\usepackage{geometry}
\usepackage[active]{srcltx}
\usepackage{booktabs}
\usepackage{babel}
\usepackage{mathrsfs}
\usepackage{enumitem}
\usepackage{amsmath}
\usepackage{amsthm}
\usepackage{bbm}
\usepackage{amssymb}
\usepackage{setspace}
\usepackage[authoryear]{natbib}
\usepackage[unicode=true,pdfusetitle,
 bookmarks=true,bookmarksnumbered=false,bookmarksopen=false,
 breaklinks=false,pdfborder={0 0 1},backref=false,colorlinks=false]
 {hyperref}
\hypersetup{
 colorlinks,linkcolor={blue!70!black},citecolor={blue!50!black},urlcolor={blue!60!black}}

\makeatletter
\theoremstyle{plain}
\newtheorem{assumption}{\protect\assumptionname}
\theoremstyle{plain}
\newtheorem{prop}{\protect\propositionname}
\theoremstyle{plain}
\newtheorem{thm}{\protect\theoremname}
\theoremstyle{plain}

\theoremstyle{plain}
\newtheorem{lem}{\protect\lemmaname}
\theoremstyle{definition}
\newtheorem{defn}{\protect\definitionname}
\theoremstyle{remark}
\newtheorem{rem}{\protect\remarkname}
\theoremstyle{definition}
 \newtheorem{example}{\protect\examplename}

\usepackage{lmodern}
\usepackage{babel}
\usepackage{float}
\usepackage{mathrsfs}
\usepackage{breakurl}
\usepackage{bbm}
\usepackage{tikz}
\usepackage{centernot}
\usepackage{amssymb,amsmath,amsthm,enumitem}

\allowdisplaybreaks

\makeatletter
\newcommand{\customlabel}[2]{%
   \protected@write \@auxout {}{\string \newlabel {#1}{{#2}{\thepage}{#2}{#1}{}} }%
   \hypertarget{#1}{}
}
\makeatother

\makeatother

\providecommand{\assumptionname}{Assumption}
\providecommand{\corollaryname}{Corollary}
\providecommand{\definitionname}{Definition}
\providecommand{\examplename}{Example}
\providecommand{\lemmaname}{Lemma}
\providecommand{\propositionname}{Proposition}
\providecommand{\remarkname}{Remark}
\providecommand{\theoremname}{Theorem}

\begin{document}
\global\long\def\a{\alpha}%
 
\global\long\def\b{\beta}%
 
\global\long\def\g{\gamma}%
 
\global\long\def\d{\delta}%
 
\global\long\def\e{\epsilon}%
 
\global\long\def\l{\lambda}%
 
\global\long\def\t{\theta}%
 
\global\long\def\o{\omega}%
 
\global\long\def\s{\sigma}%

\global\long\def\G{\Gamma}%
 
\global\long\def\D{\Delta}%
 
\global\long\def\L{\Lambda}%
 
\global\long\def\T{\Theta}%
 
\global\long\def\O{\Omega}%
 
\global\long\def\R{\mathbb{R}}%
 
\global\long\def\N{\mathbb{N}}%
 
\global\long\def\Q{\mathbb{Q}}%
 
\global\long\def\I{\mathbb{I}}%
 
\global\long\def\P{\mathbb{P}}%
 
\global\long\def\E{\mathbb{E}}%
\global\long\def\B{\mathbb{\mathbb{B}}}%
\global\long\def\S{\mathbb{\mathbb{S}}}%
\global\long\def\V{\mathbb{\mathbb{V}}\text{ar}}%
 
\global\long\def\GG{\mathbb{G}}%
\global\long\def\TT{\mathbb{T}}%

\global\long\def\X{{\bf X}}%
\global\long\def\cX{\mathscr{X}}%
 
\global\long\def\cY{\mathscr{Y}}%
 
\global\long\def\cA{\mathscr{A}}%
 
\global\long\def\cB{\mathscr{B}}%
\global\long\def\cF{\mathscr{F}}%
 
\global\long\def\cM{\mathscr{M}}%
\global\long\def\cN{\mathcal{N}}%
\global\long\def\cG{\mathcal{G}}%
\global\long\def\cC{\mathcal{C}}%
\global\long\def\sp{\,}%

\global\long\def\es{\emptyset}%
 
\global\long\def\mc#1{\mathscr{#1}}%
 
\global\long\def\ind{\mathbf{\mathbbm1}}%
\global\long\def\indep{\perp}%

\global\long\def\any{\forall}%
 
\global\long\def\ex{\exists}%
 
\global\long\def\p{\partial}%
 
\global\long\def\cd{\cdot}%
 
\global\long\def\Dif{\nabla}%
 
\global\long\def\imp{\Rightarrow}%
 
\global\long\def\iff{\Leftrightarrow}%

\global\long\def\up{\uparrow}%
 
\global\long\def\down{\downarrow}%
 
\global\long\def\arrow{\rightarrow}%
 
\global\long\def\rlarrow{\leftrightarrow}%
 
\global\long\def\lrarrow{\leftrightarrow}%

\global\long\def\abs#1{\left|#1\right|}%
 
\global\long\def\norm#1{\left\Vert #1\right\Vert }%
 
\global\long\def\rest#1{\left.#1\right|}%

\global\long\def\bracket#1#2{\left\langle #1\middle\vert#2\right\rangle }%
 
\global\long\def\sandvich#1#2#3{\left\langle #1\middle\vert#2\middle\vert#3\right\rangle }%
 
\global\long\def\turd#1{\frac{#1}{3}}%
 
\global\long\def\ellipsis{\textellipsis}%
 
\global\long\def\sand#1{\left\lceil #1\right\vert }%
 
\global\long\def\wich#1{\left\vert #1\right\rfloor }%
 
\global\long\def\sandwich#1#2#3{\left\lceil #1\middle\vert#2\middle\vert#3\right\rfloor }%

\global\long\def\abs#1{\left|#1\right|}%
 
\global\long\def\norm#1{\left\Vert #1\right\Vert }%
 
\global\long\def\rest#1{\left.#1\right|}%
 
\global\long\def\inprod#1{\left\langle #1\right\rangle }%
 
\global\long\def\ol#1{\overline{#1}}%
 
\global\long\def\ul#1{\underline{#1}}%
 
\global\long\def\td#1{\tilde{#1}}%
\global\long\def\bs#1{\boldsymbol{#1}}%

\global\long\def\upto{\nearrow}%
 
\global\long\def\downto{\searrow}%
 
\global\long\def\pto{\overset{p}{\longrightarrow}}%
 
\global\long\def\dto{\overset{d}{\longrightarrow}}%
 
\global\long\def\asto{\overset{a.s.}{\longrightarrow}}%

\setlength{\abovedisplayskip}{6pt} \setlength{\belowdisplayskip}{6pt}
\title{ReLU-Based and DNN-Based \\Generalized Maximum Score Estimators\thanks{We thank Karun Adusumilli, Joel Horowitz, Simon Lee, Charles Manski, Elie Tamer, Yuanyuan Wan, and seminar participants at Columbia and Penn for helpful comments and suggestions.}}
 %under Multi-Index Single-Crossing Conditions}
\author{Xiaohong Chen\thanks{Chen: Department of Economics and Cowles Foundation for Research in Economics,
Yale University, 28 Hillhouse Ave, New Haven, CT 06511, USA, xiaohong.chen@yale.edu}$\ $, Wayne Yuan Gao\thanks{Gao: Department of Economics, University of Pennsylvania, 133 S 36th
St., Philadelphia, PA 19104, USA, waynegao@upenn.edu.}$\ $ and Likang Wen\thanks{Wen: Department of Applied Mathematics and Statistics, Johns Hopkins University, 3400 N Charles St,
Baltimore, MD 21218, USA, lwen11@jh.edu.}\textbf{}\\
\textbf{~}}
\maketitle

% \begin{abstract}
% We propose a rectified-linear-unit-based maximum score (RMS) estimator for
% semiparametric models characterized by sign-alignment restrictions between a
% nonparametric component and a parametric index. In the binary choice model with
% a conditional median restriction, we replace Manski’s discontinuous indicator
% criterion with a composite ReLU criterion that preserves identification but is
% piecewise smooth with quadratic curvature at the truth. Under standard
% smoothness and design conditions, the leading term in the asymptotic expansion
% reduces to an integral of the first-stage function over a $(d-1)$-dimensional
% hyperplane, yielding an intermediate “one-dimensional’’ rate
% $n^{-s/(2s+1)}$, asymptotic normality, and bootstrap validity. Because the RMS
% criterion is continuous and Lipschitz, it can be computed by modern
% gradient-based optimization methods, offering substantial computational gains
% over traditional maximum score. We further embed the approach in a general
% multi-index single-crossing framework and show that the effective rate remains
% unchanged in a $J$-index setting.
% \end{abstract}

\begin{abstract}

\noindent We propose a new formulation of the maximum score estimator
that uses compositions of rectified linear unit (ReLU) functions, instead of indicator functions as in \citet*{manski1975maximum,manski1985semiparametric},
to encode the sign alignment restrictions. Since the ReLU function
is Lipschitz, our new ReLU-based maximum score criterion function is substantially easier to optimize using standard 
gradient-based optimization pacakges. We also show that our ReLU-based maximum score (RMS) estimator can be generalized to an umbrella framework defined by multi-index single-crossing (MISC) conditions, while the original maximum score estimator cannot be applied. We establish the $n^{-s/(2s+1)}$ convergence rate and asymptotic normality for the RMS estimator under order-$s$ Holder smoothness. In addition, we propose an alternative estimator using a further reformulation of RMS as a special layer in a deep neural network (DNN) architecture, which allows the estimation procedure to be implemented via state-of-the-art software and hardware for DNN. \\ %We demonstrate the finite-sample performance of our ReLU-based and DNN-based estimator via simulations.\\
%\textbf{~}\\
\end{abstract}

\section{\label{sec:Intro}Introduction}

In a sequence of papers, \citet*{manski1975maximum,manski1985semiparametric}
proposed and analyzed the properties of the \emph{maximum-score estimator}
in the context of semiparametric discrete choice models. To be specific,
consider the following canonical binary choice model
\begin{equation}
y_{i}=\ind\left\{ X_{i}^{'}\t_{0}\geq\e_{i}\right\} \label{eq:BinChoice}
\end{equation}
under the conditional median restriction $\text{med}\left(\rest{\e_{i}}X_{i}\right)=0$.
The key idea underlying the maximum score estimator is to exploit
the following identifying restriction,
\begin{equation}
h_{0}\left(X_{i}\right):=\E\left[\rest{y_{i}-\frac{1}{2}}X_{i}\right]\gtrless0\quad\iff\quad X_{i}^{'}\t_{0}\gtrless0,\label{eq:Mono_Equiv}
\end{equation}
which is a \emph{sign alignment} restriction between the function
$h_{0}$ and the the parametric index $X_{i}^{'}\t_{0}$. \citet*{manski1975maximum,manski1985semiparametric,manski1987semiparametric}
encodes this sign alignment restriction into the following population
criterion function, 
\begin{align}
Q_{MS}\left(\t\right) & :=\E\left[h_{0}\left(X_{i}\right)\ind\left\{ X_{i}^{'}\t>0\right\} \right]=\E\left[\left(y_{i}-\frac{1}{2}\right)\ind\left\{ X_{i}^{'}\t>0\right\} \right],\label{eq:Q_MS}
\end{align}
which is constructed by multiplying the function $h_{0}$ with an
indicator function of the index $X_{i}^{'}\t_{0}$ along with the
Law of Iterated Expectation. Then, $\t_{0}$ is a maximizer of $Q_{MS}\left(\t\right)$.

To see more clearly why $\t_{0}$ maximizes $Q_{MS}\left(\t\right)$,
consider the following decomposition $h_{0}\left(X_{i}\right)\equiv\left[h_{0}\left(X_{i}\right)\right]_{+}-\left[-h_{0}\left(X_{i}\right)\right]_{+},$
where $\left[t\right]_{+}:=\max\left(t,0\right)$ denotes the rectified
linear unit (ReLU) function. Then, the population criterion $Q_{MS}$
can be correspondingly decomposed as $Q_{MS}\left(\t\right)=Q_{MS+}\left(\t\right)+Q_{MS-}\left(\t\right)$
with
\begin{align*}
Q_{MS+}\left(\t\right) & =\E\left[\left[h_{0}\left(X_{i}\right)\right]_{+}\ind\left\{ X_{i}^{'}\t>0\right\} \right]\leq\E\left[\left[h_{0}\left(X_{i}\right)\right]_{+}\right]=Q_{MS+}\left(\t_{0}\right),\\
Q_{MS-}\left(\t\right) & =-\E\left[\left[-h_{0}\left(X_{i}\right)\right]_{+}\ind\left\{ X_{i}^{'}\t>0\right\} \right]\leq0=Q_{MS-}\left(\t_{0}\right).
\end{align*}
In words, the multiplication of $h_{0}$ with the indicator on $X_{i}^{'}\t_{0}$
precisely extracts the positive part of $h_{0}$ at the true $\t_{0}$,
and hence $Q_{MS}\left(\t\right)\leq\E\left[\left[h_{0}\left(X_{i}\right)\right]_{+}\right]=Q_{MS}\left(\t_{0}\right).$

In this paper, we propose a different population criterion function
that encodes exactly the same sign alignment restriction \eqref{eq:Mono_Equiv}
above. However, instead of using multiplication with indicator functions
on $X_{i}^{'}\t_{0}$ as in \eqref{eq:Q_MS}, our new formulation
employs compositions of ReLU functions. Specifically, define
\begin{align}
g_{+,\t,h}\left(x\right):=\left[h\left(x\right)-\left[-x^{'}\t\right]_{+}\right]_{+},\quad & g_{-,\t,h}\left(x\right):=\left[-h\left(x\right)-\left[x^{'}\t\right]_{+}\right]_{+},\label{eq:def_g}
\end{align}
with $Q_{+}\left(\t\right):=\E\left[g_{+,\t,h_{0}}\left(X_{i}\right)\right]$,
$Q_{-}\left(\t\right):=\E\left[g_{-,\t,h_{0}}\left(X_{i}\right)\right]$,
and 
\begin{equation}
Q\left(\t\right):=Q_{+}\left(\t\right)+Q_{-}\left(\t\right),\label{eq:def_Q}
\end{equation}
Clearly, both $g_{+}$ and $g_{-}$, and thus $Q_{+}$ and $Q_{-}$,
are by construction nonnegative. 

To see why $\t_{0}$ is also a maximizer of $Q_{MS}\left(\t\right)$,
first consider the case when $h_{0}\left(X_{i}\right)>0$. By \eqref{eq:Mono_Equiv},
\begin{align*}
h_{0}\left(X_{i}\right)>0\  & \iff\ X_{i}^{'}\t_{0}>0\ \iff\ \left[-X_{i}^{'}\t_{0}\right]_{+}=0\ \imp\ h_{0}\left(X_{i}\right)-\left[-X_{i}^{'}\t_{0}\right]_{+}=\left[h_{0}\left(X_{i}\right)\right]_{+},
\end{align*}
and thus
\[
0\leq g_{+,\t,h_{0}}\left(X_{i}\right)=\left[h_{0}\left(X_{i}\right)-\left[-X_{i}^{'}\t\right]_{+}\right]_{+}\leq\left[h_{0}\left(X_{i}\right)\right]_{+}=g_{+,\t_{0},h_{0}}\left(X_{i}\right).
\]
Furthermore, when $h_{0}\left(X_{i}\right)>0$, the negative part
degenerates to 0, i.e.,
\[
g_{-,\t,h_{0}}\left(x\right)=\left[-h_{0}\left(X_{i}\right)-\left[X^{'}\t\right]_{+}\right]_{+}\equiv0,
\]
regardless of the parameter value $\t$. Similarly, the opposite holds
for the case of $h_{0}\left(X_{i}\right)<0$. Together, we have
\begin{align*}
Q_{+}\left(\t\right) & =\E\left[\left[h_{0}\left(X_{i}\right)-\left[-X_{i}^{'}\t\right]_{+}\right]_{+}\right]\leq\E\left[\left[h_{0}\left(X_{i}\right)\right]_{+}\right]=Q_{+}\left(\t_{0}\right),\\
Q_{-}\left(\t\right) & =\E\left[\left[-h_{0}\left(X_{i}\right)-\left[X_{i}^{'}\t\right]_{+}\right]_{+}\right]\leq\E\left[\left[-h_{0}\left(X_{i}\right)\right]_{+}\right]=Q_{-}\left(\t_{0}\right),
\end{align*}
which implies that $Q\left(\t\right)\leq\E\left[\left|h_{0}\left(X_{i}\right)\right|\right]=Q\left(\t_{0}\right).$
Hence, our ReLU-based criterion $Q$, even though different from the
original maximum score criterion $Q_{MS}$ above, also incorporates
the identifying restriction about $\t_{0}$ and can thus serve as
a valid population criterion. 

More generally, in a $J$-index setting we let
\[
X_i := (X_{i1},\dots,X_{iJ}) \in \mathcal{X} \subset \mathbb{R}^{d\times J},
\]
and write $x=(x_1,\dots,x_J)$ for a generic realization. For a generic function
$h:\mathcal{X}\to\mathbb{R}$ and direction $\theta\in\Theta\subset\mathbb{S}^{d-1}$,
we define
\begin{align}
g_{+,\theta,h}(x_1,\dots,x_J)
&:= \biggl[h(x_1,\dots,x_J) - \Bigl(\min_{1\le j\le J}(-x_j'\theta)_+\Bigr)\biggr]_+,
\nonumber \\
g_{-,\theta,h}(x_1,\dots,x_J)
&:= \biggl[-h(x_1,\dots,x_J) - \Bigl(\min_{1\le j\le J}(x_j'\theta)_+\Bigr)\biggr]_+.
\label{eq:def_g_J}
\end{align}
The corresponding $J$-index RMS population criterion is $Q_J(\theta) := Q_J^+(\theta)+Q_J^-(\theta)$ with
\begin{equation}
Q_J^+(\theta) := E\bigl[g_{+,\theta,h_0}(X_i)\bigr],\qquad
Q_J^-(\theta) := E\bigl[g_{-,\theta,h_0}(X_i)\bigr].
\label{eq:def_Q_J}
\end{equation}
In the single-index case $J=1$, $X_i$ reduces to a single vector $X_i\in\mathbb{R}^d$,
$g_{+,\theta,h}$ and $g_{-,\theta,h}$, and $Q_J(\theta)$ reduce to those defined in \eqref{eq:def_g} and \eqref{eq:def_Q}.

The main focus of this paper is to show how this new ReLU-based population
criterion $Q$, as defined by \eqref{eq:def_g}-\eqref{eq:def_Q_J},
can be used for the identification, estimation and inference of $\t_{0}$, and demonstrate that this new approach relates to, differs from, and
improves upon the existing approach based on $Q_{MS}$.

~

We first focus on the binary choice
setting in Section \ref{sec:BinChoice}, which is not only a topic of important
interest on its own, but also serves as a canonical setup where our
new ReLU-based estimator can be related to the original maximum score
(MS) estimator and its previous variants in a clear manner. 

Under the binary choice setting, we propose the \emph{ReLU-based maximum
score} (RMS) estimator as a semiparametric two-stage M-estimator based
on the population criterion $Q$. Specifically, in the first stage,
we obtain an estimator $\hat{h}$ of $h_{0}$ via nonparametric regression
of $Y_{i}-\frac{1}{2}$ on $X_{i}$. Then, we define the sample criterion
function $\hat{Q}$ as the sample analog of $Q$ with $\hat{h}$ plugged
in for $h_{0}$, and obtain the RMS estimator $\hat{\t}$ as the maximizer
of the sample criterion function $\hat{Q}$ in the second stage. We
establish the convergence rate and asymptotic normality for the RMS estimator under lower-level conditions on
the primitives of the binary choice model, with $\hat{h}$ given by
kernel or linear series estimators. 

In particular, we show that, under appropriate conditions, the RMS
estimator is asymptotically normal with rate of convergence as fast as $n^{-\frac{s}{2s+1}}$ (with $s$ being the imposed order of smoothness).
This rate is slower than the $\sqrt{n}$ rate but faster than the
$n^{1/3}$-rate of the original MS estimator \citep{kim1990cube},
and it coincides with the rate of the \emph{smoothed maximum score}
(SMS) estimator in \citet*{horowitz1992smoothed}. The RMS and SMS
estimators are conceptually similar in the sense that both exploit
additional smoothness conditions (on $h_{0}$, in particular) relative
to the original MS estimator, which leads to the accelerated convergence
rates.\footnote{Recall also from \citet*{horowitz1992smoothed} that the rate $n^{-\frac{s}{2s+1}}$
cannot be further improved upon in the minimax sense.} However, the asymptotic theory of the RMS estimator differs significantly
from that for the SMS estimator given the very different forms of
population and sample criterion functions involved. 

In particular, the intermediate level of (non-)smoothness in the ReLU
function turns out to be a key driver of the asymptotic behavior of
the RMS estimator. First, the ``kink'' of the ReLU function at 0
(or more precisely, a non-zero first-order derivative from one side)
is essential for the locally quadratic curvature of the population
criterion function around the true parameter $\t_{0}$. Second, the
Lipschitz continuity of ReLU functions, in contrast with the discontinuous
indicator function, translates small deviations into small deviations,
which is key for a stochastic equicontinuity condition that reduces
the impact of the first-stage nonparametric estimation errors on the
second stage and helps with the convergence rate as well as the asymptotic
normality (instead of a Chernoff-type asymptotic distribution). Third,
the almost-everywhere differentiability of the ReLU function enables
the characterization of the leading term in the asymptotic analysis
as a plug-in estimator of an integration functional of the nonparametric
function $h_{0}\left(x\right)$ over a $\left(d-1\right)$-dimensional
hyperplane (with $d$ being the dimension of $X_{i}$, i.e., the dimension
of the first-stage nonparametric estimation of $h_{0}$). This integral
averages the first-stage estimation error in $\hat{h}$ over a $\left(d-1\right)$-dimensional
space, thus accelerating the convergence to the rate of 1-dimensional
nonparametric estimation, which is the fundamental driver of the final
$n^{-\frac{s}{2s+1}}$ rate of the RMS estimator.\footnote{Relatedly, the asymptotic theory of the SMS estimator \citep{horowitz1992smoothed}
is also driven by the convergence rate of 1-dimensional nonparametric
(kernel) estimation.}

We then (in Section~3) generalize the RMS estimator to an umbrella econometric framework characterized by multi-index single-crossing (MISC) conditions proposed in \cite{gao2020robust}. We show that MISC conditions arise naturally in a wide range of econometric models, and are particularly powerful in multi-index discrete choice and panel multinomial choice settings. In particular, the MISC framework underlies the identification and estimation strategy in \cite{gao2020robust} and \cite*{gao2023logical}, where multi-index single-crossing restrictions are exploited to obtain semiparametric identification in panel multinomial choice models. Our analysis provides a complementary perspective by showing how ReLU-based maximum score ideas can be embedded in the MISC framework and extended to a broad class of models beyond the binary choice benchmark.

Beyond the traditional two-step semiparametric implementation, we also show in Section~\ref{sec:NN} how the RMS/MISC framework can be embedded in a multi-layer neural network architecture. In particular, we construct a special ``RMS layer'' that takes as input a flexible first-stage network $h(x)$ and a low-dimensional direction $\theta$, and applies the composite ReLU transformation that encodes the sign-alignment or MISC restriction. This provides a concrete example of how
economically meaningful low-dimensional parameters can be built into (and estimated within) deep neural networks (DNN) using standard machine learning toolkits. In this way, the paper speaks directly to the broader literatures on interpretable deep learning, by demonstrating how modern neural networks can be used to capture rich nonparametric structure without sacrificing identification for the structural index parameter.

% We then (in Section \ref{sec:MISC}), generalize the R-MS estimator
% to an umbrella econometric framework characterized by \emph{multi-index
% single-crossing }(MISC) conditions. We show that MISC conditions arise
% naturally in a wide range of econometric models, and are particularly
% powerful . \citet*{gao2019robust} and \citet*{gao2023logical}

~

Our paper contributes directly to the econometric literature on maximum
score (MS) estimators, dating back to \citet{manski1975maximum,manski1985semiparametric},
and \citet*{kim1990cube}. Of particular relevance is the line of
research on the variants of the MS estimator with different forms
of smoothing. To our best knowledge, our paper is the first to propose
the ReLU-based formulation introduced above, which builds an intermediate
level of smoothness directly into the population criterion. Previously,
\citet*{horowitz1992smoothed} proposes the SMS estimator, where the
indicator function in the MS (sample) criterion is replaced by a smooth
sigmoid function with a bandwidth parameter, and establishes the accelerated
convergence rate and asymptotic normality of the SMS estimator. \citet*{blevins2013local}
works with a local nonlinear least square formulation of the SMS estimator,
and uses debiasing techinques to obtain the SMS convergence rate. \citet*{chen2015binary}
reformulates the sign alignment restriction as a local conditional
moment condition and proposes a corresponding estimator based on local
polynomial smoothing. \citet*{jun2017integrated} considers the integrated
score estimator, a quasi-Bayes estimator where smoothing is achieved
through integration of the MS criterion. Another set of related work
focuses on the inference problem, given that standard bootstrap is
known to be invalid for the MS estimator \citep*{abrevaya2005bootstrap}:
\citet*{horowitz2002bootstrap} establishes bootstrap consistency for
the SMS estimator, \citet*{patra2018consistent} formulates a smoothed
bootstrap procedure for the MS estimator using a semiparametric two-stage
estimator to center the bootstrap samples,\footnote{This semiparametric two-stage estimator in \citet*{patra2018consistent},
defined in their equation (5), utilizes a first-stage nonparametric
estimation of $h_{0}$, which is plugged in along with a nonparametric
density estimator to obtain an integrated estimator of the MS population
criterion function. However, this estimator is then used for the bootstrap
of the original MS estimator, and its properties were not fully developed
in \citet*{patra2018consistent}. } while \citet*{cattaneo2020bootstrap} proposes an alternative approach
to obtain bootstrap consistency by modifying an asymptotically non-random
component of the MS sample criterion. None of the papers cited above
considers our ReLU-based formulation. As discussed above, this new formulation not only leads to a ``more
smooth'' population criterion function that provides both theoretical
and computational advantages, but also greatly generalizes the scope
of applications to which the key idea of maximum score estimation
can be applied.

This paper also builds upon and contributes to the long line of econometric
literature on semiparametric M estimation and inference: see, for
example, \citet*{newey1994large}, \citet*{chen2007sieve}, \citet*{ichimura2007implementing},
and \citet*{kosorok2007introduction} for general surveys on this
topic. In particular, this paper is related to previous work that
analyzes nonsmooth criterion functions, such as \citet*{chen2003estimation},
\citet*{ichimura2010characterization,ichimura2018corrigendum}, \citet*{seo2018local},
and \citet*{delsol2020semiparametric}. A distinct feature of this
paper is the intermediate level of smoothness (``Lipschitz with a
kink'') of the ReLU function leads to the intermediate convergence
rate of the RMS estimator, which is faster than the cubic-root-or-slower
rates obtained in \citet*{kim1990cube}, \citet*{seo2018local} and
the example considered in \citet*{delsol2020semiparametric} (with
``less smooth'' criterion functions), but slower than the root-$n$
rate considered by \citet*{chen2003estimation} and \citet*{ichimura2010characterization,ichimura2018corrigendum}
(with ``more smooth'' criterion functions). More specifically, we show how the ``Lipschitz-with-a-kink'' property of the ReLU function leads to a characterization of the leading term in the RMS asymptotics as a nonparametric plug-in estimator of a lower-dimensional integral functional, and how this lower-dimensional integral becomes the key driver of the final intermediate convergence rate. Our results on the convergence of nonparametric integral functionals over lower-dimensional hyperplanes are of independent interest, which is closely related to the general theory of semiparametric learning of integral functionals on submanifolds developed in \cite{chen2025semiparametric}, which explicitly relates the convergence rate to the dimension of the underlying submanifold. Our contribution also supplements related work in the statistics literature on the estimation of integrals on level sets, which mostly focus on kernel regressions \citep*{dau2020exact} or density estimation \citep*{qiao2021nonparametric}.

Our DNN-based maximum score estimator under the MISC condition framework also  speaks directly to the broader machine learning literature on interpretability of deep neural networks (DNN). Surveys such as \cite{fan2021interpretability}
and \cite{zhang2021survey} review a wide range of interpretability tools, which mostly focus on explaining \emph{predictions} or internal
representations, but not on identifying or conducting inference on structural low-dimensional parameters inside a network.  In this sense, the DNN-based MISC estimator offer a way to bridge the gap between the interpretability and uncertainty literatures in deep learning and the semiparametric inference literature in econometrics. They allow researchers to use modern DNN to capture rich nonlinearities and heterogeneity in
the data, while still retaining (i) an interpretable, low-dimensional parameter $\theta$ that encodes economically meaningful structure, and (ii) a rigorous large-sample theory that supports conventional confidence intervals and hypothesis tests for that parameter.

%At the same time, a parallel literature in deep learning emphasizes the importance of uncertainty estimation: see, for example, a recent survey by \cite{gawlikowski2023survey}. 
%Our analysis is complementary: instead of focusing on uncertainty in the \emph{predictions} of a black-box network, we provide asymptotic theory and inference for a structurally defined, low-dimensional parameter $\theta$ that is explicitly encoded in the network architecture through sign-alignment or the MISC conditions.

~

The rest of the paper is organized as follows. Section~\ref{sec:BinChoice}
introduces the RMS estimator in the binary choice model, develops the basic
identification and asymptotic theory, and compares RMS to the original and
smoothed maximum score estimators. Section~\ref{sec:MISC} embeds the binary
choice setup into the general multi-index single-crossing framework, extends
the RMS criterion to the $J$-index case, and derives the corresponding
asymptotic results, highlighting the effective one-dimensional nature of the
rate. Section \ref{sec:NN} further reformulates the RMS as a specialized layer in a DNN, which allows the estimation of the index parameter to be subsumed under the training of the DNN, for which state-of-art computing software on DNN become applicable.  Section~\ref{sec:Sim} presents simulation evidence on the finite-sample
performance of the RMS estimator in both single-index and multi-index designs.
Section~\ref{sec:Con} concludes. Technical proofs and additional auxiliary
results are collected in the appendix.

\section{\label{sec:BinChoice}Special Case: Binary Choice Model}

In this section, we focus on the binary choice model \eqref{eq:BinChoice}
as described in the introduction, and develops the econometric theory
of our ReLU-based maximum score (RMS) estimator with clear lower-level
conditions on the primitives of the model. The binary choice model
is not only of important interest on its own, but also serves as a
canonical setup where our new ReLU-based estimator can be related
to the original maximum score (MS) estimator and its previous variants
in a clear manner.

\subsection{\label{subsec:SetupResults}Setup and Main Results}

Given the binary choice model \eqref{eq:BinChoice} and the ReLU-based
population criterion function $Q$ in \eqref{eq:def_Q}, we define
the ReLU-based maximum score (RMS) estimator as 
\begin{equation}
\hat{\t}:=\arg\max_{\t\in\S^{d-1}}\hat{Q}\left(\t\right)\label{eq:RMS}
\end{equation}
where the sample criterion function $\hat{Q}$ is given by
\[
\hat{Q}\left(\t\right):=\frac{1}{n}\sum_{i=1}^{n}\left(g_{+,\t,\hat{h}}\left(X_{i}\right)+g_{-,\t,\hat{h}}\left(X_{i}\right)\right)
\]
with $\hat{h}$ being some first-stage nonparametric estimator of
$h_{0}\left(x\right)=\E\left[\rest{y_{i}-\frac{1}{2}}X_{i}=x\right]$.
We seek to characterize the asymptotic behaviors of the RMS estimator
$\hat{\t}$.

It turns out that the $\hat{\t}$ is ``non-standard'' semiparametric
two-stage M-estimator, and is different from both the usual ``$\sqrt{n}$-normal''
asymptotics in the ``smooth case'' (such as in \citealp{newey1994asymptotic})
and the ``cubic-rate'' asymptotics in \citet*{kim1990cube}. 

As we will show subsequently, the ReLU-based maximum score estimator
will feature ``intermediate'' asymptotics (under appropriate conditions
to be made explicit later): $\hat{\t}$ will converge at nonparametric
rates slower than $n^{\frac{1}{2}}$ but faster than $n^{\frac{1}{3}}$
with asymptotic normal distribution, which can be viewed as a ``semiparametric
two-stage version'' of the asymptotic results in \citealp*{horowitz1992smoothed}. 

In particular, the ``intermediate asymptotics'' of ReLU-based maximum
score estimator $\hat{\t}$ is critically driven by the ``intermediate
smoothness'' allowed by the formulation of the criterion function
\eqref{eq:def_Q} using the the ReLU function $\left[\cd\right]_{+}$,
which is Lipschitz continuous and everywhere differentiable \emph{except}
at the single ``kink point'' $0$. Interestingly, both the ``smoothness''
and ``kinkiness'' of the ReLU function turns out to be important:
while the Lipschitz continuity of the ReLU function is key in delivering
a ``stochastic equicontinuity'' condition for asymptotic normality,
and the ``kinkiness'' of the ReLU function at $0$ is key in delivering
locally quadratic identification of $\t_{0}$, i.e., the quadratic
curvature of the population criterion function $Q$ in a neighborhood
of $\t_{0}$.

~

We start by imposing a set of lower-level assumptions that guarantees
the point identification of $\t_{0}$ (under scale normalization)
by the RMS criterion function \eqref{eq:def_Q} and that a variety
of densities are smooth and well-behaved. We note that these assumptions
are stronger than necessary, but lend simplicity to the exposition
of our main results. 
\begin{assumption}
\label{assu:Basic} Write ${\cal X}:=\text{Supp}\left(X_{i}\right)\subseteq\R^{d}$.
Suppose $\t_{0}\in\S^{d-1}$ and the following:
\begin{itemize}
\item[(a)]  $\left(y_{i},X_{i},\e_{i}\right)_{i=1}^{n}$ is i.i.d. and satisfies
model \eqref{eq:BinChoice}.
\item[(b)]  The conditional median of $\e_{i}$ given $X_{i}=x$ is zero, i.e.,
\[
F\left(\rest 0x\right)=\frac{1}{2},\quad\forall x\in{\cal X}.
\]
\item[(c)]  The (unknown) conditional CDF $F\left(\rest{\e}x\right)$ of $\e_{i}$
given $X_{i}=x$ is $d$ times continuously differentiable w.r.t.
$\left(\e,x\right)\in\R\times{\cal X}$ with uniformly bounded derivatives
(bounded by some positive constant $M<\infty$).
\item[(d)]  The conditional probability density function $f\left(\rest{\e}x\right)$
of $\e_{i}$ given $X_{i}=x$ is strictly positive for any $\e\in\R$
and $x\in{\cal X}$. 
\item[(e)]  Furthermore, there exists a finite $M>0$ such that 
\[
0<\frac{1}{M}\leq f\left(\rest 0x\right)\leq M,\quad\text{for all }x\in{\cal X}.
\]
\item[(f)]  ${\cal X}$ is compact in $\R^{d}$ and contains ${\bf 0}$ as an
interior point. WLOG assume $\norm{x}\leq 1, \forall x\in  {\cal X}$.
\item[(g)]  Let $p\left(x\right)$ be the probability density function of $X_{i}$.
There exists a finite $M>0$ such that 
\[
0<\frac{1}{M}\leq p\left(x\right)\leq M,\quad\text{for all }x\in{\cal X}.
\]
\end{itemize}
\end{assumption}
Assumption \ref{assu:Basic}(a) and (b) consists of a standard random-sampling
assumption for the binary choice model \eqref{eq:BinChoice} with
a conditional median restriction, which are essentially the same as
those imposed in \citet*{horowitz1992smoothed}. Note, however, we
focus on the binary choice model here as a key illustration, but,
just as maximum score estimator can be applied to many models other
than the binary choice model \eqref{eq:BinChoice}, our proposed method
can also be adapted to other settings. See XXX for a more detailed
discussion.

Assumption \ref{assu:Basic}(c)-(e) are regularity conditions on the
conditional distribution of the error term $\e_{i}$, which correspond
to Assumptions 2(b), 9 and 11 in \citet*{horowitz1992smoothed}. The
assumptions of the existence and boundedness (from above) of conditional
densities and their derivatives impose smoothness conditions on model
\eqref{eq:BinChoice} and the conditional expectation function $h_{0}\left(x\right)$
beyond \citet{manski1985semiparametric} and \citet{kim1990cube}.
As in \citet*{horowitz1992smoothed}, these smoothness conditions
are exploited to deliver faster convergence rates than the cubic rate
as well as asymptotic normality. Note that the ``bounded away from
zero'' assumption $f\left(\rest 0x\right)>\frac{1}{M}$ is a local-identification
assumption that deliver the quadratic curvature in the population
criterion function, which is imposed implicitly in Assumption 11 of
\citet*{horowitz1992smoothed}.

Assumption \ref{assu:Basic}(f) imposes assumption on ${\cal X}$,
the support of the covariates $X_{i}$. In particular, the assumption
of ${\cal X}$ containing ${\bf 0}$ as an interior point guarantees
that $X_{i}$ has full ``directional'' support, i.e. $X_{i}/\norm{X_{i}}$
is supported on the whole $\S^{d-1}$. As explained in \citet*{manski1985semiparametric},
the identification of $\t_{0}$ is driven by variations in the ``directions''
$X_{i}/\norm{X_{i}}$, and the full-directional-support condition
ensures that $\t_{0}$ is point identified on $\S^{d-1}$. As well-known
in the literature, the assumption of ${\bf 0}$ being in the interior
of ${\cal X}$ is a sufficient, but not necessary, condition for the point identification
of $\t_{0}$. Alternatively, one could work with a ``special regressor''
as in Assumptions 2(a)(c) \& 4 in \citet*{horowitz1992smoothed},
which assume that $\left|\b_{01}\right|=1$ and that the conditional
distribution of $X_{i1}$ given $\left(X_{i2},...,X_{id}\right)$
has full support on $\R$. This alternative set of assumptions allows
for discreteness in certain components of $X_{i}$ but rules out compactness
of ${\cal X}$, and thus do not nest Assumptions \ref{assu:Basic}(f)(g)
as special cases, nor vice versa. Furthermore, the scale normalization
$\left|\b_{01}\right|=1$ is dependent on the assumption that a specific
known component of $X_{i}$ has non-zero coefficient. In this paper,
we focus on the normalization $\b_{0}\in\R^{d}$, i.e., $\norm{\b_{0}}=1$
and the support Assumptions \ref{assu:Basic}(f)(g), which lends simpler
notation in our asymptotics. However, the substance of our asymptotic
results is not dependent on this specific choice of point-identifying
assumption and scale normalization, and it should be feasible, though
notationally cumbersome, to adapt our asymptotic results to the set
of assumption and normalization using the ``special regressor''
as in \citet*{horowitz1992smoothed}.

We also, note that the assumption of compactness of ${\cal X}$ in
Assumption \ref{assu:Basic}(f) is not necessary, either. Compactness
of ${\cal X}$ is often assumed in the literature, and assumed here
for simpler exposition of results on the nonparametric estimation
of $h_{0}\left(x\right)$. Hence, our results based on the consistency
and convergence rate of nonparametric estimation of $h_{0}$ on compact
${\cal X}$ can be adapted to the case where ${\cal X}$ is not compact
with standard trimming and/or weighting of ${\cal X}$. 

Lastly, Assumption \ref{assu:Basic}(g) corresponds to Assumptions
8 and 11 in \citet*{horowitz1992smoothed}, imposing both smoothness
(in terms of bounded-from-above densities) and local-identification
conditions (in terms of bounded-away-from zero densities). Again,
Assumption \ref{assu:Basic}(g) is stated in a stronger-than-necessary
but expositionally simple form. In particular, for local identifiability
it is not necessary to require that $p\left(x\right)$ is bounded
away from zero at every point in ${\cal X}$, since local identifiability
is only concerned with the hyperplane $\left\{ x:x^{'}\t_{0}=0\right\} $,
not the entire ${\cal X}$. However, given imposed compactness of
${\cal X}$ in Assumption \ref{assu:Basic}(f), the global ``bounded-away-from-zero''
condition here is not very restrictive anyway, and hence we impose
this stronger-than-necessary condition for simpler notation.

We summarize two important implications of Assumption \ref{assu:Basic}
below:
\begin{prop}
\label{prop:ID_Sobolev}Under Assumption \ref{assu:Basic}:
\begin{itemize}
\item[(i)]  $\t_{0}$ is point identified on $\S^{d-1}$: 
\begin{equation}
\t_{0}=\arg\max_{\t\in\S^{d-1}}Q\left(\t\right).\label{eq:PointID}
\end{equation}
\item[(ii)]  $h_{0}\left(x\right)$ is $\left(d+1\right)$ times differentiable
on ${\cal X}$ with uniformly bounded derivatives.
\end{itemize}
\end{prop}
~

Given mild convergence conditions on the first-stage estimator $\hat{h}$,
it is straightforward to establish the consistency of $\hat{\t}$
in Theorem \ref{thm:Consistency}.
\begin{thm}[Consistency]
\label{thm:Consistency}Suppose that $\norm{\hat{h}-h_{0}}_{\infty}=o_{p}\left(1\right)$.
Then $\hat{\t}$ is consistent, i.e., $\hat{\t}\pto\t_{0}$.
\end{thm}
~

We now proceed to characterize the convergence rate and asymptotic
distribution of $\hat{\t}$, which are the main results of this section.
While such results can be obtained under higher-level conditions on
the first-stage nonparametric estimators $\hat{h}$, for concreteness
and clarity, we consider two leading types of nonparametric estimators,
the Nadaraya-Waston kernel estimator and the linear series estimator,
and provide lower-level conditions for both.
\begin{assumption}[Kernel/Linear Series First Stage]
\label{assu:KernelSeries} Assume either of the following:
\begin{itemize}
\item[\emph{(a)}] \begin{flushleft}
 $\hat{h}$ is given by the Nadaraya-Watson kernel estimator, 
\[
\hat{h}\left(x\right):=\frac{\sum_{i=1}^{n}K\left(\frac{X_{i}-x}{b_{n}}\right)\left(Y_{i}-\frac{1}{2}\right)}{\sum_{i=1}^{n}K\left(\frac{X_{i}-x}{b_{n}}\right)}
\]
where $b_{n}$ is a bandwidth parameter and $K\left(x\right)$ is
a $d$-dimensional kernel function of smoothness order $s$ such that
\par\end{flushleft}
\item[] (a.i) $K\left(x\right)=K\left(-x\right)$, and $\int K\left(x\right)dx=1$.
\item[] (a.ii) $\left|K\left(x\right)\right|\leq M<\infty$ and $\int\prod_{j=1}^{d}\left|x_{j}\right|^{l}K\left(x\right)dx<\infty$
for all $l$.
\item[] (a.iii) $\int k\left(t\right)=0$ for $j=1,...,s-1$, and $\kappa_{s}:=\int x_{j}^{s}K\left(x\right)dx\in\left(0,\infty\right)$.
\item[] (a.iv) $K\left(x_{1},...,x_{d}\right)=K\left(x_{\pi_{1}},...,x_{\pi_{d}}\right)$
for any permutation of coordinates $\pi$.
\item[\emph{(b)}] \begin{flushleft}
\emph{ $\hat{h}$ is given by the linear series estimator,}
\[
\hat{h}\left(x\right):=\ol b^{K_{n}}\left(x\right)^{'}\left(\sum_{i=1}^{n}\ol b^{K_{n}}\left(X_{i}\right)\ol b^{K_{n}}\left(X_{i}\right)^{'}\right)^{-1}\sum_{i=1}^{n}\ol b^{K_{n}}\left(X_{i}\right)\left(Y_{i}-\frac{1}{2}\right)
\]
where $K_{n}:=J_{n}^{d}$ is the sieve dimension parameter and $\ol b^{K_{n}}\left(x\right):=\text{vec}\left(\otimes_{j=1}^{d}\left(b_{1}\left(x_{j}\right),...,b_{J_{n}}\left(x_{j}\right)\right)\right)$
is a vector of multivariate basis functions constructed from tensor
products of some univariate orthonormal basis functions $\left(b_{k}\left(\cd\right)\right)_{k=1}^{\infty}$
such that:
\par\end{flushleft}
\item[] (b.i)\emph{ $\l_{\min}\left(\E\left[\ol b^{K_{n}}\left(X_{i}\right)\ol b^{K_{n}}\left(X_{i}\right)^{'}\right]\right)>0$.}
\item[] (b.ii)\emph{ $\inf_{h\in{\cal B}_{K_{n}}}\norm{h-h_{0}}=J_{n}^{-s}$,
where }${\cal B}_{K_{n}}$ denotes the closed span of $\ol b^{K_{n}}$.
\item[] (b.iii)\emph{ $\norm{\Pi_{K_{n},n}}_{\infty}:=\sup_{h:\norm h_{\infty}\neq0}\frac{\norm{\Pi_{K_{n},n}h}_{\infty}}{\norm h_{\infty}}=O_{p}\left(1\right)$,
where} $\Pi_{K_{n},n}$ denotes the empirical projection operator
onto ${\cal B}_{K_{n}}$, i.e., 
\[
\Pi_{K_{n},n}h\left(x\right):=\ol b^{K_{n}}\left(x\right)^{'}\left(\sum_{i=1}^{n}\ol b^{K_{n}}\left(X_{i}\right)\ol b^{K_{n}}\left(X_{i}\right)^{'}\right)^{-1}\sum_{i=1}^{n}\ol b^{K_{n}}\left(X_{i}\right)h\left(x\right).
\]
\end{itemize}
\end{assumption}
\noindent \begin{flushleft}
Assumption \ref{assu:KernelSeries}(a.i-iv) are standard conditions
on the kernel function that covers both product and radial kernels
constructed under a wide range of univariate kernels. Similarly, Assumption
\ref{assu:KernelSeries} (b.i-iii) are standard conditions and properties
on linear series regressions that are satisfied under a wide variety
of sieve classes. See, for example, \citet{chen2007sieve}, \citet{chen2015optimal}
and \citet{belloni2015some} for results on spline, wavelet, Fourier
and many other sieve classes. 
\par\end{flushleft}

We now present our main results about the RMS asymptotics.
\begin{thm}[Convergence Rate]
\label{thm:Thm_Bin_Rate} Under Assumption \eqref{assu:Basic}, and
with $\hat{h}$ being given by the Nadaraya-Watson estimator that
satisfies Assumption \emph{\ref{assu:KernelSeries}(a),} for any $b_{n}\to0$
such that $nb_{n}^{2d+1}/\left(\log n\right)^{2}\to\infty$, we have
\begin{equation}
\norm{\hat{\t}-\t_{0}}=O_{p}\left(b_{n}^{s}+\frac{1}{\sqrt{nb_{n}}}\right).\label{eq:Bin_Rate_b}
\end{equation}
If $s>d$, then optimal convergence rate can be attained by setting
$b_{n}\sim n^{-\frac{1}{2s+1}}$, giving
\[
\norm{\hat{\t}-\t_{0}}=O_{p}\left(n^{-\frac{s}{2s+1}}\right).
\]
The above also holds for linear series $\hat{h}$ with Assumption
\emph{\ref{assu:KernelSeries}(a)} replaced by Assumption \emph{\ref{assu:KernelSeries}(b)}
and $b_{n}$ replaced by $J_{n}^{-1}$.
\end{thm}
The asymptotic distribution can then be derived based on the linearized
argmax theorem (Theorem 3.2.16) in \citet*{van1996weak}. 
\begin{thm}[Asymptotic Normality]
\label{thm:AsymNorm} Suppose that Assumption holds with $s>d$.
With $\hat{h}$ being given by the Nadaraya-Watson estimator as in
Assumption \emph{\ref{assu:KernelSeries}(a) with undersmoothing choice
of bandwidth} $b_{n}$ such that $nb_{n}^{2d+1}/\left(\log n\right)^{2}\to\infty$
and $b_{n}=o_{p}\left(n^{-\frac{1}{2s+1}}\right)$, we have 
\[
n^{\frac{s}{2s+1}}\left(\hat{\t}-\t_{0}\right)\dto\cN\left({\bf 0},V^{-}\O V^{-}\right).
\]
The above also holds for linear series $\hat{h}$ with Assumption
\emph{\ref{assu:KernelSeries}(a) }replaced by Assumption \emph{\ref{assu:KernelSeries}(b)}
and $b_{n}$ replaced by $J_{n}^{-1}$.
\end{thm}

% The asymptotic normality of $\hat{\t}$ established above immediately
% implies, by \citet*[Theorem 3.1]{fang2019inference}, the consistency
% of bootstrap inference, which facilitates 
% \begin{cor}[Bootstrap Consistency]
% Under the assumptions in Theorem \ref{thm:AsymNorm}, the standard
% nonparametric bootstrap procedure yields consistent inference for
% $\t_{0}$. Furthermore, bootstrap inference for $\phi\left(\t_{0}\right)$
% is consistent for function $\phi$ that is differentiable at $\t_{0}$.
% \end{cor}

\subsection{\label{subsec:AsymOutline}Outline of the RMS Asymptotic Theory}

\subsubsection{Decomposition of the Sample Criterion}

To present our formal asymptotic results, we first set up some notation.
Let $Pg_{\t,h}:=\int g_{\t,h}\left(x\right)dP\left(x\right),$ $\P_{n}g_{\t,h}:=\frac{1}{n}\sum_{i=1}^{n}g_{\t,h}\left(X_{i}\right),$
and $\GG_{n}g_{\t,h}:=\sqrt{n}\left(\P_{n}g_{\t,h}-Pg_{\t,h}\right)$,
with which we can rewrite \eqref{eq:PointID} and \eqref{eq:RMS}
as
\[
\t_{0}=\arg\max_{\t\in\S^{d-1}}Pg_{\t,h_{0}},\quad\hat{\t}:=\arg\max_{\t\in\S^{d-1}}\P_{n}g_{\t,\hat{h}}
\]
Since the asymptotic behavior of $\hat{\t}$ is driven by the asymptotic
behavior of $\P_{n}\left(g_{\hat{\t},\hat{h}}-g_{\t_{0},\hat{h}}\right)$,
we analyze it by working with the following decomposition
\begin{align}
\P_{n}\left(g_{\hat{\t},\hat{h}}-g_{\t_{0},\hat{h}}\right)= & \underset{T_{1}}{\underbrace{\frac{1}{\sqrt{n}}\GG_{n}\left(g_{\hat{\t},h_{0}}-g_{\t_{0},h_{0}}\right)}}+\underset{T_{2}}{\underbrace{\frac{1}{\sqrt{n}}\GG_{n}\left(g_{\hat{\t},\hat{h}}-g_{\t_{0},\hat{h}}-g_{\hat{\t},h_{0}}+g_{\t_{0},h_{0}}\right)}}\nonumber \\
 & +\underset{T_{3}}{\underbrace{P\left(g_{\hat{\t},h_{0}}-g_{\t_{0},h_{0}}\right)}}+\underset{T_{4}}{\underbrace{P\left(g_{\hat{\t},\hat{h}}-g_{\t_{0},\hat{h}}-g_{\hat{\t},h_{0}}+g_{\t_{0},h_{0}}\right)}}\label{eq:EP_decom}
\end{align}
and studying each of the four terms $T_{1},T_{2},T_{3}$ and $T_{4}$. 

It turns out that each of the four terms is somewhat ``nonstandard''
relative to the usual case of semiparametric two-stage estimation
theory that delivers $\sqrt{n}$ asymptotic normality under standard
smoothness conditions. Furthermore, the analysis of the four terms
$T_{1},T_{2},T_{3},T_{4}$ reveals some of the key insights in the
asymptotics of our proposed ReLU-based maximum score estimator $\hat{\t}$. 

Hence, we provide an explicit account of the four terms below, where
we show that the terms $T_{1}$ and $T_{2}$ will be of smaller stochastic
orders than $\norm{\hat{\t}-\t_{0}}^{2}$ and thus become asymptotically
negligible, while terms $T_{3}$ and $T_{4}$ will be the asymptotically
leading terms of the order $\norm{\hat{\t}-\t_{0}}^{2}$. We then
combine the results about the four terms to establish the convergence
rate and asymptotic normality.

\subsubsection{\label{subsec:Term1}Analysis of Term $T_{1}=\frac{1}{\sqrt{n}}\protect\GG_{n}\left(g_{\hat{\protect\t},h_{0}}-g_{\protect\t_{0},h_{0}}\right)$}

We start with term $T_{1}$, which captures the stochastic variation,
or loosely ``variance'', in the sample criterion function $\P_{n}g_{\hat{\t},h_{0}}$
when the nonparametric first stage is set to the the true function
$h_{0}$. Lemma \ref{lem:Term1} below presents a maximal inequality
about $T_{1}$ with respect to $\t$ in a small neighborhood of $\t_{0}$:
\begin{lem}
\label{lem:Term1} For some constant $M>0$,
\begin{equation}
P\sup_{\norm{\t-\t_{0}}\leq\d}\left|\GG_{n}\left(g_{\t,h_{0}}-g_{\t_{0},h_{0}}\right)\right|\leq M\d^{\frac{3}{2}}.\label{eq:Max1}
\end{equation}
\end{lem}
Loosely speaking, the result above in Lemma \ref{lem:Term1} translates
to the following stochastic bounds on $T_{1}$:
\[
T_{1}=O_{p}\left(\frac{1}{\sqrt{n}}\norm{\hat{\t}-\t_{0}}^{\frac{3}{2}}\right),
\]
which is $o_{p}\left(\norm{\hat{\t}-\t_{0}}^{2}\right)$ since $\norm{\hat{\t}-\t_{0}}$
converges no faster than $\frac{1}{\sqrt{n}}$ rate to zero. This
would imply that $T_{1}$ will become asymptotically negligible, which
is ``nonstandard'' in the literature.

Technically, the asymptotic negligibility of $T_{1}$ is directly
driven by the $\d^{\frac{3}{2}}$-rate bound on the right hand side
of \eqref{eq:Max1}. To see why $\d^{\frac{3}{2}}$ arises, notice
that 
\[
\left|g_{\t,h_{0}}\left(x\right)-g_{\t_{0},h_{0}}\left(x\right)\right|=\left|g_{+,\t,h_{0}}\left(x\right)-g_{+,\t,h_{0}}\left(x\right)\right|+\left|g_{-,\t,h_{0}}\left(x\right)-g_{-,\t,h_{0}}\left(x\right)\right|
\]
and thus, for any $\t$ close to $\t_{0}$ in the sense of $\norm{\t-\t_{0}}\leq\d$,
we have
\begin{align}
\left|g_{+,\t,h_{0}}\left(x\right)-g_{+,\t_{0},h_{0}}\left(x\right)\right| & =\left|\left[h_{0}\left(x\right)-\left[-x^{'}\t\right]_{+}\right]_{+}-\left[h_{0}\left(x\right)\right]_{+}\right|\nonumber \\
 & \leq\ind\left\{ h_{0}\left(x\right)>0\right\} \cd\ind\left\{ x^{'}\t<0\right\} \cd\left|x^{'}\t\right|\nonumber \\
 & =\ind\left\{ x^{'}\t_{0}>0>x^{'}\t\right\} \cd\left|x^{'}\t\right|\nonumber \\
 & =\ind\left\{ x^{'}\t_{0}>0>x^{'}\t_{0}+x^{'}\left(\t-\t_{0}\right)\right\} \cd\left|x^{'}\t_{0}+x^{'}\left(\t-\t_{0}\right)\right|\nonumber \\
 & \leq\ind\left\{ 0<x^{'}\t_{0}<-x^{'}\left(\t-\t_{0}\right)\right\} \cd\left(\left|x^{'}\t_{0}\right|+\left|x^{'}\left(\t-\t_{0}\right)\right|\right)\nonumber \\
 & \leq\ind\left\{ 0<x^{'}\t_{0}<M\norm x\d\right\} \cd2M\norm x\d\label{eq:Bd_gtheta}
\end{align}
In words, the derivation above exploits the observation that $g_{+,\t,h_{0}}\left(x\right)-g_{+,\t_{0},h_{0}}\left(x\right)$
is nonzero only if $x^{'}\t_{0}$ and $x^{'}\t$ lie on different
sides of $0$, which, given the restriction $\norm{\t-\t_{0}}\leq\d$,
implies that both $\left|x^{'}\left(\t-\t_{0}\right)\right|$ and
$x^{'}\t_{0}$ must be bounded by $M\norm x\d$. As a result, the
magnitude of $\left|g_{+,\t,h_{0}}\left(x\right)-g_{+,\t_{0},h_{0}}\left(x\right)\right|$,
which is at most $\left|x^{'}\t\right|$, is also bounded above by
a term linear in $\d$. Furthermore, since $\norm x$ is bounded by
the compactness of ${\cal X}$,\footnote{The compactness of ${\cal X}$ and the boundedness of $\norm x$ allow
for simpler exposition here but are not necessary. If $\norm{X_{i}}$
has unbounded support, the result in Lemma \ref{lem:Term1} will continue
to hold under mild tail-decay condition, or finite--fourth-moment
condition, on $\norm{X_{i}}$.} we have
\[
\left|g_{+,\t,h_{0}}\left(x\right)-g_{+,\t_{0},h_{0}}\left(x\right)\right|\leq\ol g_{\d}\left(x\right):=\ind\left\{ \left|x^{'}\t_{0}\right|<M\norm x\d\right\} \cd2M\d,
\]
and similarly for $\left|g_{-,\t,h_{0}}\left(x\right)-g_{-,\t,h_{0}}\left(x\right)\right|$.
Hence, $\ol g_{\d}\left(x\right)$ is a so-called ``envelope function''
for the function class $\left\{ g_{\t,h_{0}}\left(x\right)-g_{\t_{0},h_{0}}\left(x\right):\norm{\t-\t_{0}}\leq\d\right\} $
in the sense of 
\[
\sup_{\t:\norm{\t-\t_{0}}\leq\d}\left|g_{\t,h_{0}}\left(x\right)-g_{\t_{0},h_{0}}\left(x\right)\right|\leq\ol g_{\d}\left(x\right),\quad\forall x\in{\cal X}.
\]
By standard empirical process theory, such as in \citet*{van1996weak},
the magnitude of $\sqrt{\E\left[\ol g_{\d}\left(X_{i}\right)^{2}\right]}$
is key for the maximal inequality in the style of \eqref{eq:Max1},
which in the current setting is given by 
\[
\sqrt{\E\left[\ol g_{\d}\left(X_{i}\right)^{2}\right]}=\sqrt{\P\left(\left|\frac{X_{i}^{'}}{\norm{X_{i}}}\t_{0}\right|<M\d\right)\cd M\d^{2}}\ =\sqrt{O\left(\d\right)\cd M\d^{2}}=O\left(\d^{\frac{3}{2}}\right),
\]
where $\P\left(\frac{X_{i}^{'}}{\norm{X_{i}}}\t_0\leq M\d\right)=O\left(\d\right)$
follows from the observation that $\P\left(\frac{X_{i}^{'}}{\norm{X_{i}}}\leq M\d\right)$
is the probability of random angle between $\frac{X_{i}^{'}}{\norm{X_{i}}}$
and $\t_{0}$ being no more than $M\d$ away from $\pi/2$, which scales linearly with
$\d$ under the assumption that $p\left(x\right)$ is bounded from
above and away from zero for all $x\in{\cal X}$ in Assumption \ref{assu:Basic}(g). 

In summary, $\ol g_{\d}\left(x\right)^{2}$is at most $M\d^{2}$ and
nonzero in a region of probability measure at most $M\d$, and hence
$\E\left[\ol g_{\d}\left(X_{i}\right)^{2}\right]$ is bounded by $M\d^{3}$.
Importantly, $\left|x^{'}\t\right|$ interacts multiplicatively with
the indicator function $\ind\left\{ 0<x^{'}\t_{0}<M\norm x\d\right\} $
in \eqref{eq:Bd_gtheta}, and hence, even though indicators functions
are invariant under squaring $\ind\left\{ \cd\right\} ^{2}\equiv\ind\left\{ \cd\right\} $,
the magnitude of $\left|x^{'}\t\right|^{2}\leq M\d^{2}$ becomes smaller
in the order of magnitude after squaring, leading to the overall $M\d^{3}$
on $\E\left[\ol g_{\d}\left(X_{i}\right)^{2}\right]$.

To contrast this with the case of cubic-root asymptotics, say, in
\citet*{kim1990cube}, write the original maximum-score estimand $g_{MS,\t}\left(y,x\right):=\left(y-\frac{1}{2}\right)\ind\left\{ x^{'}\t>0\right\} $,
and observe that
\begin{align*}
\left|g_{MS,\t}\left(y,x\right)-g_{MS,\t_{0}}\left(y,x\right)\right| & =\frac{1}{2}\cd\left(\ind\left\{ x^{'}\t>0\geq x^{'}\t_{0}\right\} +\ind\left\{ x^{'}\t_{0}>0\geq x^{'}\t\right\} \right)\\
 & \leq\frac{1}{2}\cd\left\{ 0<x^{'}\t_{0}<M\norm x\d\right\} :=\ol g_{MS,\d}\left(x\right)
\end{align*}
where the envelope function $\ol g_{MS,\d}\left(x\right)$ remains
as a discrete function with 
\begin{align*}
\sqrt{\E\left[\ol g_{MS,\d}\left(x\right)^{2}\right]} & =\sqrt{\frac{1}{4}\P\left(\left|\frac{X_{i}^{'}}{\norm{X_{i}}}\t_{0}\right|<M\d\right)}\leq M\d^{\frac{1}{2}},
\end{align*}
leading to a much larger bound than $\d^{\frac{3}{2}}$ (with $\d$
thought to be close to $0$). As discussed in \citet*{kim1990cube},
the $\d^{\frac{1}{2}}$ bound above is the key driver for the cubic-root
asymptotics, and it arises both from the discreteness of the indicator
function $\ind\left\{ x^{'}\t>0\right\} $ as well as the discreteness
of the binary outcome $y_{i}-\frac{1}{2}$. In contrast, in our current
setting, the discrete outcome $y_{i}-\frac{1}{2}$ is replaced by
its conditional expectation, $h_{0}\left(x\right)=\E\left[\rest{y_{i}-\frac{1}{2}}X_{i}=x\right]$,
which is a smooth object, and furthermore the estimand $g_{+,\t,h_{0}}\left(x\right)-g_{+,\t_{0},h_{0}}\left(x\right)$
is constructed to be Lipschitz continuous in $x^{'}\t$.

\subsubsection{Analysis of Term $T_{2}=\frac{1}{\sqrt{n}}\protect\GG_{n}\left(g_{\hat{\protect\t},\hat{h}}-g_{\protect\t_{0},\hat{h}}-g_{\hat{\protect\t},h_{0}}+g_{\protect\t_{0},h_{0}}\right)$}

\noindent We now turn to the second term $T_{2}$, which involves
the first-stage nonparametric estimator $\hat{h}$ of $h$. The asymptotic
negligibility of term $T_{2}$ corresponds to the usual ``stochastic
equicontinuity'' condition, which we will seek to establish here. 

To do so, we impose the following standard sup-norm convergence of
the first-stage estimator $\hat{h}$ . First, notice that given Proposition
\ref{prop:ID_Sobolev}(b), $h_{0}\in{\cal H}$ with ${\cal H}$ denoting
the space of functions mapping from ${\cal X}$ to $\left[-\frac{1}{2},\frac{1}{2}\right]$
that possess uniformly bounded derivatives up to order $d+1$. See,
for example, \citet{hansen2008uniform}, \citet{belloni2015some}
and \citet*{chen2015optimal} for results on the sup-norm convergence
of kernel and sieve nonparametric estimators.
\begin{assumption}
\label{assu:FirstStageRate}(i) $\hat{h}\in{\cal H}$ with probability
approaching 1, and (ii) $\norm{\hat{h}-h_{0}}_{\infty}=O_{p}\left(a_{n}\right)$.
\end{assumption}
\begin{lem}
\label{lem:Term2} Under Assumptions \ref{assu:Basic}-\ref{assu:FirstStageRate},
for some constant $M>0$, 
\begin{equation}
P\sup_{\t\in\T,h\in{\cal H}:\norm{\t-\t_{0}}\leq\d,\norm{h-h_{0}}_{\infty}\leq Ka_{n}}\left|\GG_{n}\left(g_{\t,h}-g_{\t_{0},h}-g_{\t,h_{0}}+g_{\t_{0},h_{0}}\right)\right|\leq M\d.\label{eq:Max2}
\end{equation}
\end{lem}
Loosely speaking, Lemma \eqref{lem:Term2} implies that, whenever
$\norm{\hat{\t}-\t_{0}}$ converges slower than the $\sqrt{n}$ rate,
\[
T_{2}=O_{p}\left(\frac{1}{\sqrt{n}}\norm{\hat{\t}-\t_{0}}\right)=o_{p}\left(\norm{\hat{\t}-\t_{0}}^{2}\right),
\]
which will become asymptotically negligible, delivering a ``stochastic
equicontinuity'' condition that is essential for the asymptotic normality
of $\hat{\t}$. The key model ingredient underlying this result is
the encoding of the sign restrictions via compositions of the Lipschitz-continuous
ReLU-function instead of using the discrete indicator functions as
in the formulation of the original maximum score estimator. The Lipschitz
continuity of ReLU functions, and consequently the Lipschitz continuity
of the function $g_{\t,h}\left(x\right)=g_{+,\t,h}\left(x\right)+g_{-,\t,h}\left(x\right)$,
ensure that small deviations in $\t,h$ and $x$ translate into small
deviations in $g_{\t,h}\left(x\right)$, providing the level of smoothness
for the stochastic equicontintuity condition. 

\subsubsection{Analysis of Term $T_{3}=P\left(g_{\hat{\protect\t},h_{0}}-g_{\protect\t_{0},h_{0}}\right)$}

\noindent Now, we turn to the third term $T_{3}=P\left(g_{\hat{\t},h_{0}}-g_{\t_{0},h_{0}}\right)$,
which is a familiar term that captures the quadratic curvature of
the population criterion for $\t$ in a neighborhood of $\t_{0}$.
Technically, the characterization of $T_{3}$ boils down to the following
second-order Taylor expansion of $Pg_{\t,h_{0}}$ around $\t_{0}$:
\[
P\left(g_{\t,h_{0}}-g_{\t_{0},h_{0}}\right)=\Dif_{\t}Pg_{\t_{0},h_{0}}\left(\t-\t_{0}\right)+\frac{1}{2}\left(\t-\t_{0}\right)^{'}\Dif_{\t\t}Pg_{\t_{0},h_{0}}\left(\t-\t_{0}\right)+o\left(\norm{\t-\t_{0}}^{2}\right)
\]
where the gradient $\Dif_{\t}Pg_{\t_{0},h_{0}}$ and the Hessian $\Dif_{\t\t}Pg_{\t_{0},h_{0}}$
are well-defined since $Pg_{\t,h}$ is differentiable even though
$g_{\t,h}$ has kinks. Moreover, since $g_{\t,h}$ is Lipschitz-continuous
and almost surely differentiable, the gradient can be calculated easily
via $\Dif_{\t}Pg_{\t,h}=P\Dif_{\t}g_{\t,h}$ However, $\Dif_{\t}g_{\t,h}$
will no longer be Lipschitz-continuous and in fact involve indicator
functions, and thus the Hessian $\Dif_{\t\t}Pg_{\t,h}=\Dif_{\t}P\Dif_{\t}g_{\t,h}$
involves differentiation with respect to integral boundaries. As a
result, $\Dif_{\t\t}Pg_{\t,h}$ becomes a ``surface integral'',
or formally, an integral over a lower-dimensional manifolds with respect
to a lower-dimensional Hausdorff measure. 

Specifically, the $k$-dimensional Hausdorff measure in $\R^{d}$,
denoted by ${\cal H}^{k}$ for some $k\leq d$, is a ``lower-dimensional''
measure that allows us to define nontrivial integrals over lower-dimensional
subsets in $\R^{d}$ that has measure $0$ with respect to ${\cal L}^{d},$
the Lebesgue measure on $\R^{d}$. See, for example, Chapter 2 of
\citet*{evans2015measure} for the formal definition of the Hausdorff
measure. An important feature of the Hausdorff measure is the equivalence
between ${\cal H}^{k}$ and ${\cal L}^{k}$ on $\R^{k}$ for any $k$,
i.e., the $k$-dimensional Hausdorff measure is in some sense the
same as the Lebesgue measure on $\R^{k}$. On the other hand, while
a lower-dimensional space, such as a hyperplane $\left\{ x\in\R^{d}:x^{'}\t_{0}=0\right\} $
in $\R^{d}$, is a measure-0 set with respect to ${\cal L}^{d}$ and
thus the integral $\int_{\left\{ x\in\R^{d}:x^{'}\t_{0}=0\right\} }m\left(x\right)d{\cal L}^{d}\left(x\right)$
is trivially 0 for any function $m$, integrals with respect to the
$\left(d-1\right)$-dimensional Hausdorff measure of the form
\[
\int_{\left\{ x\in\R^{d}:x^{'}\t_{0}=0\right\} }m\left(x\right)d{\cal H}^{d-1}\left(x\right)
\]
is nontrivial (i.e., may take values other than $0$). 
\begin{lem}
\label{lem:Term3}For some positive semidefinite matrix of rank $d-1$,
we have
\[
P\left(g_{\t,h_{0}}-g_{\t_{0},h_{0}}\right)=-\left(\t-\t_{0}\right)^{'}V\left(\t-\t_{0}\right)+o\left(\norm{\t-\t_{0}}^{2}\right)
\]
with
\begin{equation}
V:=\int_{x^{'}\t_{0}=0}\frac{f\left(\rest 0x\right)}{f\left(\rest 0x\right)+1}xx^{'}p\left(x\right)d{\cal H}^{d-1}\left(x\right)\label{eq:V_def}
\end{equation}
where ${\cal H}^{d-1}$ denotes the $\left(d-1\right)$-dimensional
Hausdorff measure in $\R^{d}$.
\end{lem}
Lemma \ref{lem:Term3} can be viewed as a local-identification condition,
which says that $Pg_{\t,h_{0}}$ becomes smaller than $Pg_{\t_{0},h_{0}}$
locally with quadratic curvature as $\t$ moves away from the true
$\t_{0}$. Essentially, \eqref{eq:V_def} can be viewed as a ``surface
integral'' over the $\left(d-1\right)$-dimensional hyperplane $\left\{ x\in\R^{d}:x^{'}\t_{0}=0\right\} $.
Note that, even though $V$ has rank $d-1$ instead of $d$, $V$
should still be regarded to have ``full rank'' with respect to the
parameter space $\T=\S^{d-1}$, which also has dimension $d-1$ instead
of $d$. This is similar to the corresponding result in \citet*{kim1990cube}. 

Note that the formula of the Hessian matrix $V$ features the probability
density $f\left(\rest 0x\right)$ in the integrand, which reflects
the observation that the sign-restriction identification \ref{eq:Mono_Equiv}
is driven by the conditional median restriction and thus local in
nature. If, for example, $f\left(\rest 0x\right)=0$ for all $x\in{\cal X}$,
then the conditional median restriction is vacuous and thus identification
will fail. The dependence of the identification on $f\left(\rest 0x\right)$,
i.e., the ``conditional median density'', here is also featured
in \citet*{kim1990cube} and \citet*{horowitz1992smoothed}, as well
as more broadly in quantile regression settings. Hence, we assume
in Assumption \ref{assu:Basic} that $f\left(\rest 0x\right)$ is
bounded away from $0$. 

\subsubsection{Analysis of Term $T_{4}=P\left(g_{\hat{\protect\t},\hat{h}}-g_{\protect\t_{0},\hat{h}}-g_{\hat{\protect\t},h_{0}}+g_{\protect\t_{0},h_{0}}\right)$}

The last term, $T_{4}$, reflects the influence of the first-stage
nonparametric estimation on the second-stage M-estimation criterion
function, i.e., how $P\left(g_{\hat{\t},\hat{h}}-g_{\t_{0},\hat{h}}\right)$
differs from $P\left(g_{\hat{\t},h_{0}}-g_{\t_{0},h_{0}}\right).$
This term corresponds to the derivation of the influence function
through functional differentiation in standard semiparametric two-stage
asymptotic theory. 

We work with the following second-order Taylor expansion of $T_{4}$
w.r.t. $\t$ around $\t_{0}$:
\begin{align*}
 & P\left(g_{\t,h}-g_{\t_{0},h}-g_{\t,h_{0}}+g_{\t_{0},h_{0}}\right)=P\left(g_{\t,h}-g_{\t,h_{0}}\right)-P\left(g_{\t_{0},h}-g_{\t_{0},h_{0}}\right)\\
=\  & \Dif_{\t}P\left(g_{\t_{0},h}-g_{\t_{0},h_{0}}\right)\left(\t-\t_{0}\right)+\left(\t-\t_{0}\right)\Dif_{\t\t}P\left(g_{\t_{0},h}-g_{\t_{0},h_{0}}\right)\left(\t-\t_{0}\right)+o\left(\norm{\t-\t_{0}}^{2}\right).
\end{align*}
The leading term $\Dif_{\t}P\left(g_{\t_{0},h}-g_{\t_{0},h_{0}}\right)$
can be linearized through pathwise functional differentiation as
\begin{equation}
\Dif_{\t}P\left(g_{\t_{0},h}-g_{\t_{0},h_{0}}\right)=D_{h}\left[\Dif_{\t}Pg_{\t_{0},h_{0}},h-h_{0}\right]+O\left(\norm{h-h_{0}}_{\infty}\norm{\Dif_{x}\left(h-h_{0}\right)}_{\infty}\right),\label{eq:DPdg_lin}
\end{equation}
where the formula of $D_{h}\left[\Dif_{\t}Pg_{\t_{0},h_{0}},h-h_{0}\right]$
is derived in Lemma \ref{lem:Term4} below. With $\hat{\t}$ and $\hat{h}$
plugged in, the term $\left(\hat{\t}-\t_{0}\right)\Dif_{\t\t}P\left(g_{\t_{0},\hat{h}}-g_{\t_{0},h_{0}}\right)\left(\hat{\t}-\t_{0}\right)$
will become asymptotically negligible provided that $\Dif_{\t\t}P\left(g_{\t_{0},\hat{h}}-g_{\t_{0},h_{0}}\right)\pto0$
holds, which can be guaranteed by the convergence of $\Dif_{x}\hat{h}$
to $\Dif_{x}h_{0}$. 
\begin{assumption}
\label{assu:FirstStageDeriv} $\norm{\Dif_{x}\hat{h}-\Dif_{x}h_{0}}_{\infty}=O_{p}\left(c_{n}\right)$
with $c_{n}\downto0$.
\end{assumption}
\begin{lem}
\label{lem:Term4}Under Assumption \ref{assu:FirstStageDeriv}\emph{,
we have}
\begin{align*}
 & P\left(g_{\t,\hat{h}}-g_{\t_{0},\hat{h}}-g_{\t,h_{0}}+g_{\t_{0},h_{0}}\right)\\
= & D_{h}\left[P\Dif_{\t}g_{\t_{0},h_{0}},\hat{h}-h_{0}\right]^{'}\left(\t-\t_{0}\right)+O_{p}\left(\norm{\t-\t_{0}}a_{n}c_{n}\right)+o_{p}\left(\norm{\t-\t_{0}}^{2}\right)
\end{align*}
where 
\begin{align}
D_{h}\left[\Dif_{\t}Pg_{\t_{0},h_{0}},h-h_{0}\right] & :=\int_{x^{'}\t_{0}=0}\left[h\left(x\right)-h_{0}\left(x\right)\right]\frac{1}{f\left(\rest 0x\right)+1}xp\left(x\right)d{\cal H}^{d-1}\left(x\right).\label{eq:Dh_dh}
\end{align}
\end{lem}
The term $O_{p}\left(\norm{\t-\t_{0}}a_{n}c_{n}\right)$ will become
asymptotically negligible if $a_{n}c_{n}=o_{p}\left(\norm{\hat{\t}-\t_{0}}\right)$,
which can be viewed as a generalization/adaption of the usual ``$o_{p}\left(n^{-1/4}\right)$''
rate requirement on the first-stage convergence in standard semiparametric
two-stage asymptotic theory that features $n^{-1/2}$ convergence
rate for the final estimator $\hat{\t}$. As we will show in Theorem
\ref{thm:Thm_Bin_Rate} later, the requirement $\norm{\hat{h}-h_{0}}_{\infty}=o_{p}\left(\sqrt{\norm{\hat{\t}-\t_{0}}}\right)$
can be satisfied under proper smoothness condition on $h_{0}$.

Note that $D_{h}\left[\Dif_{\t}Pg_{\t_{0},h_{0}},\hat{h}-h_{0}\right]$
can be viewed as the convergence of a plug-in estimator of lower-dimensional
integral over the nonparametric function $h_{0}$ over the hyperplane
$\left\{ x:x^{'}\t_{0}=0\right\} $. Specifically, we can write 
\[
D_{h}\left[\Dif_{\t}Pg_{\t_{0},h_{0}},h-h_{0}\right]=L\left(\hat{h}\right)-L\left(h_{0}\right)
\]
with
\begin{equation}
L\left(h\right):=\int_{x^{'}\t_{0}=0}h\left(x\right)\frac{1}{f\left(\rest 0x\right)+1}xp\left(x\right)d{\cal H}^{d-1}\left(x\right).\label{eq:Lh_func}
\end{equation}
Note that $L\left(h\right)$ is a linear functional of $h$, and the
asymptotic behavior of the plugged-in estimator for linear functionals
has been widely studied in the literature on nonparametric and semiparametric
inference. While there are many results available for ``point evaluation
functionals'' and ``full-dimensional integration functionals'',
there are relatively few results for ``lower-dimensional integration
functionals'' like \eqref{eq:Lh_func}. Hence we develop results
for the asymptotic behavior of plug-in estimators of \eqref{eq:Lh_func}
in this paper. 

So far we have not restricted the form of the first-stage nonparametric
estimator $\hat{h}$, and thus all our results above hold for any
form of $\hat{h}$ that satisfies Assumption \ref{assu:FirstStageRate}.
However, now we will need to be more explicit about $\hat{h}$, and
focus our attention on the Nadaraya-Watson kernel estimators and linear
series estimators, which are two leading classes of nonparametric
estimators. We provide the required conditions and results for both
classes separately below. 

\subsubsection*{First Stage by Nadaraya-Watson Kernel Regression}

\noindent \textbf{Lemma 5a} \customlabel{lem:T4_Kern}{5a} \emph{Under
Assumption \ref{assu:KernelSeries}(a),
\begin{equation}
D\left[P\Dif_{\t}g_{\t_{0},h_{0}},\hat{h}-h_{0}\right]=O_{p}\left(\frac{1}{\sqrt{nb_{n}}}+b_{n}^{s}\right)\label{eq:DPDh_rate_Kern}
\end{equation}
Setting $b_{n}\sim n^{-\frac{1}{2s+1}}$ leads to the optimal rate
of convergence $n^{-\frac{s}{2s+1}}$. With undersmoothing bandwidth
$b_{n}=o\left(n^{-\frac{1}{2s+1}}\right)$, we have
\[
\sqrt{nb_{n}}D\left(P\Dif_{\t}g_{\t_{0},h_{0}},\hat{h}-h_{0}\right)\dto\cN\left({\bf 0},\O\right),
\]
with 
\begin{align*}
\O & :=\int G^{2}\left(t\right)dt\cd\int_{x^{'}\t_{0}=0}\frac{\s_{0}^{2}\left(x\right)}{(f\left(\rest 0x\right)+1)^2}xx^{'}p\left(x\right)d{\cal H}^{d-1}\left(x\right),\\
G\left(t\right) & :=\int_{x^{'}\t_{0}=0}K\left(x\right)d{\cal H}^{d-1}\left(x\right)\\
\s_{0}^{2}\left(x\right) & :=\text{Var}\left(\rest{Y_{i}}X_{i}=x\right)=\frac{1}{4}-h_{0}^{2}\left(x\right)
\end{align*}
}

Lemma \ref{lem:T4_Kern} shows that the asymptotics of $L\left(\hat{h}\right)$
is similar to the asymptotics of univariate nonparametric (kernel)
regressions. Specifically, the magnitude of the (square root of) variance
term in \eqref{eq:DPDh_rate_Kern} is $\left(nb_{n}\right)^{-1/2}$,
and consequently the optimal rate of convergence $n^{-\frac{s}{2s+1}}$,
do \textbf{not} depend on the dimension $d$ of the first-stage nonparametric
estimation of $h_{0}$. 

This is a highly intuitive result. It is well-known from the literature
that plug-in estimators of point evaluation functionals converge at
``fully nonparametric rate'' no faster than $n^{-\frac{s}{2s+d}}$,
while plug-in estimators of (regular) ``full-dimensional integral
functionals'' converge at ``parametric rate'' $n^{-\frac{1}{2}}$,
since the ``full-dimensional integration'' effectively reduces the
dimensionality of the estimation problem by aggregating information
(and errors) over the whole $d$-dimensional support of ${\cal X}$.
Here, we are dealing a ``$\left(d-1\right)$-dimensional integral'',
which can be viewed as an intermediate case between ``point evaluation''
and ``full-dimensional integral'' functionals, and as expected our
result shows that plug-in estimators of our $\left(d-1\right)$-dimensional
integral also features an ``intermediate'' convergence rate. This
result is also consistent to the one in \citet*{newey1994kernel},
who also demonstrates accelerated convergence rates for kernel estimation
of ``partial means'', which are defined as integrals over a subvector
of $x$.\footnote{The result on partial means in \citet*{newey1994kernel} requires
that the partial means are defined with respect to a given subvector
of $x$, while our result here covers linear combinations of the whole
vector of $x$ in the form of $x^{'}\t_{0}$.}

Lemma \ref{lem:T4_Kern} can be established by an adaption of the
proof in \citet*{newey1994kernel}. The key idea is the observation
that $G\left(t\right)$, defined as a lower-dimensional 
integral of the multivariate kernel function $K$ over the $\left(d-1\right)$-dimensional
hyperplane $\left\{ x:x^{'}\t_{0}=0\right\} $, itself qualifies as
a univariate kernel function. Furthermore, $G\left(t\right)$ is also
of smoothness order $s$. Hence, intuitively the $\left(d-1\right)$-dimensional
integral over $\left\{ x:x^{'}\t_{0}=0\right\} $ reduces the underlying
dimensionality of the kernel nonparametric regression, thus delivering
accelerated rate of convergence for $L\left(\hat{h}\right)$ relative
to $\hat{h}$. 

\subsubsection*{First Stage by Linear Series Regression}

\noindent \textbf{Lemma 5b }\customlabel{lem:T4_Series}{5b} \emph{Under
Assumptions \eqref{assu:Basic}, and \ref{assu:KernelSeries}(b),
\[
D\left[P\Dif_{\t}g_{\t_{0},h_{0}},\hat{h}-h_{0}\right]=O_{p}\left(\sqrt{\frac{J_{n}}{n}}+J_{n}^{-s}\right).
\]
With $J_{n}^{-1}=o\left(n^{-\frac{1}{2s+1}}\right)$, we have
\[
\sqrt{nJ_{n}^{-1}}D\left(P\Dif_{\t}g_{\t_{0},h_{0}},\hat{h}-h_{0}\right)\dto\cN\left({\bf 0},\O\right)
\]
for some positive semidefinite matrix with rank $d-1$ and $\t_{0}^{'}\O\t_{0}=0$.}

\subsubsection{\label{subsec:Rate}Convergence Rate and Asymptotic Normality of
$\hat{\protect\t}$}

\noindent Now, we combine the results from Lemmas \ref{lem:Term1},
\ref{lem:Term2}, \ref{lem:Term3}, and \ref{lem:Term4} to obtain
the convergence rate of the ReLU-based estimator. In the following
we use the notation of kernel bandwidth $b_{n}$ as if the first-stage
estimator $\hat{h}$ is given by the Nadaraya-Waston kernel regression.
However, note that the arguments also apply to the setting with linear
series first stages simply with $b_{n}$ replaced by $1/J_{n}$, where
$J_{n}$ is the univariate sieve dimension (with the multivariate
sieve dimension given by $K_{n}=J_{n}^{d}$).

Plugging the implications of Lemmas \ref{lem:Term1}, \ref{lem:Term2},
\ref{lem:Term3}, and \ref{lem:Term4} into the decomposition \eqref{eq:EP_decom},
we have

\begin{align*}
0\leq\  & \P_{n}\left(g_{\hat{\t},\hat{h}}-g_{\t_{0},\hat{h}}\right)\\
\asymp\  & o_{p}\left(\norm{\hat{\t}-\t_{0}}^{2}\right) & T_{1}+T_{2}\\
 & -\left(\hat{\t}-\t_{0}\right)^{'}V\left(\hat{\t}-\t_{0}\right)+o_{p}\left(\norm{\hat{\t}-\t_{0}}^{2}\right) & T_{3}\\
 & +Z_{n}^{'}\left(\hat{\t}-\t_{0}\right)+O_{p}\left(a_{n}c_{n}\norm{\hat{\t}-\t_{0}}\right)+o_{p}\left(\norm{\hat{\t}-\t_{0}}^{2}\right) & T_{4}
\end{align*}
where
\[
Z_{n}:=D\left[P\Dif_{\t}g_{\t_{0},h_{0}},\hat{h}-h_{0}\right]=O_{p}\left(\frac{1}{\sqrt{nb_{n}}}+b_{n}^{s}\right).
\]
It turns out that the convergence rate of $\hat{\t}$ is driven by
the convergence rate of $Z_{n}$ in $T_{4}$, provided that 
\[
O_{p}\left(a_{n}c_{n}\norm{\hat{\t}-\t_{0}}\right)=o_{p}\left(\norm{\hat{\t}-\t_{0}}^{2}\right),
\]
i.e. $a_{n}c_{n}=o_{p}\left(\norm{\hat{\t}-\t_{0}}\right)$, which
can be guaranteed by a condition on $s$. Hence, $T_{1}$ and $T_{2}$
are asymptotically negligible, while $T_{3}$ and $T_{4}$ are the
asymptotically leading terms. 

\section{\label{sec:MISC}General Framework: Multi-Index Single-Crossing Condition Models}

\subsection{\label{subsec:MISC_setup}RMS in the Multi-Index Single-Crossing
Framework}

We now introduce the multi-index single-crossing (MISC) condition framework as proposed in \cite{gao2020robust}, which
generalizes the single-index sign-alignment restriction \eqref{eq:Mono_Equiv} to a $J$-dimensional setting.

Formally, consider a random sample $\left(Y_{i},X_{i}\right)_{i=1}^{n}$
where $Y_{i}$ is an outcome with support ${\cal Y}\subseteq\R^{d_{y}}$
and
\[
X_{i}:=\left(X_{i1},...,X_{iJ}\right)\in\R^{d\times J}
\]
is a $d\times J$ random matrix with support ${\cal X}\subseteq\R^{d\times J}$.
Let $h_{0}:{\cal X}\to\R$ be a real-valued functional of the conditional
distribution of $Y_{i}$ given $X_{i}$ that is directly identified
and nonparametrically estimable from the data.\footnote{For example, in the binary choice model in Section \ref{sec:BinChoice},
we take $h_{0}\left(x\right)=\E\left[\left(Y_{i}-\frac{1}{2}\right)\mid X_{i}=x\right]$.
In other applications $h_{0}$ can be a conditional quantile, a conditional
variance, or a difference of such functionals across two states.}

We are interested in a direction parameter $\t_{0}\in\T\subseteq\S^{d-1}$
that enters the model through the $J$ parametric indexes
\[
X_{ij}^{'}\t_{0},\qquad j=1,...,J.
\]

\begin{defn}[Multi-Index Single-Crossing Condition]
Given observable $\left(Y_{i},X_{i}\right)$ and a pair $\left(h_{0},\t_{0}\right)$,
we say that $\left(h_{0},\t_{0}\right)$ satisfies the (\emph{multi-index
single-crossing}) MISC condition if, for all $x=\left(x_{1},...,x_{J}\right)\in{\cal X}$,
\begin{align}
x_{j}^{'}\t_{0}>0,\ \forall j=1,...,J & \quad\imp\quad h_{0}\left(x\right)\geq0,\nonumber \\
x_{j}^{'}\t_{0}<0,\ \forall j=1,...,J & \quad\imp\quad h_{0}\left(x\right)\leq0.\label{eq:MISC}
\end{align}
The condition is said to be \emph{strict} if the inequalities on the
right-hand side of \eqref{eq:MISC} are strict, i.e., $h_{0}\left(x\right)>0$ whenever $x_{j}^{'}\t_{0}>0$
for all $j$, and $h_{0}\left(x\right)<0$ whenever $x_{j}^{'}\t_{0}<0$
for all $j$.
\end{defn}

When $J=1$, \eqref{eq:MISC} reduces exactly to the sign-alignment
restriction \eqref{eq:Mono_Equiv} used in the binary choice model
in Section \ref{sec:BinChoice}. For $J\ge2$, the MISC condition
requires the sign of $h_{0}\left(x\right)$ to align with the common
sign of the $J$ indexes whenever those indexes all agree. Importantly,
it imposes no restriction on $h_{0}\left(x\right)$ when the $J$
indexes have mixed signs.

In many applications $X_{i}$ arises as a (possibly nonlinear) transformation
of a lower-dimensional regressor $Z_{i}$, so that $X_{i}=\phi\left(Z_{i}\right)$
for some known transformation $\phi$. In that case it is convenient
to state MISC in terms of such transformed regressors.

\begin{rem}[Weak MISC with transformed regressors]
Let $W_{i}=\phi\left(X_{i}\right)$ for a known measurable map $\phi:{\cal X}\to\R^{d\times J}$,
and write $W_{i}=\left(W_{i1},...,W_{iJ}\right)$. We say that $\left(h_{0},\t_{0}\right)$
satisfies the (\emph{weak}) MISC condition with respect to $W_{i}$
if, for all $x\in{\cal X}$ and $w=\phi\left(x\right)$,
\begin{align}
w_{j}^{'}\t_{0}>0,\ \forall j=1,...,J & \quad\imp\quad h_{0}\left(x\right)\geq0,\nonumber \\
w_{j}^{'}\t_{0}<0,\ \forall j=1,...,J & \quad\imp\quad h_{0}\left(x\right)\leq0.\label{eq:MISC-weak}
\end{align}
The strict version is defined analogously. In what follows, we suppress
the distinction when it is clear from context whether $X_{i}$ denotes
the original regressors or a transformed version.
\end{rem}

~

%\paragraph{RMS criterion under MISC.}

The RMS estimator extends naturally to the MISC framework. Given a
candidate direction $\t\in\T$ and a function $h:{\cal X}\to\R$,
define
\begin{align}
g_{+,\t,h}\left(x\right) & :=\left[h\left(x\right)-\min_{1\le j\le J}\left(-x_{j}^{'}\t\right)_{+}\right]_{+},\label{eq:def_gplus_MISC}\\
g_{-,\t,h}\left(x\right) & :=\left[-h\left(x\right)-\min_{1\le j\le J}\left(x_{j}^{'}\t\right)_{+}\right]_{+},\label{eq:def_gminus_MISC}
\end{align}
and the population criterion
\[
Q\left(\t\right):=Q_{+}\left(\t\right)+Q_{-}\left(\t\right),\qquad Q_{\pm}\left(\t\right):=\E\left[g_{\pm,\t,h_{0}}\left(X_{i}\right)\right].
\]
Then clearly, $$\t_{0}\in\arg\max_{\t\in\T}Q\left(\t\right).$$

Intuitively, $g_{+,\t,h}$ penalizes violations of the ``positive
sign'' restriction in \eqref{eq:MISC} when $h\left(x\right)$ is
positive but some index $x_{j}^{'}\t$ is nonpositive, while $g_{-,\t,h}$
penalizes violations of the ``negative sign'' restriction when $h\left(x\right)$
is negative but some index $x_{j}^{'}\t$ is nonnegative. The inner
$\min$ and ReLU terms ensure that, for each realization $x$, only
the index that is closest to the kink at zero contributes to the loss.

%
% \begin{lem}
% \label{lem:MISC_maximizer}Suppose the strict MISC condition \eqref{eq:MISC}
% holds. 
% \end{lem}

% \begin{proof}
% For completeness, we sketch the argument. Consider first the case
% $h_{0}\left(x\right)>0$. By the contrapositive of \eqref{eq:MISC},
% if $h_{0}\left(x\right)>0$ then there exists at least one $j^{*}$
% such that $x_{j^{*}}^{'}\t_{0}\ge0$. At $\t=\t_{0}$ this implies
% $\min_{j}\left(-x_{j}^{'}\t_{0}\right)_{+}=0$ and hence $g_{+,\t_{0},h_{0}}\left(x\right)=\left[h_{0}\left(x\right)\right]_{+}=h_{0}\left(x\right)$.
% For any other $\t$, the inner minimum in \eqref{eq:def_gplus_MISC}
% is weakly positive and therefore
% \[
% g_{+,\t,h_{0}}\left(x\right)\le\left[h_{0}\left(x\right)\right]_{+}=g_{+,\t_{0},h_{0}}\left(x\right).
% \]
% A symmetric argument applies when $h_{0}\left(x\right)<0$, using
% $g_{-,\t,h_{0}}$ and the second implication in \eqref{eq:MISC}.
% Integrating with respect to the distribution of $X_{i}$ yields $Q\left(\t\right)\le Q\left(\t_{0}\right)$
% for all $\t$, so $\t_{0}$ is a maximizer of $Q$.
% \end{proof}

Given a first-stage nonparametric estimator $\hat{h}$ of $h_{0}$,
we define the sample criterion
\[
\hat{Q}\left(\t\right):=\frac{1}{n}\sum_{i=1}^{n}\left\{ g_{+,\t,\hat{h}}\left(X_{i}\right)+g_{-,\t,\hat{h}}\left(X_{i}\right)\right\}
\]
and the RMS estimator under the MISC framework as
\[
\hat{\t}:=\arg\max_{\t\in\T}\hat{Q}\left(\t\right).
\]

The binary choice model in Section \ref{sec:BinChoice} is a strict
special case of this framework with $J=1$ and $h_{0}\left(x\right)=\E\left[\left(Y_{i}-\frac{1}{2}\right)\mid X_{i}=x\right]$.
In that case $g_{+,\t,h},g_{-,\t,h}$ reduce to the composite ReLU
functions in \eqref{eq:def_g} and the RMS estimator coincides with
the estimator studied in Section \ref{subsec:SetupResults}. When
$J\ge2$ or when $h_{0}$ is a functional other than a conditional
expectation, the traditional MS estimator cannot be applied, but the
RMS estimator remains well-defined under MISC.

\medskip{}

\noindent The MISC framework nests a large class of models, including
binary choice with awareness, selection models with multiple latent
thresholds, and panel models with multiple time indices; detailed
examples can be provided depending on the application. The key common
feature is that $h_{0}\left(x\right)$ is monotone in a common direction
$\t_{0}$ whenever the $J$ indexes share the same sign.

To further explain the economic relevance of the MISC condition framework and the general applicability of the RMS estimator, we now provide some concrete examples\footnote{Section 4  \cite{gao2020robust} also discusses some of the examples below, as well as other examples under the MISC condition framework with endogeneity.} below along with a discussion about the related work in each specific application. 

\begin{example}[Binary Choice with Awareness]
\label{exa:Bin_Aware} Consider the following modification of the
binary choice model above
\[
y_{i}=\ind\left\{ X_{i1}^{'}\t_{01}\geq u_{i}\right\} \cd\ind\left\{ X_{2i}^{'}\t_{0}\geq v_{i}\right\} 
\]
where $y_{i}$ denotes whether consumer $i$ purchases a certain,
$X_{i1}$ denotes a vector of covariates that influence the consumer's
utility from a product, and $X_{i2}$ denotes a vector of covariates
that influence the consumer's awareness of the product (such as advertising).
Here $J=2$, $X_{i}:=\left(X_{i1},X_{i2}\right)$, $W_{i1}:=X_{i1}$,
and $W_{i2}:=X_{i2}$. Let the functional $h_{0}$ be defined by $h_{0}\left(x\right):=\E\left[\rest{y_{i}}X_{i}=x\right]-\frac{1}{4}.$
Then, under the conditional median restrictions $\text{med}\left(\rest{u_{i}}X_{i}\right)=\text{med}\left(\rest{v_{i}}X_{i}\right)=0$
and the conditional independence restriction $\rest{u_{i}\indep v_{i}}X_{i}$,
it can be shown that 
\begin{align*}
X_{i1}^{'}\t_{01}>0,\ X_{i2}^{'}\t_{02}>0 & \imp\quad h_{0}\left(X_{i}\right)>0,\\
X_{i1}^{'}\t_{01}<0,\ X_{i2}^{'}\t_{02}<0 & \imp\quad h_{0}\left(X_{i}\right)<0,
\end{align*}
again satisfying the MISC condition.
\end{example}
\begin{example}[Panel Multinomial Choice]
\label{exa:PMC} Consider the following panel multinomial choice
model studied in one of the PI's working papers \citet*{gao2020robust},
\[
y_{ijt}=\ind\left\{ u\left(X_{ijt}^{'}\b_{0},\,A_{ij},\,\e_{ijt}\right)=\max_{k\in\left\{ 1,...,J\right\} }u\left(X_{ikt}^{'}\b_{0},\,A_{ik},\,\e_{ikt}\right)\right\} 
\]
where $y_{ijt}$ is a binary variable indicating whether consumer
$i$ chooses product $j$ at time $t$, $X_{ijt}$ is a vector of
observable covariates, $A_{ij}$ is an unobserved fixed effect that
can be infinite dimensional, $\e_{ijt}$ is an unobserved idiosyncratic
taste shock, and the utility function $u$ is assumed to be unknown
but increasing in its first argument. \citet*{gao2020robust} proposes
a novel strategy to identify and estimate the finite-dimensional parameter
$\b_{0}$ , and the key idea is to leverage the monotonicity of $u$
to obtain a MISC condition through a sequence of intertemporal differencing
and cross-sectional averaging. Specifically, focusing on a pair of
time periods $\left(t,s\right)$ and a particular product $j_{0}$
for illustration, define $\t_{0j}:=\b_{0}$, $X_{i}:=\left(\left(X_{ijt}\right)_{j=1}^{J},\left(X_{ijs}\right)_{j=1}^{J}\right)$,
$h_{0}\left(X_{i}\right):=\E\left[\rest{y_{ij_{0}t}-y_{ij_{0}s}}X_{i}\right]$
and
\[
W_{ij}:=\begin{cases}
X_{ijt}-X_{ijs}, & j=j_{0},\\
-\left(X_{ijt}-X_{ijs}\right) & j\neq j_{0}.
\end{cases}
\]
\citet*{gao2020robust} then shows that, under quite general conditions,
the following MISC condition holds
\begin{align*}
W_{ij}^{'}\t_{0j}>0,\ \forall j=1,...,J\quad & \imp\quad h_{0}\left(X_{i}\right)>0,\\
W_{ij}^{'}\t_{0j}<0,\ \forall j=1,...,J\quad & \imp\quad h_{0}\left(X_{i}\right)<0.
\end{align*}
\end{example}
\begin{example}[Dyadic Network Formation]
\label{exa:NetForm} Consider the following dyadic network formation
model studied in \citet*{gao2023logical}, which is a generalization
of the one studied in \citet{graham2017econometric}:
\begin{align*}
\E\left[\rest{y_{ij}}X_{i},X_{j},A_{i},A_{j}\right]= & \psi\left(w\left(X_{i},X_{j}\right)^{'}\t_{0},A_{i},A_{j}\right)
\end{align*}
Here $y_{ij}$ is a binary outcome indicating whether individuals
$i$ and $j$ are linked in an undirected network, $X_{i}$ and $X_{j}$
are the individuals' observable covariates, $w\left(X_{i},X_{j}\right)$
is a known pairwise transformation of individual covariates (with
the leading example being $w_{h}\left(X_{i},X_{j}\right):=\left|X_{i,h}-X_{j,h}\right|$
for each coordinate $h=1,...,d_{x}$), $A_{i}$ and $A_{j}$ are unobserved
individual degree heterogeneity terms, and $\psi:\R^{3}\to\R$ is
an unknown function assumed to be multivariate increasing in all its
three arguments. \citet*{gao2023logical} proposes a method, called
``logical differencing'', to cancel out the unobserved heterogeneity
terms $A_{i}$ despite the lack of additive separability in the model,
a technical complication that arises naturally under nontransferable
utility settings. Specifically, fixing a particular pair of individuals
$\ol i$ and $\ol j$ and two generic realizations $\ol x,\ul x$
of $X_{i}$, it can be shown that, with
\[
\ol w:=w\left(x_{\ol j},\ol x\right)-w\left(x_{\ol i},\ol x\right),\quad\ul w:=w\left(x_{\ol i},\ul x\right)-w\left(x_{j},\ul x\right),
\]
and
\begin{align*}
h_{0}\left(\ol x,\ul x\right):= & \left(\E\left[\rest{y_{\ol ik}-y_{\ol jk}}X_{k}=\ol x\right]\right)_{+}\E\left[\rest{y_{\ol ik}-y_{\ol jk}}X_{k}=\ul x\right],\\
 & -\left(\E\left[\rest{y_{\ol jk}-y_{\ol ik}}X_{k}=\ol x\right]\right)_{+}\E\left[\rest{y_{\ol jk}-y_{\ol ik}}X_{k}=\ul x\right]
\end{align*}
the weak MISC condition is satisfied (under quite mild additional
conditions):
\begin{align*}
\ol w^{'}\t_{0}>0,\ul w^{'}\t_{0}>0\quad & \imp\quad h_{0}\left(\ol x,\ul x\right)\geq0,\\
\ol w^{'}\t_{0}<0,\ul w^{'}\t_{0}<0\quad & \imp\quad h_{0}\left(\ol x,\ul x\right)\leq0.
\end{align*}
\end{example}
\begin{example}[Conditional Quantile Model for Continuous Outcomes]
 Consider the following model
\[
y_{i}=\phi\left(X_{i}^{'}\t+\e_{i}\right),\quad\text{med}\left(\rest{\e_{i}}X_{i}\right)=0,
\]
where $\phi$ is some unknown strictly increasing function. If ${\bf 0}\in Supp\left(X_{i}\right)$,
we can take $h_{0}$ to be the difference in conditional median functions
\[
h_{0}\left(x\right):=\text{med}\left(\rest{y_{i}}X_{i}=x\right)-\text{med}\left(\rest{y_{i}}X_{i}=0\right),
\]
so that \eqref{eq:MISC} holds since
\begin{align*}
\text{med}\left(\rest{y_{i}}X_{i}=x\right) & =\phi\left(\text{med}\left(\rest{X_{i}^{'}\t+\e_{i}}X_{i}=x\right)\right)\\
 & =\phi\left(x_{i}^{'}\t+\text{med}\left(\rest{\e_{i}}X_{i}=x\right)\right)=\phi\left(x_{i}^{'}\t\right).
\end{align*}

Alternatively, we could also state the single-crossing condition in
terms of pairwise differences by
\[
h_{0}\left(\ol x,\ul x\right):=\text{med}\left(\rest{y_{i}}X_{i}=\ol x\right)-\text{med}\left(\rest{y_{i}}X_{i}=\ul x\right)
\]
so that
\[
h_{0}\left(\ol x,\ul x\right)\lessgtr0\quad\iff\quad\left(\ol x-\ul x\right)^{'}\t\lessgtr0,
\]
which is a special case of \eqref{eq:MISC} with $J=2$ and $g\left(\ol x,\ul x\right)=\ol x-\ul x.$
\end{example}
\begin{example}[Stochastic Volatility for Continuous Outcomes]
Consider the following simple ``stochastic volatility'' model of
some centered (mean-zero) variable $y_{t}$:
\begin{align*}
y_{t} & =\s\left(X_{t}^{'}\t+\e_{t}\right)\cd u_{t}
\end{align*}
where $\s$ is some unknown strictly increasing function and $u_{t}$
is mean-zero exogenous error with $\E\left[\rest{u_{t}^{2}}X_{t}\right]=1$.
Suppose that $\e_{t}\indep\left(X_{t},u_{t}\right)$. Then we can
set
\begin{align*}
h_{0}\left(\ol x,\ul x\right) & :=\E\left[\rest{y_{t}^{2}}X_{t}=\ol x\right]-\E\left[\rest{y_{t}^{2}}X_{t}=\ul x\right]\\
 & =\E\left[\s^{2}\left(\ol x^{'}\t+\e_{t}\right)-\s^{2}\left(\ul x^{'}\t+\e_{t}\right)\right]
\end{align*}
so that
\[
h_{0}\left(\ol x,\ul x\right)\lessgtr0\quad\iff\quad\left(\ol x-\ul x\right)^{'}\t\lessgtr0.
\]
\end{example}
\noindent It should be pointed out that the above are just a few illustrations
of many plausible econometric models nested under the MISC condition
framework. Given that the exact specifications of $y,X,\phi,h_{0}$
are left largely unrestricted, they can be user-configured in very
flexibly manners depending on the economic contexts: for example,
$X$ can be decomposed into an ``endogenous/structural'' part and
an ``exogenous/IV'' part, while $W=\phi\left(X\right)$ and $h_{0}\left(X\right)$
may involve a subvector or the whole of $X$ with potentially nonlinear
transformations.

One main advantage of the MISC framework lies in its ability to identify
and estimate index parameters in models with rich forms of unobserved
heterogeneity and additively nonseparable interactions between modeling
ingredients.

\subsection{RMS Asymptotic Theory under MISC}
\label{subsec:MISC-asymptotics}

We now derive the convergence rate and asymptotic distribution of the RMS
estimator $\hat\theta$ in the multi-index single-crossing (MISC) framework of
Section~3.1. As in the single-index case, the key ingredients are: (i) a
linearization of the effect of first-stage estimation errors through a
lower-dimensional submanifold integral functional, and (ii) a local quadratic
expansion of the population criterion $Q(\theta)$ around $\theta_0$. In the
MISC case, both objects have a particularly transparent form.

Recall that
\[
Q(\theta;h)
:= P g_{\theta,h}, \qquad
Q(\theta) := Q(\theta;h_0) = P g_{\theta,h_0},
\]
and define the (vector-valued) directional derivative functional
\[
L(h)
:= D_h\big(P \nabla_\theta g_{\theta_0,h_0}\big)\big[h-h_0\big]
\in\mathbb{R}^d.
\]
Throughout this subsection, we view $L(h)$ as a map on a suitable function
space $H$ containing $h_0$ and the first-stage estimator $\hat h$.

The next lemma collects the two structural properties that drive the
asymptotics: a submanifold-integral representation of the linear functional
$L(h)$ and a quadratic expansion of $Q(\theta)$ around $\theta_0$.

\begin{lem}[Asymptotics via Submanifold Integrals]
\label{lem:MISC-curv-L}
Under the strict MISC condition~\textup{(20)} hold,

\begin{enumerate}
\item[(a)] %\textbf{Submanifold linearization of the directional derivative.}
For any $c\in\mathbb{R}^d$, define the scalar functional
\[
\Gamma_c(h)
:= c' P\nabla_\theta g_{\theta_0,h}.
\]
Then $\Gamma_c$ is Fr\'echet differentiable at $h_0$ and its derivative
satisfies
\begin{equation}
D_h\Gamma_c(h_0)[v]
= \sum_{j=1}^J \int_{\{x\in\mathcal{X}: x_j'\theta_0=0\}}
v(x)\, w_{c,j}(x)\, d\mathcal{H}^{d-1}(x),
\qquad \forall v\in H,
\label{eq:MISC-submanifold-deriv}
\end{equation}
for some uniformly bounded weight functions $w_{c,j}:\mathcal{X}\to\mathbb{R}$.
In particular, each component of $L(h)$ can be written as a finite sum of
integrals of $(h-h_0)$ over the hyperplanes $\{x:x_j'\theta_0=0\}$.

\item[(b)] 
There exists a symmetric positive semidefinite $d\times d$ matrix $V$ of rank
$d-1$ such that, for all $\theta$ in a neighborhood of $\theta_0$ with
$\|\theta\|=1$,
\begin{equation}
Q(\theta) - Q(\theta_0)
= -(\theta-\theta_0)'V(\theta-\theta_0)
+ o\big(\|\theta-\theta_0\|^2\big),
\label{eq:MISC-quadratic}
\end{equation}
and $V\theta_0=0$. Moreover, $V$ admits the representation
\begin{equation}
V
= \sum_{j=1}^J \int_{\{x\in\mathcal{X}: x_j'\theta_0=0\}}
m_j(x,\theta_0)\, x_j x_j' p(x)\, d\mathcal{H}^{d-1}(x),
\label{eq:MISC-V-rep}
\end{equation}
for some nonnegative Lipschitz functions $m_j(\cdot,\theta_0)$,
$j=1,\dots,J$, and $(d-1)$-dimensional Hausdorff measure
$\mathcal{H}^{d-1}$.
\end{enumerate}
\end{lem}

Lemma~\ref{lem:MISC-curv-L} shows that the second-stage curvature is governed by a $(d-1)$-dimensional matrix $V$ and that the first-stage impact enters only through submanifold integrals of $h-h_0$ over those $(d-1)$-dimensional hyperplanes. This is precisely the setting analyzed in \cite{chen2025semiparametric}, with submanifold dimension $m=d-1$ (codimension $d-m=1$). 

\begin{assumption}
\label{ass:MISC-CG}
Suppose that:
\begin{enumerate}
\item[(i)] The true function $h_0$ belongs to a H\"older (or Sobolev) ball of
smoothness order $s>1$ on a compact support $\mathcal{X}\subset\mathbb{R}^d$.

\item[(ii)] The first-stage estimator $\hat h$ is either a kernel or linear series
(sieve) estimator constructed as in Section~2, with smoothing parameter
(bandwidth or sieve dimension) chosen so that the conditions of Assumptions~6–8
in Chen and Gao (2025) hold for the regressors $X_i$ and the basis. In
particular, if $K_n$ denotes the sieve dimension, then
\[
K_n\log K_n / n \to 0
\quad\text{and}\quad
K_n^{-s/d} = o\Big(\sqrt{K_n^{(d-1)/d}/n}\Big).
\]

\item[(iii)] For each $c\in\mathbb{S}^{d-1}$, the scalar functional
$\Gamma_c(h)=c'P\nabla_\theta g_{\theta_0,h}$ satisfies the linearization
and regularity conditions in Assumptions~9–11 of \cite{chen2025semiparametric} with
submanifold dimension $m=d-1$ and level-set function
$g_j(x)=x_j'\theta_0$, $j=1,\dots,J$. 
% In particular, its pathwise derivative
% has the form \eqref{eq:MISC-submanifold-deriv}, the nonlinear remainder is
% negligible relative to the linear term, and the conditional variance of the
% first-stage regression errors satisfies the corresponding Lindeberg condition.
\end{enumerate}
\end{assumption}

Assumption~\ref{ass:MISC-CG}(c) is essentially a restatement, in our notation,
of the high-level conditions required to apply Theorems~2 and~3 of \cite{chen2025semiparametric} to the functionals $c'L(h)$. Under these conditions, those
theorems yield both the convergence rate and the asymptotic normality of
$L(\hat h)$ as an estimator of $L(h_0)$.

We can now state the main result of this subsection.

\begin{thm}[RMS Asymptotics under MISC]
\label{thm:Asymp-MISC}
Suppose the MISC condition \eqref{eq:MISC}, and Assumption~\ref{ass:MISC-CG} hold.

\begin{enumerate}
\item[(a)] %\textbf{Linear functional $L(\hat h)$.}
For the linear functional $L(h)$ defined above, under undersmoothing,
\begin{equation}
c_n\big(L(\hat h) - L(h_0)\big)
\;\xrightarrow{d}\;
\mathcal{N}(0,\Omega),
\label{eq:Lhat-CLT}
\end{equation}
with $c_n$ can be taken to be slower than but arbitrarily close to $n^{-s/(2s+1)}$.

\item[(b)] 
Let $\hat\theta$ denote the RMS estimator under MISC,
% \[
% \hat\theta := \arg\max_{\theta\in\Theta\subset\mathbb{S}^{d-1}} \hat Q(\theta),
% \qquad
% \hat Q(\theta) := P_n g_{\theta,\hat h}.
% \]
Then
\[
c_n(\hat\theta - \theta_0)
= - V^{-} c_n L(\hat h) + o_p(1),
\]
where $V$ is the Hessian in \eqref{eq:MISC-quadratic} and $V^{-}$
its Moore–Penrose inverse restricted to the tangent space orthogonal to $\theta_0$. Consequently, with undersmoothing,
\begin{equation}
c_n(\hat\theta - \theta_0)
\;\xrightarrow{d}\;
\mathcal{N}\big(0,\,V^{-}\Omega V^{-}\big).
\label{eq:thetahat-CLT}
\end{equation}
\end{enumerate}
\end{thm}

\begin{rem}[Effective one-dimensional rate in the $J$-index case]
\label{rem:MISC-1D-rate}
By Lemma~\ref{lem:MISC-curv-L}(ii), the submanifold functional $L(h)$ depends
on $h$ only through its restriction to the $(d-1)$-dimensional hyperplanes
$\{x:x_j'\theta_0=0\}$, $j=1,\dots,J$. The analysis in \cite{chen2025semiparametric}
shows that, for kernel or sieve estimators of $h_0$ on a $d$-dimensional
support, the minimax-optimal rate for such submanifold integrals is
$n^{-s/(2s+1)}$, independent of $J$. Thus $c_n=n^{s/(2s+1)}$ in
Theorem~\ref{thm:Asymp-MISC}, and the RMS estimator under MISC achieves the
same “one-dimensional” nonparametric rate as in the single-index binary choice
model. Increasing $J$ affects only the constants and the asymptotic variance
matrix $V^{-}\Omega V^{-}$, not the convergence rate.
\end{rem}

\section{DNN-Based Maximum Score Estimator}
\label{sec:NN}

In this section we show how the RMS estimator can be further adapted to be implemented within a neural network architecture. The key observation is that the RMS criterion is itself a composition of ReLU units with a simple, interpretable structure. This allows us to view the RMS estimator as a special multi-layer network with a dedicated ``RMS layer'' that extracts the sign information of the index parameter $\theta$, and to estimate $\theta$ using standard machine learning software.

\subsection{RMS as a Special Neural Network Layer}
\label{subsec:NN-layer}

We first describe the single-index binary choice model. Let $x\in\R^d$ denote
the covariate and recall that in Section~\ref{sec:BinChoice} we defined, for a
generic function $h$ and direction $\theta\in\T\subset\S^{d-1}$,
\[
g_{+,\theta,h}(x)
:= \bigl[h(x) - [-x'\theta]_+\bigr]_+,
\qquad
g_{-,\theta,h}(x)
:= \bigl[-h(x) - [x'\theta]_+\bigr]_+,
\]
and the RMS population criterion $Q(\theta)=E[g_{+,\theta,h_0}(X_i) +
g_{-,\theta,h_0}(X_i)]$. These maps are compositions of three elementary
operations:
\begin{enumerate}
\item a \emph{directional projection} $s(x;\theta)=x'\theta$;
\item a \emph{sign-extracting pair} of ReLU units $[s(x;\theta)]_+$ and
$[-s(x;\theta)]_+$; and
\item a final \emph{RMS transform} that compares $h(x)$ to the
ReLU-transformed index via an outer ReLU.
\end{enumerate}
This structure can be encoded as a small neural network module
$R_\theta(h)(x)$ that takes as input the scalar $h(x)$ and the vector $x$,
computes $x'\theta$, passes it through ReLUs, and outputs $g_{+,\theta,h}(x)$
and $g_{-,\theta,h}(x)$ (or their difference). In particular, for any fixed
$\theta$, $h\mapsto R_\theta(h)$ is a Lipschitz, piecewise linear operator.

A convenient way to embed RMS into a network is to treat $h$ as the output of a
generic multi-layer perceptron $f_\beta:\R^d\to\R$ with parameters
$\beta\in\R^p$, and then apply the RMS layer to $(x,f_\beta(x))$. In notation,
set
\[
g_{+}(x;\theta,\beta)
:= \bigl[f_\beta(x) - [-x'\theta]_+\bigr]_+,
\qquad
g_{-}(x;\theta,\beta)
:= \bigl[-f_\beta(x) - [x'\theta]_+\bigr]_+,
\]
and define
\[
h_{\theta,\beta}(x) := g_{+}(x;\theta,\beta) - g_{-}(x;\theta,\beta).
\]
The map $x\mapsto h_{\theta,\beta}(x)$ is then a neural network with one
special ``RMS layer'' on top of a generic (deep) regression network $f_\beta$. When
$\theta=\theta_0$ and $f_\beta$ approximates $h_0$, the outputs
$(g_{+},g_{-})$ implement the same sign-alignment structure as in the
population RMS criterion, and the resulting $h_{\theta,\beta}$ inherits the
economic interpretation of $h_0$.

\subsection{DNN-Based MISC Estimation}
\label{subsec:NN-MISC}

In the $J$-index MISC setting of Section~\ref{sec:MISC}, the relevant
population criterion is defined based on the following: for each $x=(x_1,\dots,x_J)$,
\[
g_{+,\theta,h_0}(x)
= \Bigl[h_0(x) - \bigl(\min_{1\le j\le J}(-x_j'\theta)_+\bigr)\Bigr]_+,
\qquad
g_{-,\theta,h_0}(x)
= \Bigl[-h_0(x) - \bigl(\min_{1\le j\le J}(x_j'\theta)_+\bigr)\Bigr]_+.
\]
which can be encoded in a neural network with the following special architecture:
\begin{enumerate}
\item A MLP neural network to approximate $h_0$.
\item A \emph{multi-index generation layer} that computes the $J$ scalar indexes
$s_j(x;\theta)=x_j'\theta$ and their ReLU transforms
$[s_j(x;\theta)]_+$, $[-s_j(x;\theta)]_+$.
\item A \emph{MISC aggregation layer} that takes the elementwise minimum
\[
u(x;\theta) := \min_j [-s_j(x;\theta)]_+,
\qquad
v(x;\theta) := \min_j [s_j(x;\theta)]_+,
\]
and passes them, together with $h(x)$, through outer ReLUs as above.
\end{enumerate}
The resulting multi-layer neural network encodes exactly the MISC conditions as in  Section~\ref{sec:MISC}. The MISC parameter $\theta$ appears only in the linear projections $x_j'\theta$ inside this special layer, while the possibly high-dimensional parameters $\beta$ govern flexible, nonparametric features through $h(x)=f_\beta(x)$.

From an applied perspective, one of the main appeals of the DNN-based MISC formulation
is that it provides a principled way to extract an economically meaningful index parameter $\theta$ from a high-dimensional black-box DNN.

% Within econometrics, our setup is also closely related in spirit to the literature on valid inference with machine learning, where a low-dimensional target parameter is estimated in the presence of high-dimensional or nonparametric nuisance components estimated by flexible ML methods, such as  the double/debiased machine learning framework for treatment and structural parameters \citep{chernozhukov2018double}. In these papers, machine learning is used to estimate nuisance
% functions (propensities, regressions, Riesz representers), while orthogonal
% moment conditions and sample splitting deliver regular, often $\sqrt{n}$-rate inference for the parameter of interest.

% Our contribution is to show that a similar semiparametric perspective can be brought to deep neural networks in discrete choice and MISC models. By
% designing an RMS layer that separates a low-dimensional direction $\theta$ from a high-dimensional first-stage component $h_0$, we obtain a network architecture in which $\theta$ has a clear economic interpretation (as the index in a binary or multinomial choice model, or the common monotone direction under MISC), and our theoretical results provide a full asymptotic characterization of the RMS estimator $\hat\theta$ even when $h_0$ is estimated by deep networks. In particular, under the conditions in Sections~\ref{sec:BinChoice} and
% \ref{sec:MISC}, we show that $\hat\theta$ converges at the intermediate $n^{-s/(2s+1)}$ rate with an asymptotically normal limit, %and bootstrap validity, 
% despite the presence of complex, high-dimensional nuisance parameters.

\subsection{Implementation using Machine Learning Packages}
\label{subsec:NN-implementation}

The network architectures described above are straightforward to implement in
standard machine learning frameworks such as \texttt{PyTorch} or
\texttt{TensorFlow}. The main ingredients are:
\begin{itemize}
\item a base MLP $f_\beta$ with ReLU activation (possibly deep),
\item a directional parameter $\theta$ constrained to lie on the unit sphere,
implemented via explicit normalization or a reparameterization, and
\item a custom ``RMS layer'' that takes $(x,f_\beta(x),\theta)$ as input and
outputs $g_{+}(x;\theta,\beta)$ and $g_{-}(x;\theta,\beta)$.
\end{itemize}
Since all components are compositions of affine maps and ReLU activations,
the network is differentiable almost everywhere and compatible with automatic
differentiation. Training can therefore be carried out using standard
gradient-based optimizers (e.g.\ ADAM) with GPU acceleration.

There are two natural training strategies:
\begin{enumerate}
\item \emph{Two-step RMS:} First estimate $h_0$ by training $f_\beta$ to
minimize a standard loss (e.g.\ squared error between $Y_i$ and
$f_\beta(X_i)$). Then plug in $\hat h(x)=f_{\hat\beta}(x)$ and optimize
$\hat Q(\theta)$ over $\theta$ only, using the RMS layer as in
Sections~\ref{sec:BinChoice} and~\ref{sec:MISC}.
\item \emph{Joint DNN:} Parameterize the outcome as
$Y_i\approx h_{\theta,\beta}(X_i)$ via the RMS layer and estimate both
$\theta$ and $\beta$ jointly by minimizing a loss such as
$\frac{1}{n}\sum_i (Y_i-h_{\theta,\beta}(X_i))^2$ subject to $\|\theta\|=1$.
This corresponds to embedding the MISC structure directly into a deep network
and training it with standard backpropagation.
\end{enumerate}
The two-step approach falls directly under our existing asymptotic theory, once
$\hat h$ is shown to satisfy the first-stage conditions. The joint-DNN approach is more demanding theoretically but conceptually attractive, as it treats $\theta$ as a low-dimensional ``interpretable head'' on top of a deep, flexible feature extractor.

Formally establishing the asymptotic properties of $\hat{\theta}$ in the joint DNN estimation approach is an interesting direction for future research. One natural route would be to show that, under suitable conditions on the loss, architecture and regularization, the joint estimator of $\theta$ is asymptotically equivalent to the two-step/profile RMS estimator studied here, given that the MISC parameter $\t$ only shows up in the ``outmost'' hidden layer of the DNN. An alternative route would be to use sample-splitting or cross-fitting to obtain valid inference for $\theta$ directly from the joint optimization problem.

\section{Simulation}
\label{sec:Sim}

Our goal in this section is to investigate the finite-sample performance of the
RMS estimator $\hat\theta$ for $\theta_0$ in both the single-index binary
choice model and the two-index MISC setting. We first describe the common
simulation design and implementation, and then discuss an alternative neural
network implementation that embeds the MISC structure directly into the network
architecture.

\subsection{Simulation Design and Implementation}
\label{subsec:SimDesign}

Each Monte Carlo experiment follows the same basic four-step procedure:
\begin{enumerate}
\item Generate a random sample from a given data-generating process (DGP).
\item Obtain an estimate $\hat{\t}$ either using a two-step plug-in procedure or the joint DNN procedure.
\item Evaluate the performance of $\hat\theta$ across $B$ Monte Carlo
replications.
\end{enumerate}

\subsubsection{DGP Specification}
\paragraph{Single-index DGP.}
In the baseline design we consider the binary choice model
\[
y_i = 1\{X_i'\theta_0 > \varepsilon_i\},
\]
with
\[
\theta_0
= \Bigl(\tfrac{\sqrt{3}}{3}, -\tfrac{\sqrt{3}}{3}, \tfrac{\sqrt{3}}{3}\Bigr)',
\qquad \| \theta_0 \| = 1.
\]
The regressors are drawn independently as
$X_{i1}, X_{i2}, X_{i3} \sim \mathrm{Unif}[-2,2]$, and the error terms
$\varepsilon_i$ are i.i.d.\ logistic. Denoting by $F$ the logistic CDF, the
true first-stage function is
\[
h_0(x)
:= E\bigl[y_i - \tfrac{1}{2} \mid X_i=x\bigr]
= F(x'\theta_0) - \tfrac{1}{2}
= \frac{1}{1 + \exp(-x'\theta_0)} - \tfrac{1}{2},
\]
which is known in closed form but treated as unknown in the estimation
procedure.

\paragraph{Two-index (MISC) DGP.}
To illustrate the multi-index setting, we also consider a two-index model
($J=2$) that satisfies the MISC condition. For each $i$, we generate
\[
y_i
= 1\{X_{i1}'\theta_0 > \varepsilon_{i1}\}\,
  1\{X_{i2}'\theta_0 > \varepsilon_{i2}\},
\]
where $\varepsilon_{i1},\varepsilon_{i2}$ are i.i.d.\ logistic and each
component of $X_{i1}$ and $X_{i2}$ is i.i.d.\ $\mathrm{Unif}[-2,2]$. Writing
$X_i = (X_{i1},X_{i2})$ and using the same $\theta_0$ as above, we have
\[
P(y_i=1 \mid X_i=(x_1,x_2))
= F(x_1'\theta_0)F(x_2'\theta_0),
\]
so that
\[
h_0(x_1,x_2)
:= E\bigl[y_i - \tfrac{1}{4} \mid X_{i1}=x_1,X_{i2}=x_2\bigr]
= F(x_1'\theta_0)F(x_2'\theta_0) - \tfrac{1}{4}.
\]
This DGP satisfies the strict MISC condition: $h_0(x_1,x_2)>0$ whenever both
$x_1'\theta_0$ and $x_2'\theta_0$ are positive, and $h_0(x_1,x_2)<0$ whenever
both are negative.

\subsubsection{Two-Stage Implementation}

\paragraph{First-Stage Nonparametric Regression}
Given simulated data, we estimate $h_0$ nonparametrically by regressing
$y_i - \tfrac{1}{2}$ on $X_i$ in the single-index design, and
$y_i - \tfrac{1}{4}$ on $(X_{i1},X_{i2})$ in the two-index design. We consider
two main classes of estimators (implemented using standard R packages):
\begin{itemize}
\item \emph{Kernel regression}, with a polynomial kernel and bandwidth
selected over a small grid (e.g.\ using a rule of thumb or simple
cross-validation). In the reported simulations we use a polynomial kernel with
tuning parameters $\alpha = 0.1$ and $\gamma = 0.0001$.
\item \emph{Series (sieve) regression}, based on tensor-product spline bases,
with the number of basis functions playing the role of the smoothing parameter.
\item \emph{Neural Network regression}: a standard multi-layer perceptron (MLP) with
ReLU activation, where the main tuning parameters are the number of hidden
units and layers. In the experiments reported below, a typical configuration
uses a hidden size of 10, 2 hidden layers, a learning rate of 0.01, and 100
epochs of training with the ADAM optimizer.
\end{itemize}

\paragraph{Second-stage optimization of the RMS criterion.}
Given $\hat h$, we form the sample analogue of the RMS criterion,
\[
\hat Q(\theta)
:= \frac{1}{n}\sum_{i=1}^n
\Bigl\{ g_{+,\theta,\hat h}(X_i) + g_{-,\theta,\hat h}(X_i) \Bigr\},
\]
and maximize $\hat Q(\theta)$ over
$\theta$ on the unit sphere $\{\theta:\theta'\theta=1\}$. We use a
gradient-based algorithm (ADAM) together with a simple projection step to
enforce the unit-norm constraint. In practice, this amounts to running ADAM
updates on the unconstrained parameter vector and renormalizing $\theta$ to
unit length after each update. The learning rate is set to 0.01 and we run 500
epochs for each replication. The use of ReLU functions makes the objective
continuous and Lipschitz in $\theta$, so gradients are well defined almost
everywhere and standard optimization routines are stable in these simulations.

\subsubsection{Joint Implementation via Neural Networks}

We also consider the DNN-based joint estimation of $h_0$ and $\t_0$ as described in Section \ref{sec:NN}. Specifically, we use a three-stage training strategy:
\begin{itemize}
    \item \textbf{Stage 1}: Freeze $\theta$ parameters (initialized to zero vectors), and train only the MLP component parameters to learn basic function approximation.
    \item \textbf{Stage 2}: Freeze the MLP component, reinitialize and train only the directional parameter $\theta$.
    \item \textbf{Stage 3}: Jointly train all parameters for fine-tuning.
\end{itemize}

\subsubsection{Performance Measures}
For each design, we consider two sample sizes $N \in \{1000,5000\}$, and re report summary measures of the distribution of $\hat\theta$ across $B=1000$ Monte Carlo replications. The basic componentwise diagnostics are the Monte Carlo mean squared error (MSE), bias, and standard deviation (SD) of each
coordinate of $\hat\theta$. To capture overall performance in a
rotation-invariant way, we also report: the $\ell^1$ error of each coordinate, the $\ell^2$ norm of the bias vector, and mean/median ``angular similarity'', defined as one minus the cosine of the angle between $\hat\theta$ and $\theta_0$.

\subsection{Results}

\subsubsection{Single-Index DGP}

For the single-index design, Tables 1–3 report the performance of the RMS
estimator with three different first-stage implementations: kernel regression
(Table 1), a separate neural network nonparametric estimator (Table 2), and an
“all-in-one’’ neural network that jointly estimates the first stage and
$\theta$ (Table 3). In all cases, increasing the sample size from $N=1000$ to
$N=5000$ substantially reduces MSE, standard deviations, and angular errors:
$1-\text{mean angular similarity}$ falls from roughly $6\times 10^{-3}$ to
$4\times 10^{-3}$ for the kernel, and from about $1.0\times 10^{-2}$ to
$3\text{–}4\times 10^{-3}$ for the neural network implementations. The kernel
first stage is slightly more accurate than the neural network alternatives at
$N=1000$, but by $N=5000$ all three approaches deliver very similar accuracy,
with small biases and tight angular concentration around $\theta_0$.

% For kernel method, we choose polynomial kernel, set $\alpha = 0.1, \gamma = 0.0001$. Learning rate for ADAM is set to $0.01$, run 500 epochs.

\begin{table}
\centering
\caption{Two-Stage RMS with Kernel First Stage}
\begin{tabular}{lcc}
\toprule
\textbf{Metric} & \textbf{N=1000} & \textbf{N=5000} \\
\midrule
MSE of $\theta_1$ & 0.00375 & 0.00249 \\
MSE of $\theta_2$ & 0.00405 & 0.00228 \\
MSE of $\theta_3$ & 0.00395 & 0.00257 \\
Bias of $\theta_1$ & -0.00279 & -0.00157 \\
Bias of $\theta_2$ & 0.00380 & 0.00154 \\
Bias of $\theta_3$ & -0.00358 & -0.00325 \\
SD of $\theta_1$ & 0.06114 & 0.04986 \\
SD of $\theta_2$ & 0.06356 & 0.04771 \\
SD of $\theta_3$ & 0.06274 & 0.05058 \\
L1 Error of $\theta_1$ & 0.04643 & 0.03657 \\
L1 Error of $\theta_2$ & 0.04886 & 0.03562 \\
L1 Error of $\theta_3$ & 0.04620 & 0.03615 \\
\midrule
L2 Norm of Bias & 0.005923 & 0.003920 \\
1-- Mean Angular Similarity & 0.005874 & 0.003668 \\
1-- Median Angular Similarity & 0.003250 & 0.001833 \\
\bottomrule
\end{tabular}
\end{table}

% For nerual network to do nonparametric estimation method, set $hiddensize = 10, layers = 2$, run 100 epochs with $lr = 0.01$.

\begin{table}
\centering
\caption{Two-Stage RMS with Neural-Net First Stage}
\begin{tabular}{lcc}
\toprule
\textbf{Metric} & \textbf{N=1000} & \textbf{N=5000} \\
\midrule
MSE of $\theta_1$ & 0.00703 & 0.00260 \\
MSE of $\theta_2$ & 0.00728 & 0.00241 \\
MSE of $\theta_3$ & 0.00678 & 0.00242 \\
Bias of $\theta_1$ & -0.00410 & -0.00328 \\
Bias of $\theta_2$ & 0.00640 & 0.00312 \\
Bias of $\theta_3$ & -0.00776 & -0.00004 \\
SD of $\theta_1$ & 0.08374 & 0.05085 \\
SD of $\theta_2$ & 0.08508 & 0.04903 \\
SD of $\theta_3$ & 0.08195 & 0.04924 \\
L1 Error of $\theta_1$ & 0.06505 & 0.03846 \\
L1 Error of $\theta_2$ & 0.06705 & 0.03713 \\
L1 Error of $\theta_3$ & 0.06458 & 0.03759 \\
\midrule
L2 Norm of Bias & 0.010862 & 0.004529 \\
1-- Mean Angular Similarity & 0.010542 & 0.003717 \\
1-- Median Angular Similarity & 0.006961 & 0.002064 \\
\bottomrule
\end{tabular}
\end{table}

\begin{table}
\centering
\caption{Joint DNN-Based Estimation}
\begin{tabular}{lcc}
\toprule
\textbf{Metric} & \textbf{N=1000} & \textbf{N=5000} \\
\midrule
MSE of $\theta_1$ & 0.01047 & 0.00275 \\
MSE of $\theta_2$ & 0.00997 & 0.00260 \\
MSE of $\theta_3$ & 0.00990 & 0.00285 \\
Bias of $\theta_1$ & -0.00856 & -0.00174 \\
Bias of $\theta_2$ &  0.01187 &  0.00184 \\
Bias of $\theta_3$ & -0.00584 & -0.00352 \\
SD of $\theta_1$   & 0.10198 & 0.05238 \\
SD of $\theta_2$   & 0.09916 & 0.05097 \\
SD of $\theta_3$   & 0.09933 & 0.05327 \\
L1 Error of $\theta_1$ & 0.07945 & 0.04148 \\
L1 Error of $\theta_2$ & 0.07802 & 0.04006 \\
L1 Error of $\theta_3$ & 0.07886 & 0.04231 \\
\midrule
L2 Norm of Bias & 0.015764 & 0.004335 \\
1-- Mean Angular Similarity & 0.015174 & 0.004099 \\
1-- Median Angular Similarity & 0.009594 & 0.002806 \\
\bottomrule
\end{tabular}
\end{table}

\subsubsection{Two-Index Design}

For the two-index MISC design, Tables 4–6 show the same three
implementations. The problem is clearly harder: MSEs and angular errors are larger than in the single-index case, though they still improve remarkably with sample size. Here the choice of first-stage method matters more. The kernel version (Table 4) achieves reasonable performance, but the two-step neural network first stage (Table 5) delivers substantially smaller MSE and angular error, especially at $N=5000$ (where MSEs drop from about $10^{-2}$ to roughly
$3\times 10^{-3}$, and $1-\text{mean angular similarity}$ from about $1.5\times 10^{-2}$ to around $4.6\times 10^{-3}$). The all-in-one neural network (Table 6) performs similarly to the kernel in this two-index setting and does\ not match the accuracy of the two-step neural network. Overall, the tables confirm that (i) the RMS estimator behaves in line with the theory as $N$ grows, (ii) the two-step architecture is robust and competitive in the single-index case, and (iii) in more complex multi-index designs, flexible neural network first stages can yield clear gains over standard kernel
smoothing.

\begin{table}
\centering
\caption{Two-Stage RMS with Kernel First Stage: $J=2$}
\begin{tabular}{lcc}
\toprule
\textbf{Metric} & \textbf{N=1000} & \textbf{N=5000} \\
\midrule
MSE of $\theta_1$ & 0.02901 & 0.01047 \\
MSE of $\theta_2$ & 0.02810 & 0.01003 \\
MSE of $\theta_3$ & 0.02805 & 0.00995 \\
Bias of $\theta_1$ & -0.02186 & -0.01029 \\
Bias of $\theta_2$ &  0.02827 &  0.00720 \\
Bias of $\theta_3$ & -0.02361 & -0.00887 \\
SD of $\theta_1$ & 0.16890 & 0.10178 \\
SD of $\theta_2$ & 0.16522 & 0.09989 \\
SD of $\theta_3$ & 0.16580 & 0.09937 \\
L1 Error of $\theta_1$ & 0.12860 & 0.07589 \\
L1 Error of $\theta_2$ & 0.12873 & 0.07545 \\
L1 Error of $\theta_3$ & 0.12967 & 0.07591 \\
\midrule
L2 Norm of Bias & 0.042832 & 0.015379 \\
1-- Mean Angular Similarity & 0.042575 & 0.015223 \\
1-- Median Angular Similarity & 0.029377 & 0.008477 \\
\bottomrule
\end{tabular}
\end{table}

\begin{table}
\centering
\caption{Two-Stage RMS with Neural-Net First Stage: $J=2$}
\begin{tabular}{lcc}
\toprule
\textbf{Metric} & \textbf{N=1000} & \textbf{N=5000} \\
\midrule
MSE of $\theta_1$ & 0.01890 & 0.00303 \\
MSE of $\theta_2$ & 0.02165 & 0.00314 \\
MSE of $\theta_3$ & 0.01637 & 0.00302 \\
Bias of $\theta_1$ & -0.01421 & -0.00304 \\
Bias of $\theta_2$ & 0.02280 & 0.00113 \\
Bias of $\theta_3$ & -0.01229 & -0.00380 \\
SD of $\theta_1$ & 0.13674 & 0.05498 \\
SD of $\theta_2$ & 0.14536 & 0.05606 \\
SD of $\theta_3$ & 0.12736 & 0.05485 \\
L1 Error of $\theta_1$ & 0.09536 & 0.04228 \\
L1 Error of $\theta_2$ & 0.10210 & 0.04332 \\
L1 Error of $\theta_3$ & 0.09336 & 0.04255 \\
\midrule
L2 Norm of Bias & 0.029541 & 0.004994 \\
1-- Mean Angular Similarity & 0.028461 & 0.004600 \\
1-- Median Angular Similarity & 0.013980 & 0.002799 \\
\bottomrule
\end{tabular}
\end{table}

\begin{table}
\centering
\caption{Joint DNN-Based Estimation: $J=2$}
\begin{tabular}{lcc}
\toprule
\textbf{Metric} & \textbf{N=1000} & \textbf{N=5000} \\
\midrule
MSE of $\theta_1$ & 0.02751 & 0.01160 \\
MSE of $\theta_2$ & 0.02881 & 0.01188 \\
MSE of $\theta_3$ & 0.02689 & 0.01183 \\
Bias of $\theta_1$ & -0.02517 & -0.00710 \\
Bias of $\theta_2$ & 0.01779  &  0.01367 \\
Bias of $\theta_3$ & -0.02910 & -0.00981 \\
SD of $\theta_1$ & 0.16394 & 0.10745 \\
SD of $\theta_2$ & 0.16881 & 0.10816 \\
SD of $\theta_3$ & 0.16137 & 0.10834 \\
L1 Error of $\theta_1$ & 0.12984 & 0.08602 \\
L1 Error of $\theta_2$ & 0.13342 & 0.08658 \\
L1 Error of $\theta_3$ & 0.13054 & 0.08681 \\
\midrule
L2 Norm of Bias & 0.042388 & 0.018265 \\
1-- Mean Angular Similarity & 0.041604 & 0.017658 \\
1-- Median Angular Similarity & 0.026760 & 0.012385 \\
\bottomrule
\end{tabular}
\end{table}

\section{\label{sec:Con}Conclusion}

We have proposed a rectified-linear-unit-based maximum score (RMS) estimator for models characterized by sign-alignment restrictions.  By replacing the
discontinuous indicator in Manski’s maximum score with composite ReLU functions,
the population criterion becomes piecewise smooth with quadratic curvature,
while preserving the underlying identification logic.  This structure delivers
an intermediate, ``one-dimensional'' rate $n^{-s/(2s+1)}$ and asymptotic
normality, but also yields a sample objective that is much more amenable to
modern gradient-based optimization methods.  In practice, RMS can be optimized
using off-the-shelf routines from machine learning, avoiding the fragile,
combinatorial searches often required for discontinuous maximum score criteria.

We also embed the binary choice model in a general multi-index single-crossing
(MISC) framework, where several indexes enter through a common direction
parameter.  Even in this multi-index setting, the leading term in the
asymptotic expansion depends on the nonparametric component only through
its restriction to a finite union of $(d-1)$-dimensional hyperplanes, so the
effective nonparametric dimension remains one and the convergence rate is
unchanged.  Taken together, these results show that ReLU-based formulations can
retain the robustness and partial identification features of maximum score,
while offering significant computational advantages and extending naturally to
richer multi-index environments.

\bibliographystyle{ecca}
\bibliography{ReMS}

\begin{thebibliography}{36}
\providecommand{\natexlab}[1]{#1}

\bibitem[{Abrevaya and Huang(2005)}]{abrevaya2005bootstrap}
\textsc{Abrevaya, J.} and \textsc{Huang, J.} (2005). On the bootstrap of the
  maximum score estimator. \textit{Econometrica}, \textbf{73}~(4), 1175--1204.

\bibitem[{Belloni \textit{et~al.}(2015)Belloni, Chernozhukov, Chetverikov and
  Kato}]{belloni2015some}
\textsc{Belloni, A.}, \textsc{Chernozhukov, V.}, \textsc{Chetverikov, D.} and
  \textsc{Kato, K.} (2015). Some new asymptotic theory for least squares
  series: Pointwise and uniform results. \textit{Journal of Econometrics},
  \textbf{186}~(2), 345--366.

\bibitem[{Blevins and Khan(2013)}]{blevins2013local}
\textsc{Blevins, J.~R.} and \textsc{Khan, S.} (2013). Local nlls estimation of
  semi-parametric binary choice models. \textit{The Econometrics Journal},
  \textbf{16}~(2), 135--160.

\bibitem[{Cattaneo \textit{et~al.}(2020)Cattaneo, Jansson and
  Nagasawa}]{cattaneo2020bootstrap}
\textsc{Cattaneo, M.~D.}, \textsc{Jansson, M.} and \textsc{Nagasawa, K.}
  (2020). Bootstrap-based inference for cube root asymptotics.
  \textit{Econometrica}, \textbf{88}~(5), 2203--2219.

\bibitem[{Chen and Zhang(2015)}]{chen2015binary}
\textsc{Chen, S.} and \textsc{Zhang, H.} (2015). Binary quantile regression
  with local polynomial smoothing. \textit{Journal of Econometrics},
  \textbf{189}~(1), 24--40.

\bibitem[{Chen(2007)}]{chen2007sieve}
\textsc{Chen, X.} (2007). Large sample sieve estimation of semi-nonparametric
  models. In \textit{Handbook of Econometrics}, vol.~6B, Elsevier B.V.

\bibitem[{Chen and Christensen(2015)}]{chen2015optimal}
\textsc{---} and \textsc{Christensen, T.~M.} (2015). Optimal uniform
  convergence rates and asymptotic normality for series estimators under weak
  dependence and weak conditions. \textit{Journal of Econometrics},
  \textbf{188}~(2), 447--465.

\bibitem[{Chen and Gao(2025)}]{chen2025semiparametric}
\textsc{---} and \textsc{Gao, W.~Y.} (2025). Semiparametric learning of
  integral functionals on submanifolds. \textit{arXiv preprint
  arXiv:2507.12673}.

\bibitem[{Chen \textit{et~al.}(2003)Chen, Linton and
  Van~Keilegom}]{chen2003estimation}
\textsc{---}, \textsc{Linton, O.} and \textsc{Van~Keilegom, I.} (2003).
  Estimation of semiparametric models when the criterion function is not
  smooth. \textit{Econometrica}, \textbf{71}~(5), 1591--1608.

\bibitem[{Dau \textit{et~al.}(2020)Dau, Lalo{\"e} and Servien}]{dau2020exact}
\textsc{Dau, H.~D.}, \textsc{Lalo{\"e}, T.} and \textsc{Servien, R.} (2020).
  Exact asymptotic limit for kernel estimation of regression level sets.
  \textit{Statistics \& Probability Letters}, \textbf{161}, 108721.

\bibitem[{Delsol and Van~Keilegom(2020)}]{delsol2020semiparametric}
\textsc{Delsol, L.} and \textsc{Van~Keilegom, I.} (2020). Semiparametric
  m-estimation with non-smooth criterion functions. \textit{Annals of the
  Institute of Statistical Mathematics}, \textbf{72}~(2), 577--605.

\bibitem[{Evans and Gariepy(2015)}]{evans2015measure}
\textsc{Evans, L.~C.} and \textsc{Gariepy, R.~F.} (2015). \textit{Measure
  Theory and Fine Properties of Functions}. CRC Press.

\bibitem[{Fan \textit{et~al.}(2021)Fan, Xiong, Li and
  Wang}]{fan2021interpretability}
\textsc{Fan, F.-L.}, \textsc{Xiong, J.}, \textsc{Li, M.} and \textsc{Wang, G.}
  (2021). On interpretability of artificial neural networks: A survey.
  \textit{IEEE Transactions on Radiation and Plasma Medical Sciences},
  \textbf{5}~(6), 741--760.

\bibitem[{Gao and Li(2024)}]{gao2020robust}
\textsc{Gao, W.~Y.} and \textsc{Li, M.} (2024). Identification of
  semiparametric panel multinomial choice models with infinite-dimensional
  fixed effects. \textit{arXiv preprint arXiv:2009.00085v2}.

\bibitem[{Gao \textit{et~al.}(2023)Gao, Li and Xu}]{gao2023logical}
\textsc{---}, \textsc{---} and \textsc{Xu, S.} (2023). Logical differencing in
  dyadic network formation models with nontransferable utilities.
  \textit{Journal of Econometrics}, \textbf{235}~(1), 302--324.

\bibitem[{Graham(2017)}]{graham2017econometric}
\textsc{Graham, B.~S.} (2017). An econometric model of network formation with
  degree heterogeneity. \textit{Econometrica}, \textbf{85}~(4), 1033--1063.

\bibitem[{Hansen(2008)}]{hansen2008uniform}
\textsc{Hansen, B.~E.} (2008). Uniform convergence rates for kernel estimation
  with dependent data. \textit{Econometric Theory}, pp. 726--748.

\bibitem[{Horowitz(1992)}]{horowitz1992smoothed}
\textsc{Horowitz, J.~L.} (1992). A smoothed maximum score estimator for the
  binary response model. \textit{Econometrica: journal of the Econometric
  Society}, pp. 505--531.

\bibitem[{Horowitz(2002)}]{horowitz2002bootstrap}
\textsc{---} (2002). Bootstrap critical values for tests based on the smoothed
  maximum score estimator. \textit{Journal of Econometrics}, \textbf{111}~(2),
  141--167.

\bibitem[{Ichimura and Lee(2010)}]{ichimura2010characterization}
\textsc{Ichimura, H.} and \textsc{Lee, S.} (2010). Characterization of the
  asymptotic distribution of semiparametric m-estimators. \textit{Journal of
  Econometrics}, \textbf{159}~(2), 252--266.

\bibitem[{Ichimura and Lee(2018)}]{ichimura2018corrigendum}
\textsc{---} and \textsc{---} (2018). Corrigendum to “characterization of the
  asymptotic distribution of semiparametric m-estimators”[j. econometrics 159
  (2)(2010) 252--266]. \textit{Journal of Econometrics}, \textbf{202}~(2),
  306--307.

\bibitem[{Ichimura and Todd(2007)}]{ichimura2007implementing}
\textsc{---} and \textsc{Todd, P.~E.} (2007). Implementing nonparametric and
  semiparametric estimators. \textit{Handbook of econometrics}, \textbf{6},
  5369--5468.

\bibitem[{Jun \textit{et~al.}(2017)Jun, Pinkse and Wan}]{jun2017integrated}
\textsc{Jun, S.~J.}, \textsc{Pinkse, J.} and \textsc{Wan, Y.} (2017).
  Integrated score estimation. \textit{Econometric Theory}, \textbf{33}~(6),
  1418--1456.

\bibitem[{Kim and Pollard(1990)}]{kim1990cube}
\textsc{Kim, J.} and \textsc{Pollard, D.} (1990). Cube root asymptotics.
  \textit{The Annals of Statistics}, pp. 191--219.

\bibitem[{Kosorok(2008)}]{kosorok2007introduction}
\textsc{Kosorok, M.~R.} (2008). \textit{Introduction to empirical processes and
  semiparametric inference}. Springer Science \& Business Media.

\bibitem[{Manski(1975)}]{manski1975maximum}
\textsc{Manski, C.~F.} (1975). Maximum score estimation of the stochastic
  utility model of choice. \textit{Journal of econometrics}, \textbf{3}~(3),
  205--228.

\bibitem[{Manski(1985)}]{manski1985semiparametric}
\textsc{---} (1985). Semiparametric analysis of discrete response: Asymptotic
  properties of the maximum score estimator. \textit{Journal of econometrics},
  \textbf{27}~(3), 313--333.

\bibitem[{Manski(1987)}]{manski1987semiparametric}
\textsc{---} (1987). Semiparametric analysis of random effects linear models
  from binary panel data. \textit{Econometrica}, \textbf{55}~(2), 357--362.

\bibitem[{Newey and McFadden(1994)}]{newey1994large}
\textsc{Newey, K.} and \textsc{McFadden, D.} (1994). Large sample estimation
  and hypothesis testing. \textit{Handbook of Econometrics, IV, Edited by RF
  Engle and DL McFadden}, pp. 2112--2245.

\bibitem[{Newey(1994{\natexlab{a}})}]{newey1994asymptotic}
\textsc{Newey, W.~K.} (1994{\natexlab{a}}). The asymptotic variance of
  semiparametric estimators. \textit{Econometrica: Journal of the Econometric
  Society}, pp. 1349--1382.

\bibitem[{Newey(1994{\natexlab{b}})}]{newey1994kernel}
\textsc{---} (1994{\natexlab{b}}). Kernel estimation of partial means and a
  general variance estimator. \textit{Econometric Theory}, \textbf{10}~(2),
  1--21.

\bibitem[{Patra \textit{et~al.}(2018)Patra, Seijo and
  Sen}]{patra2018consistent}
\textsc{Patra, R.~K.}, \textsc{Seijo, E.} and \textsc{Sen, B.} (2018). A
  consistent bootstrap procedure for the maximum score estimator.
  \textit{Journal of Econometrics}.

\bibitem[{Qiao(2021)}]{qiao2021nonparametric}
\textsc{Qiao, W.} (2021). Nonparametric estimation of surface integrals on
  level sets. \textit{Bernoulli}, \textbf{27}~(1), 155--191.

\bibitem[{Seo and Otsu(2018)}]{seo2018local}
\textsc{Seo, M.~H.} and \textsc{Otsu, T.} (2018). Local m-estimation with
  discontinuous criterion for dependent and limited observations. \textit{The
  Annals of Statistics}, \textbf{46}~(1), 344--369.

\bibitem[{Van Der~Vaart and Wellner(1996)}]{van1996weak}
\textsc{Van Der~Vaart, A.~W.} and \textsc{Wellner, J.~A.} (1996). \textit{Weak
  Convergence and Empirical Processes}. Springer.

\bibitem[{Zhang \textit{et~al.}(2021)Zhang, Ti{\v{n}}o, Leonardis and
  Tang}]{zhang2021survey}
\textsc{Zhang, Y.}, \textsc{Ti{\v{n}}o, P.}, \textsc{Leonardis, A.} and
  \textsc{Tang, K.} (2021). A survey on neural network interpretability.
  \textit{IEEE transactions on emerging topics in computational intelligence},
  \textbf{5}~(5), 726--742.

\end{thebibliography}

\appendix

\section{\label{sec:MainProofs}Main Proofs}

\subsection{\label{subsec:Pf_Term1}Proof of Lemma \ref{lem:Term1}}
\begin{proof}
\noindent Recall that $g_{\t,h}=g_{+,\t,h}+g_{-,\t,h}$ and
\[
g_{+,\t,h}\left(x\right)=\left[h\left(x\right)-\left[-x^{'}\t\right]_{+}\right]_{+},\quad g_{-,\t,h}\left(x\right)=\left[-h\left(x\right)-\left[x^{'}\t\right]_{+}\right]_{+}.
\]
For any $x,\t$ and $h$, observe first that $\left|g_{+,\t,h}-g_{+,\t_{0},h}\right|\leq\left[h\right]_{+}$
and, by the Lipschitz continuity of the ReLU function,
\begin{align*}
\left|g_{+,\t,h}\left(x\right)-g_{+,\t_{0},h}\left(x\right)\right|\leq & \left|h\left(x\right)-\left[-x^{'}\t\right]_{+}-\left(h\left(x\right)-\left[-x^{'}\t_{0}\right]_{+}\right)\right|\\
= & \left|\left[-x^{'}\t\right]_{+}-\left[-x^{'}\t_{0}\right]_{+}\right|\leq\left|x^{'}\left(\t-\t_{0}\right)\right|,
\end{align*}
or, in summary, 
\begin{equation}
\left|g_{+,\t,h}\left(x\right)-g_{+,\t_{0},h}\left(x\right)\right|\leq\min\left(\left[h\left(x\right)\right]_{+},\left|x^{'}\left(\t-\t_{0}\right)\right|\right).\label{eq:ReLU_ub}
\end{equation}

With $h=h_{0}$, we have $g_{+,\t_{0},h_{0}}=\left[h\right]_{+}$
and thus
\[
\left|g_{+,\t,h_{0}}-g_{+,\t_{0},h_{0}}\right|=\left[h_{0}\left(x\right)\right]_{+}-\left[h_{0}\left(x\right)-\left[-x^{'}\t\right]_{+}\right]_{+},
\]
which is nonzero only if $h_{0}\left(x\right)>0$ and $x^{'}\t<0$,
which, by the sign alignment restriction \eqref{eq:Mono_Equiv}, is
equivalent to the event $x^{'}\t<0<x^{'}\t_{0}$. Combing this with
\eqref{eq:ReLU_ub}, we have
\begin{align*}
\left|g_{+,\t,h_{0}}\left(x\right)-g_{+,\t_{0},h_{0}}\left(x\right)\right| & \leq\ind\left\{ x^{'}\t<0<x^{'}\t_{0}\right\} \left|x^{'}\left(\t-\t_{0}\right)\right|\\
 & =\ind\left\{ x^{'}\t_{0}+x^{'}\left(\t-\t_{0}\right)<0<x^{'}\t_{0}\right\} \left|x^{'}\left(\t-\t_{0}\right)\right|\\
 & \leq\ind\left\{ x^{'}\t_{0}-\norm x\norm{\t-\t_{0}}<0<x^{'}\t_{0}\right\} \norm x\norm{\t-\t_{0}}\\
 & \leq\ind\left\{ 0<x^{'}\t_{0}<\norm x\norm{\t-\t_{0}}\right\} \norm x\norm{\t-\t_{0}}
\end{align*}

\noindent Similarly, the arguments above can be adapted for $g_{-}$:
\[
\left|g_{-,\t,h_{0}}\left(x\right)-g_{-,\t_{0},h_{0}}\left(x\right)\right|\leq\ind\left\{ -\norm x\norm{\t-\t_{0}}<x^{'}\t_{0}<0\right\} \norm x\norm{\t-\t_{0}}.
\]
Together, we have
\[
\left|g_{\t,h_{0}}\left(x\right)-g_{\t_{0},h_{0}}\left(x\right)\right|\leq\ind\left\{ \left|\frac{x^{'}}{\norm x}\t_{0}\right|<\norm{\t-\t_{0}}\right\} \norm x\norm{\t-\t_{0}}.
\]

Define ${\cal \cG}_{1,\d}:=\left\{ g_{\t,h_{0}}-g_{\t_{0},h_{0}}:\ \norm{\t-\t_{0}}\leq\d\right\} .$
By the arguments above, ${\cal \cG}_{1,\d}$ has an envelope $G_{1,\d}$
given by
\[
G_{1,\d}:=\ind\left\{ \left|\frac{x^{'}}{\norm x}\t_{0}\right|<\d\right\} \norm x\d
\]
with
\begin{align*}
PG_{1,\d}^{2} & =\E\left[\ind\left\{ \left|\frac{X_{i}^{'}}{\norm{X_{i}}}\t_{0}\right|<\d\right\} \norm{X_{i}}^{2}\d^{2}\right]\leq\d^{2}\P\left(\left|\frac{X_{i}^{'}}{\norm{X_{i}}}\t_{0}\right|\leq\d\right)\leq C\d^{3}.
\end{align*}
Now, since ${\cal \cG}_{1,\d}\subseteq\cG$, we have $\mathscr{N}\left(\e,\cG_{1,\d},L_{2}\left(P\right)\right)\leq\mathscr{N}\left(\e,\cG,L_{2}\left(P\right)\right)$
\[
J_{1,\d}:=\int_{0}^{1}\sqrt{1+\log\mathscr{N}\left(\e,\cG_{1,},L_{2}\left(P\right)\right)}d\e\leq J<\infty.
\]
Then, by VW Theorem 2.14.1, we have
\[
P\sup_{g\in{\cal \cG}_{1,\d}}\left|\GG_{n}\left(g\right)\right|\leq J_{1,\d}\sqrt{PG_{1,\d}^{2}}\leq J_{1}C\d^{\frac{3}{2}}=C_{1}\d^{\frac{3}{2}}.
\]
\end{proof}

\subsection{\label{subsec:Pf_Term2}Proof of Lemma \ref{lem:Term2}}
\begin{proof}
Observe first that, by the construction of $g_{+}$, we have
\begin{align}
\left|g_{+,\t,h}-g_{+,\t_{0},h}-g_{+,\t,h_{0}}+g_{+,\t_{0},h_{0}}\right| & \leq2\left|x^{'}\left(\t-\t_{0}\right)\right|\leq2\norm x\norm{\t-\t_{0}}\label{eq:DiD_2xdt}
\end{align}
Define ${\cal \cG}_{2,\d}:=\left\{ g_{\t,h}-g_{\t_{0},h}-g_{\t,h_{0}}+g_{\t_{0},h_{0}}:\ \norm{\t-\t_{0}}\leq\d,h\in{\cal H}\right\} .$
By the arguments above, ${\cal \cG}_{2,\d}$ has an envelope $G_{2,\d}$
given by $G_{2,\d}:=M\d$ with 
\[
PG_{2,n,\d}^{2}=M^{2}\d^{2}.
\]
By VW Theorem 2.14.1, we have
\[
P\sup_{g\in{\cal \cG}_{2,\d}}\norm{\GG_{n}\left(g\right)}\leq J_{2,\d}\sqrt{PG_{2,\d}^{2}}\leq M\d.
\]
\end{proof}

\subsection{\label{subsec:Pf_Term3}Proof of Lemma \ref{lem:Term3}}
\begin{proof}
Noting that $g_{+},g_{-}$ are all Lipschitz continuous, 
\begin{align*}
\Dif_{\t}g_{+,\t,h}\left(x\right) & :=\Dif_{\t}\left[h\left(x\right)-\left[-x^{'}\t\right]_{+}\right]_{+}=\ind\left\{ h\left(x\right)>-x^{'}\t>0\right\} x\\
\Dif_{\t}g_{-,\t,h}\left(x\right) & :=-\ind\left\{ h\left(x\right)<-x^{'}\t<0\right\} x\\
\Dif_{\t}g_{\t,h}\left(x\right) & :=\Dif_{\t}g_{+,\t,h}\left(x\right)+\Dif_{\t}g_{-,\t,h}\left(x\right)
\end{align*}
are well-defined almost everywhere, and furthermore we have
\begin{align*}
\Dif_{\t}Pg_{+,\t,h_{0}} & =P\Dif_{\t}g_{+,\t,h_{0}}=\int\ind\left\{ h_{0}\left(x\right)>-x^{'}\t>0\right\} xdP\left(x\right)\\
\Dif_{\t}Pg_{-,\t,h_{0}} & =P\Dif_{\t}g_{-,\t,h_{0}}=\int-\ind\left\{ h_{0}\left(x\right)<-x^{'}\t<0\right\} xdP\left(x\right)
\end{align*}
Note that, at $\t=\t_{0}$, we have
\begin{align*}
\Dif_{\t}Pg_{+,\t_{0},h_{0}} & =\int\ind\left\{ h_{0}\left(x\right)>-x^{'}\t_{0}>0\right\} xdP\left(x\right)={\bf 0},\\
\Dif_{\t}Pg_{-,\t_{0},h_{0}} & =-\int\ind\left\{ h_{0}\left(x\right)<-x^{'}\t_{0}<0\right\} xdP\left(x\right)={\bf 0},\\
\Dif_{\t}g_{\t_{0},h_{0}}\left(x\right) & =\Dif_{\t}g_{+,\t_{0},h_{0}}\left(x\right)+\Dif_{\t}g_{-,\t_{0},h_{0}}\left(x\right)={\bf 0}.
\end{align*}
Recall that
\begin{align*}
\Dif_{\t}Pg_{+,\t,h}\left(x\right) & =\int\ind\left\{ h\left(x\right)>-x^{'}\t>0\right\} xdP\left(x\right)\\
 & =\left(\int_{x^{'}\t<0}-\int_{x^{'}\t<-h\left(x\right)}\right)\ind\left\{ h\left(x\right)>0\right\} xp\left(x\right)dx
\end{align*}
while 
\begin{align*}
\Dif_{\t}Pg_{-,\t,h}\left(x\right) & =-\int\ind\left\{ -h\left(x\right)>x^{'}\t>0\right\} xdP\left(x\right)\\
 & =-\left(\int_{x^{'}\t<-h\left(x\right)}-\int_{x^{'}\t<0}\right)\ind\left\{ h\left(x\right)<0\right\} xp\left(x\right)dx\\
 & =\left(\int_{x^{'}\t<0}-\int_{x^{'}\t<-h\left(x\right)}\right)\ind\left\{ h\left(x\right)<0\right\} xp\left(x\right)dx
\end{align*}
Hence, 
\begin{align*}
\Dif_{\t}Pg_{\t,h}\left(x\right) & =\Dif_{\t}Pg_{+,\t,h}\left(x\right)+\Dif_{\t}Pg_{-,\t,h}\left(x\right)=\left[\int_{x^{'}\t<0}-\int_{x^{'}\t<-h\left(x\right)}\right]xp\left(x\right)dx
\end{align*}
Since $\Dif_{x}\left(x^{'}\t\right)=\t$ and $\Dif_{x}\left(h\left(x\right)+x^{'}\t\right)=\Dif_{x}h\left(x\right)+\t$,
we have
\begin{align*}
\Dif_{\t\t}Pg_{\t,h}\left(x\right) & =\int_{x^{'}\t=0}\frac{1}{\norm{\t}}xx^{'}p\left(x\right)d{\cal H}^{d-1}\left(x\right)-\int_{x^{'}\t=-h\left(x\right)}\frac{1}{\norm{\Dif_{x}h\left(x\right)+\t}}xx^{'}p\left(x\right)d{\cal H}^{d-1}\left(x\right)\\
 & =\int_{x^{'}\t=0}xx^{'}p\left(x\right)d{\cal H}^{d-1}\left(x\right)-\int_{x^{'}\t=-h\left(x\right)}\frac{1}{\norm{\Dif_{x}h\left(x\right)+\t}}xx^{'}p\left(x\right)d{\cal H}^{d-1}\left(x\right)
\end{align*}

Recall that $h_{0}\left(x\right)=F\left(\rest{x^{'}\t_{0}}x\right)$
with $F\left(\rest 0x\right)\equiv0$. Hence, 
\[
\Dif_{x}h_{0}\left(x\right)=f\left(\rest{x^{'}\t_{0}}x\right)\t_{0}+F_{x}\left(\rest{x^{'}\t_{0}}x\right)
\]
with
\[
F_{x}\left(\rest 0x\right)={\bf 0}.
\]
Hence, evaluating $\Dif_{\t\t}Pg_{\t,h}\left(x\right)$ at $\left(\t_{0},h_{0}\right)$,
we have
\begin{align*}
 & \Dif_{\t\t}Pg_{\t_{0},h_{0}}\left(x\right)\\
= & \int_{x^{'}\t_{0}=0}xx^{'}p\left(x\right)d{\cal H}^{d-1}\left(x\right)-\int_{x^{'}\t_{0}=-h_{0}\left(x\right)}\frac{1}{\norm{\Dif_{x}h_{0}\left(x\right)+\t_{0}}}xx^{'}p\left(x\right)d{\cal H}^{d-1}\left(x\right)\\
= & \int_{x^{'}\t_{0}=0}xx^{'}p\left(x\right)d{\cal H}^{d-1}\left(x\right)-\int_{x^{'}\t_{0}=0}\frac{1}{\norm{f\left(\rest 0x\right)\t_{0}+F_{x}\left(\rest 0x\right)+\t_{0}}}xx^{'}p\left(x\right)d{\cal H}^{d-1}\left(x\right)\\
= & \int_{x^{'}\t_{0}=0}xx^{'}p\left(x\right)d{\cal H}^{d-1}\left(x\right)-\int_{x^{'}\t_{0}=0}\frac{1}{f\left(\rest 0x\right)+1}xx^{'}p\left(x\right)d{\cal H}^{d-1}\left(x\right)\\
= & \int_{x^{'}\t_{0}=0}\frac{f\left(\rest 0x\right)}{f\left(\rest 0x\right)+1}xx^{'}p\left(x\right)d{\cal H}^{d-1}\left(x\right)=V
\end{align*}
\end{proof}
Note that $\text{rank}\left(V\right)=d-1$ given that the integral
above is restricted to the $\left(d-1\right)$-dimensional hyperplane$\left\{ x:x^{'}\t_{0}=0\right\} .$

\subsection{Lebesgue Representation of Hausdorff Integrals via Change of Coordinates}

\noindent It will become subsequently convenient to work with an alternative
representation of the Lebesgue measure
\begin{defn}[Change of Coordinates]
\label{def:BasisT} Let $\left\{ \t,\tilde{e}_{\t,2},..,\tilde{e}_{\t,d}\right\} $
be an orthonormal basis in $\R^{d}$. Define $T_{\t}$ to be the $d\times d$
basis transformation matrix
\[
T_{\t}:=\left(\t,\tilde{e}_{\t,2},..,\tilde{e}_{\t,d}\right).
\]
We write $u:=T_{\t}^{'}x=\left(x^{'}\t,x^{'}\tilde{e}_{\t,2},..,x^{'}\tilde{e}_{\t,d}\right)$. 
\end{defn}
Clearly, since $T_{\t}^{'}=T_{\t}^{-1}$, we have $x=T_{\t}u$. Furthermore,
notice that $\left|\text{det}\left(T_{\t}\right)\right|=1$ due to
orthonormality.
\begin{lem}
Let $m\left(x\right)$ be a $P$-square-integrable function, and write
$m_{u}\left(u\right):=m\left(T_{\t}u\right)$ as the representation
of $m$ under the change of coordinates from $x$ to $u$ as in Definition
\ref{def:BasisT}. Then,
\[
\int_{x^{'}\t_{0}=t}m\left(x\right)d{\cal H}^{d-1}\left(x\right)\equiv\int_{u_{1}=t}m_{u}\left(t,u_{-1}\right)du_{-1}.
\]
\end{lem}
\begin{proof}
By 
\begin{align*}
\int_{x^{'}\t_{0}=t}m\left(x\right)d{\cal H}^{d-1}\left(x\right) & =\Dif_{t}\left[\int_{x^{'}\t_{0}\leq t}m\left(x\right)dx\right]=\Dif_{t}\left[\int_{u_{1}\leq t}m_{u}\left(u\right)du\right]\\
 & =\int\left[\Dif_{t}\int_{-\infty}^{t}m_{u}\left(u_{1},u_{-1}\right)du_{1}\right]du_{-1}=\int m_{u}\left(t,u_{-1}\right)du_{-1}.
\end{align*}
\end{proof}

\subsection{\label{subsec:Pf_Term4}Proof of Lemma \ref{lem:Term4}}
\begin{proof}
Recall that $\Dif_{\t}Pg_{+,\t_{0},h}=P\Dif_{\t}g_{+,\t_{0},h}$ with
\begin{align*}
P\Dif_{\t}g_{+,\t_{0},h}= & \int\ind\left\{ h\left(x\right)>-x^{'}\t_{0}>0\right\} xp\left(x\right)dx\\
= & \int\ind\left\{ h\left(T_{\t_{0}}u\right)>-u_{1}>0\right\} T_{\t_{0}}up\left(T_{\t_{0}}u\right)du\\
= & \int\ind\left\{ h\left(T_{\t_{0}}u\right)>-u_{1}>0\right\} T_{\t_{0}}up\left(T_{\t_{0}}u\right)du\\
= & \int\left[\int\ind\left\{ h\left(T_{\t_{0}}u\right)>-u_{1}>0\right\} p\left(T_{\t_{0}}u\right)du_{1}\right]T_{\t_{0}}\ol u_{-1}du_{-1}
\end{align*}
Taking directional derivative of $\Dif_{\t}Pg_{+,\t_{0},h}$ w.r.t.
$h$ around $h_{0}$ in the direction of $h-h_{0}$, we have
\begin{align*}
 & \frac{1}{t}\left(\Dif_{\t}Pg_{+,\t_{0},h_{0}+t\left(h-h_{0}\right)}-\Dif_{\t}Pg_{+,\t_{0},h_{0}}\right)=\frac{1}{t}\Dif_{\t}Pg_{+,\t_{0},h_{0}+t\left(h-h_{0}\right)}\\
 & \frac{1}{t}\int\int\ind\left\{ h_{0}\left(T_{\t_{0}}u\right)+t\left(h\left(T_{\t_{0}}u\right)-h_{0}\left(T_{\t_{0}}u\right)\right)>-u_{1}>0\right\} p\left(T_{\t_{0}}u\right)du_{1}T_{\t}\ol u_{-1}du_{-1}\\
= & \frac{1}{t}\int\int\ind\left\{ h_{0u}\left(u_{1},u_{-1}\right)+t\left(h_{u}\left(u_{1},u_{-1}\right)-h_{0u}\left(u_{1},u_{-1}\right)\right)>-u_{1}>0\right\} p\left(T_{\t_{0}}u\right)du_{1}T_{\t}\ol u_{-1}du_{-1}\\
= & \int\left[\frac{1}{t}\int_{u_{1}^{*}\left(u_{-1},t\right)}^{0}p_{u}\left(u_{1},u_{-1}\right)du_{1}\right]T_{\t}\ol u_{-1}du_{-1}
\end{align*}
where 
\[
u_{1}^{*}\left(u_{-1},t\right):=\inf\left\{ u_{1}\leq0:h_{0u}\left(u_{1},u_{-1}\right)+t\left(h_{u}\left(u_{1},u_{-1}\right)-h_{0u}\left(u_{1},u_{-1}\right)\right)+u_{1}\geq0\right\} .
\]
Since $h_{u}\left(0,u_{-1}\right)>0$, then
\[
h_{0u}\left(0,u_{-1}\right)+t\left(h_{u}\left(0,u_{-1}\right)-h_{0u}\left(0,u_{-1}\right)\right)+0=th_{u}\left(0,u_{-1}\right)>0
\]
and thus 
\[
u_{1}^{*}\left(u_{-1},t\right)<0
\]
with
\[
h_{0u}\left(u_{1}^{*}\left(u_{-1},t\right),u_{-1}\right)+t\left(h_{u}\left(u_{1}^{*}\left(u_{-1},t\right),u_{-1}\right)-h_{0u}\left(u_{1}^{*}\left(u_{-1},t\right),u_{-1}\right)\right)+u_{1}^{*}\left(u_{-1},t\right)=0
\]
and thus
\[
\left[\Dif_{u_{1}}h_{0u}+t\left(\Dif_{u_{1}}h_{u}-\Dif_{u_{1}}h_{0u}\right)+1\right]\Dif_{t}u_{1}^{*}\left(u_{-1},t\right)+h_{u}\left(u_{1}^{*}\left(u_{-1},t\right),u_{-1}\right)-h_{0u}\left(u_{1}^{*}\left(u_{-1},t\right),u_{-1}\right)=0
\]
and thus
\[
\Dif_{t}u_{1}^{*}\left(u_{-1},t\right)=-\frac{1}{\Dif_{u_{1}}h_{0u}+t\left(\Dif_{u_{1}}h_{u}-\Dif_{u_{1}}h_{0u}\right)+1}\left(h_{u}-h_{0u}\right)
\]
with all functions of $u$ in the formulas above evaluated $\left(u_{1}^{*}\left(u_{-1},t\right),u_{-1}\right)$.
Hence, 
\begin{align*}
 & \lim_{t\to0}\frac{1}{t}\int_{u_{1}^{*}\left(u_{-1},t\right)}^{0}p\left(T_{\t_{0}}u\right)du_{1}\\
= & -p\left(T_{\t_{0}}\ol u_{-1}\right)\cd\rest{\Dif_{t}u_{1}^{*}\left(u_{-1},t\right)}_{t=0}\\
= & p\left(T_{\t_{0}}\ol u_{-1}\right)\cd\frac{1}{\Dif_{u_{1}}h_{0u}\left(0,u_{-1}\right)+1}\left[h_{u}\left(0,u_{-1}\right)-h_{0u}\left(0,u_{-1}\right)\right]
\end{align*}
and thus
\begin{align*}
D_{h}\left(P\Dif_{\t}g_{+,\t_{0},h_{0}},h-h_{0}\right) & =\int\left[h\left(T_{\t_{0}}\ol u_{-1}\right)-h_{0}\left(T_{\t_{0}}\ol u_{-1}\right)\right]\frac{1}{\Dif_{x}h_{0}\left(T_{\t_{0}}\ol u_{-1}\right)^{'}\t_{0}+1}T_{\t}\ol u_{-1}p_{u}\left(\ol u_{-1}\right)du_{-1}
\end{align*}
\end{proof}
Then, noticing that
\begin{align*}
D_{h}\left(P\Dif_{\t}g_{-,\t_{0},h_{0}},h-h_{0}\right) & =-\int\left[-\left(h\left(T_{\t_{0}}\ol u_{-1}\right)-h_{0}\left(T_{\t_{0}}\ol u_{-1}\right)\right)\right]\frac{1}{\Dif_{x}h_{0}\left(T_{\t_{0}}\ol u_{-1}\right)^{'}\t_{0}+1}T_{\t}\ol u_{-1}p_{u}\left(\ol u_{-1}\right)du_{-1}
\end{align*}
we have
\begin{align*}
D_{h}\left(P\Dif_{\t}g_{\t_{0},h_{0}},h-h_{0}\right) & =\int\left[h\left(T_{\t_{0}}\ol u_{-1}\right)-h_{0}\left(T_{\t_{0}}\ol u_{-1}\right)\right]\frac{1}{\Dif_{x}h_{0}\left(T_{\t_{0}}\ol u_{-1}\right)^{'}\t_{0}+1}T_{\t}\ol u_{-1}p_{u}\left(\ol u_{-1}\right)du_{-1}
\end{align*}
Reversing the change of coordinates, we have
\begin{align*}
D_{h}\left(P\Dif_{\t}g_{\t_{0},h_{0}},h-h_{0}\right) & =\int_{x^{'}\t_{0}=0}\left[h\left(x\right)-h_{0}\left(x\right)\right]\frac{1}{\Dif_{x}h_{0}\left(x\right)^{'}\t_{0}+1}xp\left(x\right)d{\cal H}^{d-1}\left(x\right)
\end{align*}
Recall that $\Dif_{x}h_{0}\left(x\right)=f\left(\rest{x^{'}\t_{0}}x\right)\t_{0}+F_{x}\left(\rest{x^{'}\t_{0}}x\right)$
and $F_{x}\left(\rest 0x\right)\equiv0$. Hence, for any $x$ s.t.
$x^{'}\t_{0}=0$, we have
\begin{align*}
\Dif_{x}h_{0}\left(x\right) & =f\left(\rest 0x\right)\t_{0}
\end{align*}
and thus
\[
\Dif_{x}h_{0}\left(x\right)^{'}\t_{0}=f\left(\rest 0x\right).
\]
Hence, 
\begin{align*}
D_{h}\left(P\Dif_{\t}g_{\t_{0},h_{0}},h-h_{0}\right) & =\int_{x^{'}\t_{0}=0}\left[h\left(x\right)-h_{0}\left(x\right)\right]\frac{1}{f\left(\rest 0x\right)+1}xp\left(x\right)d{\cal H}^{d-1}\left(x\right).
\end{align*}

Now, we control the size of the remainder term from the linearization
above. Notice that
\begin{align*}
\int_{u_{1}^{*}\left(u_{-1},t\right)}^{0}p_{u}\left(u_{1},u_{-1}\right)du_{1}= & \Dif_{t}\int_{u_{1}^{*}\left(u_{-1},t\right)}^{0}p_{u}\left(u_{1},u_{-1}\right)du_{1}\cd t\\
 & +\Dif_{t}^{2}\int_{u_{1}^{*}\left(u_{-1},\tilde{t}\right)}^{0}p_{u}\left(u_{1},u_{-1}\right)du_{1}\cd t^{2}
\end{align*}
for some $\tilde{t}\in\left[0,t\right]$, where
\begin{align*}
\Dif_{t}\int_{u_{1}^{*}\left(u_{-1},t\right)}^{0}p_{u}\left(u_{1},u_{-1}\right)du_{1} & =-p_{u}\left(u_{1}^{*}\left(u_{-1},t\right),u_{-1}\right)\cd\Dif_{t}u_{1}^{*}\left(u_{-1},t\right)
\end{align*}
and
\begin{align*}
 & \Dif_{t}^{2}\int_{u_{1}^{*}\left(u_{-1},t\right)}^{0}p_{u}\left(u_{1},u_{-1}\right)du_{1}\\
= & -\Dif_{t}\left[p_{u}\left(u_{1}^{*}\left(u_{-1},t\right),u_{-1}\right)\cd\Dif_{t}u_{1}^{*}\left(u_{-1},t\right)\right]\\
= & -\Dif_{u_{1}}p_{u}\left(u_{1}^{*}\left(u_{-1},t\right),u_{-1}\right)\cd\left[\Dif_{t}u_{1}^{*}\left(u_{-1},t\right)\right]^{2}-p_{u}\left(u_{1}^{*}\left(u_{-1},t\right),u_{-1}\right)\Dif_{t}^{2}u_{1}^{*}\left(u_{-1},t\right).
\end{align*}
Hence,
\begin{align}
 & \norm{\frac{1}{t}\Dif_{\t}Pg_{+,\t_{0},h_{0}+t\left(h-h_{0}\right)}-D_{h}\left(P\Dif_{\t}g_{+,\t_{0},h_{0}},t\left(h-h_{0}\right)\right)}\nonumber \\
\leq & t^{2}\norm{\int\Dif_{t}^{2}\int_{u_{1}^{*}\left(u_{-1},\tilde{t}\right)}^{0}p_{u}\left(u_{1}|u_{-1}\right)du_{1}T_{\t_{0}}\ol u_{-1}p_{u}\left(u_{-1}\right)du_{-1}}\nonumber \\
\leq & t^{2}\norm{\int\Dif_{u_{1}}p_{u}\left(u_{1}^{*}\left(u_{-1},\tilde{t}\right),u_{-1}\right)\cd\left[\Dif_{t}u_{1}^{*}\left(u_{-1},t\right)\right]^{2}T_{\t_{0}}\ol u_{-1}du_{-1}}\nonumber \\
 & +t^{2}\norm{\int p_{u}\left(u_{1}^{*}\left(u_{-1},\tilde{t}\right),u_{-1}\right)\Dif_{t}^{2}u_{1}^{*}\left(u_{-1},\tilde{t}\right)T_{\t_{0}}\ol u_{-1}du_{-1}}\label{eq:Dh_rem}
\end{align}
Recall that 
\begin{align*}
\Dif_{t}u_{1}^{*}\left(u_{-1},t\right) & =-\frac{h_{u}-h_{0u}}{\Dif_{u_{1}}\left[h_{0u}+t\left(h_{u}-h_{0u}\right)\right]+1}
\end{align*}
with all functions of $u$ in the above evaluated at $\left(u_{1}^{*}\left(u_{-1},t\right),u_{-1}\right)$.
Hence,
\[
\left|\Dif_{t}u_{1}^{*}\left(u_{-1},t\right)\right|^{2}\leq M\norm{h_{u}-h_{0u}}_{\infty}^{2}=M\norm{h-h_{0}}_{\infty}^{2}
\]
since $\frac{1}{\Dif_{u_{1}}\left[h_{0u}+t\left(h_{u}-h_{0u}\right)\right]+1}\leq M$.
Furthermore, since $\Dif_{u_{1}}p_{u}\left(\rest{u_{1}}u_{-1}\right)\leq M$,
we have
\begin{align}
 & \left|\int\Dif_{u_{1}}p_{u}\left(u_{1}^{*}\left(u_{-1},\tilde{t}\right),u_{-1}\right)\cd\left[\Dif_{t}u_{1}^{*}\left(u_{-1},t\right)\right]^{2}T_{\t_{0}}\ol u_{-1}du_{-1}\right|\nonumber \\
\leq & M\left|\int p_{u}\left(u_{-1}\right)\norm{h-h_{0}}_{\infty}^{2}T_{\t_{0}}\ol u_{-1}du_{-1}\right|\leq M\norm{h-h_{0}}_{\infty}^{2}.\label{eq:Dh_rem1}
\end{align}
Now, for the last term in \eqref{eq:Dh_rem}, notice that
\begin{align*}
\Dif_{t}^{2}u_{1}^{*}\left(u_{-1},t\right)= & -\frac{1}{\left[\Dif_{u_{1}}\left[h_{0u}+t\left(h_{u}-h_{0u}\right)\right]+1\right]^{2}}\cd\\
 & \cd\left\{ \begin{array}{c}
\Dif_{u_{1}}\left(h_{u}-h_{0}\right)\cd\left(\Dif_{u_{1}}\left[h_{0u}+t\left(h_{u}-h_{0u}\right)\right]+1\right)\cd\Dif_{t}u_{1}^{*}\left(u_{-1},t\right)\\
-\left(h_{u}-h_{0u}\right)\cd\left[\Dif_{u_{1}}^{2}\left(h_{0u}+t\left(h_{u}-h_{0u}\right)\right)\cd\Dif_{t}u_{1}^{*}\left(u_{-1},t\right)\right]\\
-\left(h_{u}-h_{0u}\right)\Dif_{u_{1}}\left(h_{u}-h_{0u}\right)
\end{array}\right\} 
\end{align*}
Since $\left|\Dif_{u_{1}}h_{u}\right|,\left|\Dif_{u_{1}}h_{0u}\right|$,
$\left|\Dif_{u_{1}}^{2}h_{u}\right|,\left|\Dif_{u_{1}}^{2}h_{0u}\right|$
and $\frac{1}{\left|\Dif_{u_{1}}\left[h_{0u}+t\left(h_{u}-h_{0u}\right)\right]+1\right|}$
are all uniformly bounded from above by some constant $M$, we have
\begin{align}
 & \norm{\Dif_{u_{1}}\left(h_{u}-h_{0}\right)\cd\left(\Dif_{u_{1}}\left[h_{0u}+t\left(h_{u}-h_{0u}\right)\right]+1\right)\cd\Dif_{t}u_{1}^{*}\left(u_{-1},t\right)}\nonumber \\
\leq & M\cd\norm{\Dif_{x}\left(h-h_{0}\right)}\cd\norm{h-h_{0}}_{\infty}\label{eq:Dh_rem2}
\end{align}
and 
\begin{align}
\norm{-\left(h_{u}-h_{0u}\right)\cd\left[\Dif_{u_{1}}^{2}\left(h_{0u}+t\left(h_{u}-h_{0u}\right)\right)\cd\frac{\p}{\p t}u_{1}^{*}\left(u_{-1},t\right)\right]}\leq M\norm{h-h_{0}}_{\infty}^{2}\label{eq:Dh_rem3}
\end{align}
and
\begin{equation}
\norm{-\left(h_{u}-h_{0u}\right)\Dif_{u_{1}}\left(h_{u}-h_{0u}\right)}\leq\norm{\Dif_{x}\left(h-h_{0}\right)}\cd\norm{h-h_{0}}_{\infty}.\label{eq:Dh_rem4}
\end{equation}
Combining \eqref{eq:Dh_rem1}- \eqref{eq:Dh_rem4}, we can bound \eqref{eq:Dh_rem}
by
\begin{align*}
 & \norm{\frac{1}{t}\Dif_{\t}Pg_{+,\t_{0},h_{0}+t\left(h-h_{0}\right)}-D_{h}\left(P\Dif_{\t}g_{+,\t_{0},h_{0}},t\left(h-h_{0}\right)\right)}\\
\leq & t^{2}M\norm{h-h_{0}}_{\infty}\left(\norm{h-h_{0}}_{\infty}+\norm{\Dif_{x}\left(h-h_{0}\right)}_{\infty}\right)
\end{align*}

Now, with $\hat{h}$ plugged in place of $h$, we write $\norm{\hat{h}-h_{0}}_{\infty}=O_{p}\left(a_{n}\right)$
and $\norm{\Dif_{x}\left(\hat{h}-h_{0}\right)}_{\infty}=O_{p}\left(c_{n}\right)$.
Since it is well-known that the convergence rate of $\Dif_{x}\left(\hat{h}-h_{0}\right)$
is slower than $\hat{h}-h_{0}$, we have
\[
\norm{\Dif_{\t}Pg_{+,\t_{0},\hat{h}}-D_{h}\left(P\Dif_{\t}g_{+,\t_{0},h_{0}},\hat{h}-h_{0}\right)}\leq Ma_{n}c_{n}.
\]

Lastly, recall that

\begin{align*}
\Dif_{\t\t}Pg_{\t,h} & =\int_{x^{'}\t=0}xx^{'}p\left(x\right)d{\cal H}^{d-1}\left(x\right)-\int_{x^{'}\t=-h\left(x\right)}\frac{1}{\norm{\Dif_{x}h\left(x\right)+\t}}xx^{'}p\left(x\right)d{\cal H}^{d-1}\left(x\right)
\end{align*}
and thus
\begin{align*}
\Dif_{\t\t}P\left(g_{\t_{0},h}-g_{\t_{0},h_{0}}\right) & =-\left[\int_{x^{'}\t_{0}=-h\left(x\right)}\frac{1}{\norm{\Dif_{x}h\left(x\right)+\t_{0}}}-\int_{x^{'}\t_{0}=0}\frac{1}{\norm{\Dif_{x}h_{0}\left(x\right)+\t_{0}}}\right]xx^{'}p\left(x\right)d{\cal H}^{d-1}\left(x\right)
\end{align*}
Given that $\norm{\hat{h}-h_{0}}_{\infty}=o_{p}\left(1\right)$ and
$\norm{\Dif_{x}\hat{h}-\Dif_{x}h_{0}}_{\infty}=o_{p}\left(1\right)$
in Assumption \ref{assu:FirstStageDeriv}, we have
\[
\Dif_{\t\t}P\left(g_{\t_{0},\hat{h}}-g_{\t_{0},h_{0}}\right)=o_{p}\left(1\right),
\]
and thus
\[
\left(\t-\t_{0}\right)^{'}\Dif_{\t\t}P\left(g_{\t_{0},\hat{h}}-g_{\t_{0},h_{0}}\right)\left(\t-\t_{0}\right)=o_{p}\left(\norm{\t-\t_{0}}^{2}\right).
\]

\subsection{\label{subsec:LowerDim}Proof of Lemma \ref{lem:T4_Kern}}

We first provide a lemma that we will used in the proof of Lemma \ref{lem:T4_Kern}.
\begin{lem}[Lower-Dimensional Integral of Kernels]
\label{lem:KernInteg} Define
\[
G\left(t\right):=\int_{x^{'}\t_{0}=t}K\left(x\right)d{\cal H}^{d-1}\left(x\right).
\]
Then:
\begin{itemize}
\item[(a)]  $G\left(t\right)=\int K_{u}\left(t,u_{-1}\right)du_{-1}$ under
the change of coordinates in Definition \ref{def:BasisT}.
\item[(b)]  $G\left(t\right)$ is a unidimensional kernel of smoothness order
$s$.
\end{itemize}
\end{lem}
\begin{proof}
Note that
\[
\Dif_{t}\int_{x^{'}\t_{0}\leq t}K\left(x\right)dx=\int_{x^{'}\t_{0}=0}K\left(x\right)d{\cal H}^{d-1}\left(x\right)=G\left(t\right)
\]
and thus (iii) holds since
\begin{align*}
\int G\left(t\right)dt & =\int\left[\int_{x^{'}\t_{0}=0}K\left(x\right)d{\cal H}^{d-1}\left(x\right)\right]dv\\
 & =\int\left[\Dif_{t}\int_{x^{'}\t_{0}\leq t}K\left(x\right)dx\right]dt=\int K\left(x\right)dx=1,
\end{align*}
(ii) holds since
\begin{align*}
G\left(-t\right) & =\int_{x^{'}\t_{0}=0}K\left(x\right)d{\cal H}^{d-1}\left(x\right)\\
 & =\int_{u^{'}\t_{0}=0}K\left(-u\right)d{\cal H}^{d-1}\left(-u\right)\quad\text{with }u:=-x\\
 & =\int_{x^{'}\t_{0}=0}K\left(u\right)d{\cal H}^{d-1}\left(u\right)=G\left(v\right)
\end{align*}
For (v) note that, for any $l\leq s-1$,
\begin{align*}
\int t^{l}G\left(t\right)dt & =\int t^{l}\int_{x^{'}\t_{0}=0}K\left(x\right)d{\cal H}^{d-1}\left(x\right)dt\\
 & =\int\Dif_{t}\int_{x^{'}\t_{0}\leq t}t^{l}K\left(x\right)dxdt\\
 & =\int\Dif_{t}\left[\int_{x^{'}\t_{0}\leq t}\left(x^{'}\t_{0}\right)^{l}K\left(x\right)dx\right]dt\\
 & =\int\left(x^{'}\t_{0}\right)^{l}K\left(x\right)dx\\
 & =\int\left(\sum_{j}\t_{0j}x_{j}\right)^{l}K\left(x\right)dx=0
\end{align*}
since, for any $\left(\a_{1},...,\a_{d}\right)$ s.t. $\a_{j}\in\mathbb{N}$
and $0\leq\a_{j}\leq s-1$, we have
\begin{align*}
\int x_{1}^{\a_{1}}...x_{d}^{\a_{d}}K\left(x\right)dx & =0.
\end{align*}
Furthermore, since $\int K\left(x\right)dx=1$ we have (i)
\[
\left|G\left(t\right)\right|\leq M:=\int\left[K\left(x\right)\right]_{+}dx<\infty
\]
and (iv):
\begin{align*}
\int\left|t\right|^{l}G\left(t\right)dt & =\int\int\left|x^{'}\t_{0}\right|^{l}K\left(x\right)d{\cal H}^{d-1}\left(x\right)dt\\
 & =\int\int\left|x^{'}\t_{0}\right|^{l}K\left(x\right)dx<\infty.
\end{align*}
Lastly, (vi) is trivially true since $G$ is a univariate function.
\end{proof}

\subsubsection*{Proof of Lemma \ref{lem:T4_Kern}}
\begin{proof}
Write $w\left(x\right):=\frac{1}{f\left(\rest 0x\right)+1}x$ so that
\[
L\left(h\right)=\int_{x^{'}\t_{0}=0}h\left(x\right)w\left(x\right)p\left(x\right)d{\cal H}^{d-1}\left(x\right).
\]
Write $m\left(x\right):=h_{0}\left(x\right)p\left(x\right)$, $\hat{m}\left(x\right):=\frac{1}{nb_{n}^{d}}\sum_{i=1}^{n}K\left(\frac{X_{i}-x}{b_{n}}\right)\left(Y_{i}-\frac{1}{2}\right)$
and $\hat{p}\left(x\right):=\frac{1}{n}\sum_{i=1}^{n}K\left(\frac{X_{i}-x}{b_{n}}\right)$
so that $h_{0}\left(x\right)=m\left(x\right)/p\left(x\right)$ and
$\hat{h}\left(x\right)=\hat{m}\left(x\right)/\hat{p}\left(x\right)$,
we have

\begin{align}
L\left(\hat{h}\right)-L\left(h_{0}\right) & =\int_{x^{'}\t_{0}=0}\left[\hat{h}\left(x\right)-h_{0}\left(x\right)\right]w\left(x\right)p\left(x\right)d{\cal H}^{d-1}\left(x\right)\nonumber \\
 & =\int_{x^{'}\t_{0}=0}\left[\frac{\hat{m}\left(x\right)}{\hat{p}\left(x\right)}-\frac{m\left(x\right)}{p\left(x\right)}\right]w\left(x\right)p\left(x\right)d{\cal H}^{d-1}\left(x\right)\nonumber \\
 & =\int_{x^{'}\t_{0}=0}\left[\frac{\hat{m}\left(x\right)-m\left(x\right)}{p\left(x\right)}-\frac{m\left(x\right)}{p^{2}\left(x\right)}\left(\hat{p}\left(x\right)-p\left(x\right)\right)\right]w\left(x\right)p\left(x\right)d{\cal H}^{d-1}\left(x\right)+R_{1}\nonumber \\
 & =\int_{x^{'}\t_{0}=0}\left[\hat{m}\left(x\right)-m\left(x\right)-h_{0}\left(x\right)\left(\hat{p}\left(x\right)-p\left(x\right)\right)\right]w\left(x\right)d{\cal H}^{d-1}\left(x\right)+R_{1}\nonumber \\
 & =\int_{x^{'}\t_{0}=0}\left[\hat{m}\left(x\right)-h_{0}\left(x\right)\hat{p}\left(x\right)\right]w\left(x\right)d{\cal H}^{d-1}\left(x\right)+R_{1}\nonumber \\
 & =\int_{x^{'}\t_{0}=0}\frac{1}{nb_{n}^{d}}\sum_{i=1}^{n}K\left(\frac{x-X_{i}}{b_{n}}\right)\left(Y_{i}-\frac{1}{2}-h_{0}\left(x\right)\right)w\left(x\right)d{\cal H}^{d-1}\left(x\right)+R_{1}\nonumber \\
 & =\underset{T_{41}}{\underbrace{\frac{1}{nb_{n}^{d}}\sum_{i=1}^{n}\int_{x^{'}\t_{0}=0}K\left(\frac{x-X_{i}}{b_{n}}\right)\left(Y_{i}-\frac{1}{2}-h_{0}\left(x\right)\right)w\left(x\right)d{\cal H}^{d-1}\left(x\right)}}+R_{1},\label{eq:Kern_Ldh}
\end{align}
where the remainder term $R_{1}=O\left(\norm{\hat{m}-m}^{2}+\norm{\hat{p}-m}^{2}+2\norm{\hat{m}-m}\norm{\hat{p}-p}\right)$
is asymptotically negligible.

Next, we will work with the change of coordinate $u=T_{\t_{0}}x$
with $u_{1}=x^{'}\t_{0}$, and write $U_{i}:=T_{\t_{0}}X_{i}$. Then
we apply the usual ``kernel change of variable'' technique on $u_{-1}$,
setting
\[
v_{-1}:=\frac{u_{-1}-U_{i,-1}}{b_{n}},\quad\text{so that }u_{-1}=U_{i,-1}+b_{n}v_{-1},
\]
Correspondingly, the constraint $x^{'}\t_{0}=0$ becomes $u_{1}=0$,
and thus we can write
\begin{align}
T_{41} & =\frac{1}{nb_{n}^{d}}\sum_{i=1}^{n}\int_{u_{1}=0}K_{u}\left(\frac{u-U_{i}}{b_{n}}\right)\left(Y_{i}-\frac{1}{2}-h_{0u}\left(u\right)\right)w_{u}\left(u\right)d{\cal H}^{d-1}\left(u\right)\nonumber \\
 & =\frac{1}{nb_{n}^{d}}\sum_{i=1}^{n}\int K_{u}\left(-\frac{U_{i1}}{b_{n}},\frac{u_{-1}-U_{i,-1}}{b_{n}}\right)\left(Y_{i}-\frac{1}{2}-h_{0u}\left(0,u_{-1}\right)\right)w_{u}\left(0,u_{-1}\right)du_{-1}\nonumber \\
 & =\frac{1}{nb_{n}^{d}}\sum_{i=1}^{n}\int K_{u}\left(-\frac{U_{i1}}{b_{n}},v_{-1}\right)\left(Y_{i}-\frac{1}{2}-h_{0u}\left(0,U_{i,-1}+b_{n}v_{-1}\right)\right)w_{u}\left(0,U_{i,-1}+b_{n}v_{-1}\right)b_{n}^{d-1}dv_{-1}\nonumber \\
 & =\frac{1}{nb_{n}}\sum_{i=1}^{n}\int K_{u}\left(-\frac{U_{i1}}{b_{n}},v_{-1}\right)\left(Y_{i}-\frac{1}{2}-h_{0u}\left(0,U_{i,-1}\right)+O\left(b_{n}\right)\right)\left[w_{u}\left(0,U_{i,-1}\right)+O\left(b_{n}\right)\right]dv_{-1}\nonumber \\
 & =\frac{1}{nb_{n}}\sum_{i=1}^{n}\int K_{u}\left(-\frac{U_{i1}}{b_{n}},v_{-1}\right)\left(Y_{i}-\frac{1}{2}-h_{0u}\left(0,U_{i,-1}\right)\right)w_{u}\left(0,U_{i,-1}\right)dv_{-1}+R_{2}\nonumber \\
 & =\frac{1}{nb_{n}}\sum_{i=1}^{n}\left(Y_{i}-\frac{1}{2}-h_{0u}\left(0,U_{i,-1}\right)\right)w_{u}\left(0,U_{i,-1}\right)\int K_{u}\left(-\frac{U_{i1}}{b_{n}},v_{-1}\right)dv_{-1}+R_{2}\nonumber \\
 & \underset{T_{42}}{=\underbrace{\frac{1}{nb_{n}}\sum_{i=1}^{n}\left(Y_{i}-\frac{1}{2}-h_{0u}\left(0,U_{i,-1}\right)\right)w_{u}\left(0,U_{i,-1}\right)G\left(\frac{U_{i1}}{b_{n}}\right)}}+R_{2}\label{eq:Kern_d1_est}
\end{align}
where $R_{2}$ is asymptotically negligible. By Lemma \ref{lem:KernInteg},
$G\left(t\right)$ is a univariate kernel function of smoothness order
$s$, and hence the asymptotic behavior of the leading term $T_{42}$
in \eqref{eq:Kern_d1_est} can be established in the same way as for
a univariate kernel estimator. 

Formally, we analyze $\E\left[T_{41}\right]$ and $\text{Var}\left[T_{41}\right]$
separately. For $\E\left[T_{41}\right]$, we have 
\begin{align*}
 & \E\left[T_{42}\right]\\
=\  & \E\left[\frac{1}{nb_{n}}\sum_{i=1}^{n}\left(Y_{i}-\frac{1}{2}-h_{0u}\left(0,U_{i,-1}\right)\right)w_{u}\left(0,U_{i,-1}\right)G\left(\frac{U_{i,1}}{b_{n}}\right)\right]\\
=\  & \frac{1}{b_{n}}\E\left[\left(Y_{i}-\frac{1}{2}-h_{0u}\left(0,U_{i,-1}\right)\right)w_{u}\left(0,U_{i,-1}\right)G\left(\frac{U_{i,1}}{b_{n}}\right)\right]\\
=\  & \frac{1}{b_{n}}\E\left[\left(h_{0u}\left(U_{i}\right)-h_{0u}\left(0,U_{i,-1}\right)\right)w_{u}\left(0,U_{i,-1}\right)G\left(\frac{U_{i,1}}{b_{n}}\right)\right]\\
=\  & \frac{1}{b_{n}}\int\left(h_{0u}\left(u_{1},u_{-1}\right)-h_{0u}\left(0,u_{-1}\right)\right)w_{u}\left(0,u_{-1}\right)G\left(\frac{u_{1}}{b_{n}}\right)p_{u}\left(u_{1},u_{-1}\right)du_{1}du_{-1}\\
=\  & \frac{1}{b_{n}}\int\left(h_{0u}\left(b_{n}v_{1},u_{-1}\right)-h_{0u}\left(0,u_{-1}\right)\right)w_{u}\left(0,u_{-1}\right)G\left(v_{1}\right)p_{u}\left(b_{n}v_{1},u_{-1}\right)b_{n}dv_{1}du_{-1}\quad\text{with }\frac{u_{1}}{b_{n}}=v_{1}\\
=\  & \int\left(h_{0u}\left(b_{n}v_{1},u_{-1}\right)-h_{0u}\left(0,u_{-1}\right)\right)w_{u}\left(0,u_{-1}\right)G\left(v_{1}\right)p_{u}\left(b_{n}v_{1},u_{-1}\right)dv_{1}du_{-1}\\
=\  & \int\left[\int\phi\left(b_{n}v_{1},u_{-1}\right)G\left(v_{1}\right)dv_{1}\right]w_{u}\left(0,u_{-1}\right)du_{-1}\\
=\  & \int G\left(v_{1}\right)dv_{1}\cd\int h_{0u}\left(0,u_{-1}\right)w_{u}\left(0,u_{-1}\right)p_{u}\left(0,u_{-1}\right)+0+b^{s}\int v_{1}^{s}G\left(v_{1}\right)dv_{1}\Dif_{u_{1}}\left[hp\right]\left(0,u_{-1}\right)w_{u}\left(0,u_{-1}\right)du_{-1}\\
=\  & \int h_{0u}\left(0,u_{-1}\right)w_{u}\left(0,u_{-1}\right)p_{u}\left(0,u_{-1}\right)du_{-1}+b^{s}\kappa_{s}+o\left(b^{s}\right)\\
=\  & \int_{x^{'}\t_{0}=0}h_{0}\left(x\right)w\left(x\right)p\left(x\right)d{\cal H}^{d-1}\left(x\right)+b^{s}\kappa_{s}+o\left(b^{s}\right)
\end{align*}
By
\begin{align*}
\phi\left(b_{n}v_{1},u_{-1}\right):= & \left(h_{0u}\left(b_{n}v_{1},u_{-1}\right)-h_{0u}\left(0,u_{-1}\right)\right)p_{u}\left(b_{n}v_{1},u_{-1}\right)\\
= & 0+\sum_{j=1}^{s-1}\Dif_{u_{1}}^{\left(j\right)}\phi\left(0,u_{-1}\right)b_{n}^{j}v_{1}^{j}+\Dif_{u_{1}}^{\left(s\right)}\phi\left(0,u_{-1}\right)b_{n}^{s}v_{1}^{s}+o\left(b_{n}^{s}\right)
\end{align*}
which implies that 
\begin{align*}
 & \int\phi\left(b_{n}v_{1},u_{-1}\right)G\left(v_{1}\right)dv_{1}\\
= & \sum_{j=1}^{s-1}\Dif_{u_{1}}^{\left(j\right)}\phi\left(0,u_{-1}\right)b_{n}^{j}\underset{=0}{\underbrace{\int v_{1}^{j}G\left(v_{1}\right)dv_{1}}}+\Dif_{u_{1}}^{\left(s\right)}\phi\left(0,u_{-1}\right)b_{n}^{s}\underset{=:\kappa_{G,s}}{\underbrace{\int v_{1}^{s}G\left(v_{1}\right)dv_{1}}}+o\left(b_{n}^{s}\right)\\
= & b_{n}^{s}\kappa_{G,s}\cd\Dif_{u_{1}}^{\left(s\right)}\phi\left(0,u_{-1}\right)
\end{align*}
and thus, writing $B_{s}:=\int\Dif_{u_{1}}^{\left(s\right)}\phi\left(0,u_{-1}\right)w_{u}\left(0,u_{-1}\right)du_{-1}$,
we have
\begin{align}
\E\left[T_{42}\right] & =\int b_{n}^{s}\kappa_{G,s}\cd\Dif_{u_{1}}^{\left(s\right)}\phi\left(0,u_{-1}\right)w_{u}\left(0,u_{-1}\right)du_{-1}+o\left(b_{n}^{s}\right)=b_{n}^{s}B_{s}+o\left(b_{n}^{s}\right).\label{eq:Bias_T42}
\end{align}

`Next, for $\text{Var}\left(T_{42}\right)$, we have
\begin{align}
 & \text{Var}\left(T_{42}\right)\nonumber \\
=\  & \text{Var}\left(\frac{1}{nb_{n}}\sum_{i}\left(Y_{i}-\frac{1}{2}-h_{0u}\left(0,U_{i,-1}\right)\right)w_{u}\left(0,U_{i,-1}\right)G\left(\frac{U_{i,1}}{b_{n}}\right)\right)\nonumber \\
=\  & \frac{1}{n}\E\left[\left(\frac{1}{b_{n}^{2}}\left(Y_{i}-\frac{1}{2}-h_{0u}\left(0,U_{i,-1}\right)\right)^{2}w_{u}\left(0,U_{i,-1}\right)w_{u}\left(0,U_{i,-1}\right)^{'}G^{2}\left(\frac{U_{i,1}}{b_{n}}\right)\right)\right]-\frac{1}{n}\left(\E\left[T_{42}\right]\right)^{2}\nonumber \\
=\  & \frac{1}{nb_{n}^{2}}\E\left[\left(Y_{i}-\frac{1}{2}-h_{0u}\left(0,U_{i,-1}\right)\right)^{2}w_{u}w_{u}^{'}\left(0,U_{i,-1}\right)G^{2}\left(\frac{U_{i,1}}{b_{n}}\right)\right]+o\left(\frac{1}{n}\right)\nonumber \\
=\  & \frac{1}{nb_{n}^{2}}\E\left[\E\left[\rest{\left(Y_{i}-\frac{1}{2}-h_{0u}\left(0,U_{i,-1}\right)\right)^{2}}U_{i}\right]w_{u}w_{u}^{'}\left(0,U_{i,-1}\right)G^{2}\left(\frac{U_{i,1}}{b_{n}}\right)\right]+o\left(\frac{1}{n}\right)\nonumber \\
=\  & \frac{1}{nb_{n}^{2}}\E\left[\left(\s_{0u}^{2}\left(U\right)+\left(h_{0u}\left(U_{i}\right)-h_{0u}\left(0,U_{i,-1}\right)\right)^{2}\right)w_{u}w_{u}^{'}\left(0,U_{i,-1}\right)G^{2}\left(\frac{U_{i,1}}{b_{n}}\right)\right]+o\left(\frac{1}{n}\right)\nonumber \\
=\  & \frac{1}{nb_{n}^{2}}\int\left(\s_{0u}^{2}\left(u\right)+\left(h_{0u}\left(u\right)-h_{0u}\left(0,u_{-1}\right)\right)^{2}\right)w_{u}w_{u}^{'}\left(0,u_{-1}\right)G^{2}\left(\frac{u_{1}}{b_{n}}\right)p_{u}\left(u\right)du+o\left(\frac{1}{n}\right)\nonumber \\
=\  & \frac{1}{nb_{n}}\int\left(\s_{0u}^{2}\left(b_{n}v_{1},u_{-1}\right)+\left(h_{0u}\left(b_{n}v_{1},u\right)-h_{0u}\left(0,u_{-1}\right)\right)^{2}\right)w_{u}w_{u}^{'}\left(0,u_{-1}\right)G^{2}\left(v_{1}\right)p_{u}\left(b_{n}v_{1},u_{-1}\right)dv_{1}du_{-1}+o\left(\frac{1}{n}\right)\nonumber \\
=\  & \frac{1}{nb_{n}}\int G^{2}\left(v_{1}\right)dv_{1}\cd\int\s_{0u}^{2}\left(0,u_{-1}\right)w_{u}w_{u}^{'}\left(0,u_{-1}\right)p_{u}\left(0,u_{-1}\right)du_{-1}+o\left(\frac{1}{nb_{n}}\right)\nonumber \\
=: & \frac{1}{nb_{n}}\O+o\left(\frac{1}{nb}\right)\label{eq:Var_T42}
\end{align}
\end{proof}
where
\begin{align*}
\s_{0}^{2}\left(x\right) & :=\text{Var}\left(\rest{Y_{i}}X_{i}=x\right)=\frac{1}{4}-h_{0}^{2}\left(x\right)\\
\O & :=\int G^{2}\left(v_{1}\right)dv_{1}\cd\int\s_{0u}^{2}\left(0,u_{-1}\right)w_{u}\left(0,u_{-1}\right)w_{u}\left(0,u_{-1}\right)^{'}p_{u}\left(0,u_{-1}\right)du_{-1}\\
 & =R_{G,2}\cd\int_{x^{'}\t_{0}=0}\s_{0}^{2}\left(x\right)w\left(x\right)w\left(x\right)^{'}p\left(x\right)d{\cal H}^{d-1}\left(x\right)\\
 & =R_{G,2}\cd\int_{x^{'}\t_{0}=0}\frac{\s_{0}^{2}\left(x\right)}{\left(f\left(\rest 0x\right)+1\right)^{2}}xx^{'}p\left(x\right)d{\cal H}^{d-1}\left(x\right).
\end{align*}
and the second last line in \eqref{eq:Var_T42} follows from a first-order
Taylor expansion of 
\[
\left(\s_{0u}^{2}\left(b_{n}v_{1},u_{-1}\right)+\left(h_{0u}\left(b_{n}v_{1},u\right)-h_{0u}\left(0,u_{-1}\right)\right)^{2}\right)p_{u}\left(b_{n}v_{1},u_{-1}\right)
\]
with respect to $b_{n}v_{1}$ around $0$.

Combining \eqref{eq:Kern_Ldh}-\eqref{eq:Var_T42}, we have 
\[
L\left(\hat{h}\right)-L\left(h_{0}\right)=O_{p}\left(\frac{1}{\sqrt{nb_{n}}}+b_{n}^{s}\right),
\]
the rate of which is minimized by setting $b_{n}\sim n^{-\frac{1}{2s+1}}$so
that $\frac{1}{\sqrt{nb_{n}}}\sim b_{n}^{s}$ with 
\[
n^{\frac{1}{2s+1}}\left(L\left(\hat{h}\right)-L\left(h_{0}\right)\right)\dto\cN\left(B_{s},\O\right).
\]
With undersmoothing bandwidth $b_{n}=o\left(n^{-\frac{1}{2s+1}}\right)$,
the asymptotic bias $B_{s}$ becomes asymptotically negligible and
thus
\[
\sqrt{nb_{n}}\left(L\left(\hat{h}\right)-L\left(h_{0}\right)\right)\dto\cN\left({\bf 0},\O\right).
\]

\subsection{Proof Lemma \ref{lem:T4_Series}}

For $\hat{h}$ obtained through linear series regression, we apply
the results in \citet*{chen2015optimal} to the characterization of
the asymptotic behavior of $L\left(\hat{h}\right)-L\left(h\right)$.
Since the results in \citet*{chen2015optimal} are stated for scalar-valued
functionals while our $L\left(h\right)$ here is $d$-dimensional,
we consider arbitrary linear combinations of $L\left(h\right)$ by
working with 
\[
L_{c}\left(h\right)=c^{'}L\left(h\right)
\]
for any $c\in\S^{d-1}$. Clearly, $L_{c}\left(h\right)$ is a scalar-valued
linear functional. 

Write $w\left(x\right):=\frac{1}{f\left(\rest 0x\right)+1}x$. Since
$L_{c}$ is linear, 
\begin{align*}
D_{h}\left[L_{c}\left(h_{0}\right),v\right] & =\int_{x^{'}\t_{0}=0}v\left(x\right)c^{'}w\left(x\right)p\left(x\right)d{\cal H}^{d-1}\left(x\right).
\end{align*}
By \citet*{chen2015optimal}, the sieve representer of $D\left[L_{c}\left(h_{0}\right),v\right]$
on the sieve space ${\cal V}_{K_{n}}$ is given by $v_{cK}^{*}\left(\cd\right)=v_{K}^{*}\left(\cd\right)^{'}c$
with 
\begin{align*}
v_{K}^{*}\left(\cd\right) & =\ol b^{K}\left(\cd\right)^{'}\E\left[\ol b^{K}\left(X_{i}\right)\ol b^{K}\left(X_{i}\right)^{'}\right]^{-1}\int_{x^{'}\t_{0}=0}\ol b^{K}\left(x\right)w\left(x\right)^{'}p\left(x\right)d{\cal H}^{d-1}\left(x\right)
\end{align*}
which ensures that, for any $v=b^{K}\left(\cd\right)^{'}\a_{v_{j}}\in{\cal V}_{K}$,
\begin{align*}
 & \E\left[v_{c}\left(X_{i}\right)v_{jK}^{*}\left(X_{i}\right)\right]\\
=\  & \E\left[\a_{v_{j}}^{'}\ol b^{K}\left(X_{i}\right)\ol b^{K}\left(X_{i}\right)^{'}\E\left[\ol b^{K}\left(X_{i}\right)\ol b^{K}\left(X_{i}\right)^{'}\right]^{-1}\int_{x^{'}\t_{0}=0}\ol b^{K}\left(x\right)c^{'}w\left(x\right)p\left(x\right)d{\cal H}^{d-1}\left(x\right)\right]\\
=\  & \a_{v_{j}}^{'}\E\left[\ol b^{K}\left(X_{i}\right)\ol b^{K}\left(X_{i}\right)^{'}\right]\E\left[\ol b^{K}\left(X_{i}\right)\ol b^{K}\left(X_{i}\right)^{'}\right]^{-1}\int_{x^{'}\t_{0}=0}\ol b^{K}\left(x\right)c^{'}w\left(x\right)p\left(x\right)d{\cal H}^{d-1}\left(x\right)\\
=\  & \int_{x^{'}\t_{0}=0}\left[\a_{v_{j}}^{'}\ol b^{K}\left(x\right)\right]c^{'}w\left(x\right)p\left(x\right)d{\cal H}^{d-1}\left(x\right)\\
=\  & \int_{x^{'}\t_{0}=0}v_{j}\left(x\right)c^{'}w\left(x\right)p\left(x\right)d{\cal H}^{d-1}\left(x\right)\\
=\  & D_{h}\left[L_{c}\left(h_{0}\right),v_{j}\right].
\end{align*}
Furthermore, define $\O_{cK}:=c^{'}\O_{K}c$ where
\[
\O_{K}:=\E\left[\s_{0}^{2}\left(X_{i}\right)v_{K}^{*}\left(X_{i}\right)v_{K}^{*}\left(X_{i}\right)^{'}\right].
\]
Notice that $\O_{K}$ has rank $d-1$ and $c^{'}\O_{k}c=0$ if $c=\t_{0}$.
By \citet*{chen2015optimal}, $\O_{cK}$ and $\norm{v_{cK}^{*}}_{L_{2}\left(X\right)}$
share the same rate of growth, and the convergence rate of $L_{c}\left(\hat{h}\right)$
is driven by the rates of $\norm{v_{cK}^{*}}_{L_{2}\left(X\right)}$
and $\O_{cK}$. Note that Assumptions 1-4 in \citet*{chen2015optimal}
are automatically satisfied in our setting.

We first derive the bound on $\norm{v_{cK_{n}}^{*}}_{L_{2}\left(X\right)}^{2}$
in the following Lemma.
\begin{lem}
\label{lem:LinSeriesReisz} $\norm{v_{cK_{n}}^{*}}_{L_{2}\left(X\right)}^{2}\sim MJ_{n}$
for any $c\neq\t_{0}$.
\end{lem}
\begin{proof}
Note that
\end{proof}
\begin{align*}
\norm{v_{cK_{n}}^{*}}_{L_{2}\left(X\right)}^{2}= & \E\left[\left(v_{cK_{n}}^{*}\left(X_{i}\right)\right)^{2}\right]\\
= & \int_{x^{'}\t_{0}=0}\ol b^{'K_{n}}\left(x\right)c^{'}w\left(x\right)p\left(x\right)d{\cal H}^{d-1}\left(x\right)\\
 & \E\left[\ol b^{K_{n}}\left(X_{i}\right)\ol b^{K_{n}}\left(X_{i}\right)^{'}\right]^{-1}\int_{x^{'}\t_{0}=0}b^{K_{n}}\left(x\right)w\left(x\right)^{'}cp\left(x\right)d{\cal H}^{d-1}\left(x\right)\\
\leq & M\sum_{k=1}^{J_{n}}\left[\underset{=:T_{43}}{\underbrace{\int_{x^{'}\t_{0}=0}b_{k}\left(x\right)c^{'}w\left(x\right)p\left(x\right)d{\cal H}^{d-1}\left(x\right)}}\right]^{2}=:T_{44},
\end{align*}
Since $\t_{0}\neq0$, there exists some $j^{*}$ s.t. $\t_{0,j^{*}}\neq0$.
WLOG write $j^{*}=1$, and then 
\[
x^{'}\t_{0}=x_{1}\t_{0,1}+x_{-1}^{'}\t_{0,-1}=0\quad\iff\quad x_{1}=-\frac{x_{-1}^{'}\t_{0,-1}}{\t_{0,1}}.
\]
Hence, writing $\psi_{c}\left(x\right):=c^{'}w\left(x\right)p\left(\rest{x_{1}}x_{-1}\right)$,
we have
\begin{align*}
T_{43}=\  & \int_{x^{'}\t_{0}=0}\ol b_{k}\left(x\right)c^{'}w\left(x\right)p\left(x\right)d{\cal H}^{d-1}\left(x\right)\\
=\  & \int_{x^{'}\t_{0}=0}\ol b_{k}\left(x\right)\psi_{c}\left(x\right)p\left(x_{-1}\right)d{\cal H}^{d-1}\left(x\right)\\
=\  & \int\ol b_{k}\left(-\frac{x_{-1}^{'}\t_{0,-1}}{\t_{0,1}},x_{-1}\right)\psi_{c}\left(-\frac{x_{-1}^{'}\t_{0,-1}}{\t_{0,1}},x_{-1}\right)p\left(x_{-1}\right)dx_{-1}
\end{align*}
Since $\left(\ol b_{k}\right)$ is constructed as tensor products
of univariate $\left(b_{k}\right)$, for any $k\leq K_{n}$, there
exist some $k_{1},...,k_{d}\leq J_{n}$ such that
\[
\ol b_{k}\left(x\right)=b_{k_{1}}\left(x_{1}\right)b_{k_{2}}\left(x_{2}\right)...b_{k_{d}}\left(x_{d}\right).
\]
Hence, we can write 
\begin{align*}
T_{43}=\  & \int\left[b_{k_{1}}\left(-\frac{x_{-1}^{'}\t_{0,-1}}{\t_{0,1}}\right)\psi_{c}\left(-\frac{x_{-1}^{'}\t_{0,-1}}{\t_{0,1}},x_{-1}\right)\right]b_{k_{2}}\left(x_{2}\right)...b_{k_{d}}\left(x_{d}\right)p\left(x_{-1}\right)dx_{-1}\\
=\  & <b_{k_{1}}\psi_{c},b_{k,-1}>_{P_{X_{-1}}}
\end{align*}
where
\[
<m_{1},m_{2}>_{P_{X_{-1}}}:=\int m_{1}\left(x_{-1}\right)m_{2}\left(x_{-1}\right)p\left(x_{-1}\right)dx_{-1}
\]
denotes the natural inner product between functions $m_{1}\left(x_{-1}\right)$
and $m_{2}\left(x_{-1}\right)$ with respect to $P_{X_{-1}}$. 

Since $\left\{ \prod_{l=2}^{d}b_{k_{l}}:k_{l}=1,...,J_{n},l=2,...,d\right\} $
is a basis function for $L_{2}\left(X_{-1}\right)$, we have
\begin{align*}
 & \sum_{k_{2},...,k_{d}}\left(\int_{x^{'}\t_{0}=0}b_{k_{1}}\left(x\right)\prod_{l=2}^{d}b_{k_{l}}\left(x_{j}\right)c^{'}w\left(x\right)p\left(x\right)d{\cal H}^{d-1}\left(x\right)\right)^{2}\\
= & \sum_{k_{2},...,k_{d}}<b_{k_{1}}\psi_{c},\prod_{l=2}^{d}b_{k_{l}}>_{P_{X_{-1}}}^{2}\\
\leq & \norm{b_{k_{1}}\psi_{j}}_{L_{2}\left(X_{-1}\right)}^{2}\text{ by the Bessel's inequality}\\
= & \int b_{k_{1}}^{2}\left(-\frac{x_{-1}^{'}\t_{0,-1}}{\t_{0,1}}\right)\psi_{c}^{2}\left(-\frac{x_{-1}^{'}\t_{0,-1}}{\t_{0,1}},x_{-1}\right)p\left(x_{-1}\right)dx_{-1}\\
= & \int_{x^{'}\t_{0}=0}b_{k_{1}}^{2}\left(x_{1}\right)c^{'}w\left(x\right)w\left(x\right)'cp^{2}\left(\rest{x_{1}}x_{-1}\right)p\left(x_{-1}\right)d{\cal H}^{d-1}\left(x\right)\\
= & c^{'}\int_{x^{'}\t_{0}=0}b_{k_{1}}^{2}\left(x_{1}\right)w\left(x\right)w\left(x\right)'p^{2}\left(\rest{x_{1}}x_{-1}\right)p\left(x\right)d{\cal H}^{d-1}\left(x\right)c\\
\leq & \norm c^{2}M=M
\end{align*}
with ``$\leq$'' replaced by ``$\sim$'' whenever $c\neq\t_{0}$.
Hence, 
\begin{align*}
T_{44} & =M\sum_{k=1}^{J_{n}}\left[\int_{x^{'}\t_{0}=0}b_{k}\left(x\right)c^{'}w\left(x\right)p\left(x\right)d{\cal H}^{d-1}\left(x\right)\right]^{2}\\
 & =M\sum_{k_{1}}\left[\sum_{k_{2},...,k_{d}}\left(\int_{x^{'}\t_{0}=0}b_{k_{1}}\left(x\right)\prod_{l=2}^{d}b_{k_{l}}\left(x_{j}\right)c^{'}w\left(x\right)p\left(x\right)d{\cal H}^{d-1}\left(x\right)\right)^{2}\right]\\
 & \leq M\sum_{k_{1}=1}^{J_{n}}M=M^{2}J_{n}
\end{align*}
and thus, for some $M$,
\[
\norm{v_{cK_{n}}^{*}}_{L_{2}\left(X\right)}^{2}\leq MJ_{n}
\]
with ``$\leq$'' replaced by ``$\sim$'' whenever $c\neq\t_{0}$. 

\subsubsection*{Proof of Lemma \ref{lem:T4_Series}}
\begin{proof}
We first apply Theorem 3.1 of \citet*{chen2015optimal} for $L_{c}\left(\hat{h}\right)$
to $L_{c}\left(h\right)$. Clearly, Assumptions 1(i), 2(i)(ii)(iv)(v),
4(iii) in \citet*{chen2015optimal} are satisfied in the current paper
given Assumptions \ref{assu:Basic} and \ref{assu:KernelSeries}. Assumption
9 in \citet*{chen2015optimal} follows from the sufficient conditions
in Remark 3.1 in \citet*{chen2015optimal} given our Assumption \ref{assu:KernelSeries}(b.ii)(b.iii),
the undersmoothing condition $J_{n}^{-s}=o\left(\sqrt{\frac{J_{n}}{n}}\right)$
implied by \emph{$J_{n}^{-1}=o\left(n^{-\frac{1}{2s+1}}\right)$,
}and Lemma \ref{lem:LinSeriesReisz}. Hence, we have
\[
\frac{\sqrt{n}\left(L_{c}\left(\hat{h}\right)-L_{c}\left(h\right)\right)}{\sqrt{\O_{cK_{n}}}}\dto\cN\left(0,1\right).
\]
Defineing $\O:=\lim_{n\to\infty}\frac{1}{J_{n}}\O_{K_{n}}$ and $\O_{c}:=c^{'}\O c$,
we have
\[
\sqrt{nJ_{n}^{-1}}\left(L_{c}\left(\hat{h}\right)-L_{c}\left(h\right)\right)\dto\cN\left(0,\O_{c}=c^{'}\O c\right).
\]
Since the above holds for any $c\in\S^{d-1}$, we have
\[
\sqrt{nJ_{n}^{-1}}\left(L_{c}\left(\hat{h}\right)-L_{c}\left(h\right)\right)\dto\cN\left(0,\O\right).
\]
\end{proof}

\subsection{\label{subsec:Pf_Rate}Proof of Theorem \ref{thm:Thm_Bin_Rate}}
\begin{proof}
For consistency, we observe that
\[
\sup_{\t\in\T}\sup_{h\in{\cal H}}\left|\P_{n}g_{\t,h}-Pg_{\t,h}\right|=o_{p}\left(1\right).
\]
since $\cG$ is Gilvenko-Cantelli. Moreover, 
\begin{align*}
\sup_{\t\in\T}\sup_{\norm{h-h_{0}}_{\infty}\leq\e}\left|Pg_{\t,h}-Pg_{\t,h_{0}}\right| & \leq P\left(\left|h-h_{0}\right|\right)\leq\e\to0\quad\text{as }\d\to0.
\end{align*}
As $\norm{\hat{h}-h_{0}}_{\infty}=o_{p}\left(1\right)$ and $\hat{h}\in{\cal H}$
with probability approaching 1 by Assumption \ref{assu:FirstStageRate},
we conclude by Theorem 1 of \citet[DvK thereafter]{delsol2020semiparametric}
that $\norm{\hat{\t}-\t_{0}}=o_{p}\left(1\right).$ 

To derive the rate of convergence for $\hat{\t}$, we apply Theorem
2 of DvK by verifying their Conditions B1-B4. We present the results
below using the notation for kernel bandwidth ``$b_{n}$'' to represent
the tuning parameter in the first-stage nonparamatric estimation,
but note that the proof goes through for linear series estimators
as well with $b_{n}$ replaced by $1/J_{n}.$

Recall that $\norm{\hat{h}-h_{0}}_{\infty}=O_{p}\left(a_{n}=\sqrt{\frac{\log n}{nb_{n}^{d}}}+b_{n}^{s}\right)$
and $\norm{\Dif_{x}\hat{h}-\Dif_{x}h_{0}}_{\infty}=O_{p}\left(c_{n}=\sqrt{\frac{\log n}{nb_{n}^{d+2}}}+b_{n}^{s}\right)$.
See, for example, \citet*{hansen2008uniform} for such results on
the sup-norm convergence rate for kernel estimator $\hat{h}$ and
\citet*{chen2015optimal} for linear series $\hat{h}$. To guarantee
that the term $a_{n}c_{n}=o_{p}\left(\norm{\hat{\t}-\t_{0}}\right)$,
we need to ensure
\[
\norm{\hat{h}-h_{0}}_{\infty}\norm{\Dif_{x}\hat{h}-\Dif_{x}h_{0}}_{\infty}=\left(\sqrt{\frac{\log n}{nb_{n}^{d}}}+b_{n}^{s}\right)\left(\sqrt{\frac{\log n}{nb_{n}^{d+2}}}+b_{n}^{s}\right)=o_{p}\left(\frac{1}{\sqrt{nb_{n}}}+b_{n}^{s}\right),
\]
which is satisfied if 
\[
\sqrt{\frac{\log n}{nb_{n}^{d}}}\sqrt{\frac{\log n}{nb_{n}^{d+2}}}=o_{p}\left(\frac{1}{\sqrt{nb_{n}}}\right).
\] This can be ensured by  $\frac{\log n}{nb_{n}^{d}}\cd\sqrt{nb_{n}}\to0$,
or equivalently, $nb_{n}^{2d+1}/\left(\log n\right)^{2}\to\infty$,
as imposed in the statement of the theorem. In addition, to guarantee
that $\norm{\Dif_{x}\hat{h}-\Dif_{x}h_{0}}_{\infty}=o_{p}\left(1\right)$
as in Assumption \ref{assu:FirstStageDeriv}, we need $\sqrt{\frac{\log n}{nb_{n}^{d+2}}}=o_{p}\left(1\right)$,
i.e., $nb_{n}^{d+2}/\log n\to\infty$, which is also implied by $nb_{n}^{2d+1}/\left(\log n\right)^{2}\to\infty$.

B1 directly follows from the consistency of $\hat{\t}$ and the assumption
that$\norm{\hat{h}-h_{0}}_{\infty}=O_{p}\left(a_{n}\right)$.

For their Condition B2, observe that
\[
\GG_{n}\left(g_{\t,h}-g_{\t_{0},h}\right)=\GG_{n}\left(g_{\t,h_{0}}-g_{\t_{0},h_{0}}\right)+\GG_{n}\left(g_{\t,h}-g_{\t_{0},h}-g_{\t,h_{0}}+g_{\t_{0},h_{0}}\right)
\]
and thus, by Lemmas \ref{lem:Term1} and \ref{lem:Term2}, 
\[
P\sup_{\norm{\t-\t_{0}}\leq\d,\norm{h-h_{0}}_{\infty}\leq Ka_{n}}\left|\GG_{n}\left(g_{\t,h}-g_{\t_{0},h}\right)\right|\leq M\left(\sqrt{\d}+\sqrt{a_{n}}\right)\d.
\]
so that $\Phi_{n}\left(\d\right)=\left(\sqrt{\d}+1\right)\d$
in the notation of DvK.

Letting $\norm{\hat{\t}-\t_{0}}:=O_{p}\left(\d_{n}\right)$, we seek
to find the smallest $\d_{n}$ that verifies Condition B3 and B4 in
DvK\footnote{$\d_{n}=r_{n}^{-1}$ in DvK's notation.}. For Condition
B4 to hold, i.e., for
\begin{align*}
\frac{1}{\d_{n}^{2}}\Phi_{n}\left(\d_{n}\right) & =\frac{1}{\d_{n}^{2}}\left(\sqrt{\d_{n}}+1\right)\d_{n}=\d_{n}^{-\frac{1}{2}}+\d_{n}^{-1},
\end{align*}
to be $O\left(\sqrt{n}\right),$we need
\[
\d_{n}^{-\frac{1}{2}}\leq\sqrt{n},\quad \d_{n}^{-1}\leq\sqrt{n},
\]
which is satisfied as long as 
\[
\quad\frac{1}{\sqrt{n}}=o\left(\d_{n}\right).
\]
As a result, B4 is always satisfied provided that $\d_{n}$ is converging
no faster than the standard $n^{-\frac{1}{2}}$ rate.

Setting $\d_{n}\sim\frac{1}{\sqrt{nb_{n}}}+b_{n}^{s}$, we note that
B3 in DvK is satisfied with
\[
W_{n}:=\int_{x^{'}\t_{0}=0}\left[\hat{h}\left(x\right)-h_{0}\left(x\right)\right]\frac{1}{f\left(\rest 0x\right)+1}xp\left(x\right)d{\cal H}^{d-1}\left(x\right)
\]

To make $\d_{n}$ as small as possible, we set $b_{n}^{*}$ to solve
\[
\frac{1}{\sqrt{nb_{n}^{*}}}\sim b_{n}^{*s}\quad\iff b_{n}^{*}\sim n^{-\frac{1}{2s+1}},
\]
which delivers 
\[
\d_{n}^{*}=n^{-\frac{s}{2s+1}}.
\]
Note that we need to ensure that $nb_{n}^{2d+1}/\left(\log n\right)^{2}\to\infty$
holds with $b_{n}^{*}\sim n^{-\frac{1}{2s+1}}$, which is satisfied
if 
\[
1-\frac{2d+1}{2s+1}>0\quad\iff\quad2s+1>2d+1\quad\iff\quad s>d.
\]
\end{proof}

\subsection{Proof of Theorem \ref{thm:AsymNorm}}
\begin{proof}
We apply Theorem 3.2.16 of \citet*{van1996weak} with $\mathbb{M}_{n}\left(\t\right):=\P_{n}g_{\t,\hat{h}}$,
$\mathbb{M}\left(\t\right):=-\left(\t-\t_{0}\right)^{'}V\left(\t-\t_{0}\right)$
and $r_{n}:=\sqrt{nb_{n}}$ (for kernel first-stage estimators) or
$\sqrt{nJ_{n}^{-1}}$ (for linear series first-stage estimators) with
undersmoothing choice of $b_{n}$ or $J_{n}$ so that $a_{n}c_{n}=o_{p}\left(\norm{\hat{\t}-\t_{0}}\right)$.

Plugging Lemmas \ref{lem:Term1}, \ref{lem:Term2}, \ref{lem:Term3},
and \ref{lem:Term4} into the decomposition \eqref{eq:EP_decom},
we have
\begin{align*}
0\leq\  & \mathbb{M}_{n}\left(\tilde{\t}\right)-\mathbb{M}_{n}\left(\t\right)=\P_{n}\left(g_{\tilde{\t}\hat{h}}-g_{\t_{0},\hat{h}}\right)\\
=\  & -\left(\tilde{\t}-\t_{0}\right)^{'}V\left(\tilde{\t}-\t_{0}\right)+Z_{n}^{'}\left(\tilde{\t}-\t_{0}\right)+o_{p}\left(\norm{\tilde{\t}-\t_{0}}^{2}\right)
\end{align*}
with $Z_{n}:=D\left[P\Dif_{\t}g_{\t_{0},h_{0}},\hat{h}-h_{0}\right]=O_{p}\left(r_{n}^{-1}\right)$
and
\[
r_{n}Z_{n}\dto\cN\left({\bf 0},\O\right).
\]
Hence, the key condition in Theorem 3.2.16 of \citet*{van1996weak}
can be verified with
\begin{align*}
 & r_{n}\left(\mathbb{M}_{n}-\mathbb{M}\right)\left(\tilde{\t}_{n}\right)-r_{n}\left(\mathbb{M}_{n}-\mathbb{M}\right)\left(\t_{0}\right)\\
=\  & \left(r_{n}Z_{n}\right)^{'}\left(\tilde{\t}-\t_{0}\right)+o_{p}\left(r_{n}\norm{\tilde{\t}-\t_{0}}^{2}\right)
\end{align*}
for any $\tilde{\t}_{n}$ s.t. $\norm{\tilde{\t}_{n}-\t_{0}}=O_{p}\left(r_{n}^{-1}\right)$.
Hence, 
\[
r_{n}\left(\hat{\t}-\t_{0}\right)=V^{-}r_{n}Z_{n}+o_{p}\left(1\right),
\]
and
\[
r_{n}\left(\hat{\t}-\t_{0}\right)\dto\cN\left({\bf 0},V^{-}\O V^{-}\right).
\]
\end{proof}

\subsection{Proof of Lemma 5}
\begin{proof}
(i) Fix $c\in\mathbb{R}^d$ and write
\[
\Gamma_c(h)
= c' P\nabla_\theta g_{\theta_0,h}
= \int \psi_c(x,h(x),\theta_0)\, p(x)\,dx,
\]
where $\psi_c$ is obtained by differentiating $g_{\theta,h}(x)$ with respect
to $\theta$ at $\theta=\theta_0$ and contracting with $c$. By the definition
of $g_{\theta,h}$ in \eqref{eq:def_g_J}, each component
of $\psi_c$ is a finite linear combination of indicator functions of regions of
the form
\[
\Bigl\{
\min\bigl(h(x),\min_{k\neq j}(-x_k'\theta_0)\bigr)\ge -x_j'\theta_0\ge 0
\Bigr\}
\]
and their analogues for the negative part, multiplied by $x_j$ or $-x_j$. In
particular, as a function of $(h,\theta)$, $\psi_c(x,h(x),\theta)$ is Lipschitz
and piecewise affine in $h(x)$.

Consider a path $h_t := h_0 + t v$ with $t\in\mathbb{R}$ small and $v\in H$.
Then
\[
\frac{\Gamma_c(h_t) - \Gamma_c(h_0)}{t}
= \int \frac{\psi_c(x,h_0(x)+t v(x),\theta_0)
            -\psi_c(x,h_0(x),\theta_0)}{t}\, p(x)\,dx.
\]
For each fixed $x$ such that $x_j'\theta_0\neq0$ for all $j$, the integrand
is eventually constant in $t$ near zero, because the inequalities defining the
regions above do not change sign when $t$ is small. Thus the pointwise
derivative with respect to $h$ exists and
\[
\dot\psi_c(x)
:= \partial_h \psi_c(x,h_0(x),\theta_0)
\]
is nonzero only when some index $j$ is near binding, i.e. when $x_j'\theta_0$
is close to zero and the composite ReLU terms kink. By dominated convergence,
we may differentiate under the integral sign to obtain
\[
D_h\Gamma_c(h_0)[v]
= \int v(x)\, \dot\psi_c(x)\, p(x)\,dx.
\]

To rewrite this as an integral over the submanifolds $\{x:x_j'\theta_0=0\}$,
note that on each branch where a particular index $j$ is the minimum and
binds, the contribution of $\dot\psi_c(x)$ depends only on $x_j$ and the
sign pattern of the remaining indexes. Under the strict MISC condition, the
boundary of the region where index $j$ changes sign is exactly the hyperplane
\[
\{x:x_j'\theta_0=0\}.
\]
Applying the coarea formula (or submanifold integral formula) to the scalar
level-set map $x\mapsto x_j'\theta_0$ then yields
\[
\int v(x)\, \dot\psi_c(x)\, p(x)\,dx
= \sum_{j=1}^J \int_{\{x:x_j'\theta_0=0\}} v(x)\, w_{c,j}(x)\,
   d\mathcal{H}^{d-1}(x),
\]
for some weights $w_{c,j}(x)$ that are continuous and uniformly bounded on
$\{x:x_j'\theta_0=0\}$ by the continuity of $p$, $h_0$, and the MISC
structure. This gives \eqref{eq:MISC-submanifold-deriv}. The representation
for the vector functional $L(h)$ follows by taking $c$ equal to each canonical
basis vector and stacking the resulting derivatives.

(ii) For the quadratic expansion, write
\[
Q(\theta) = P g_{\theta,h_0}.
\]
Since $g_{\theta,h_0}$ is Lipschitz in $\theta$ and piecewise affine, $Q$ is
twice differentiable at $\theta_0$ and we may apply a second-order Taylor
expansion around $\theta_0$:
\[
Q(\theta) - Q(\theta_0)
= (\theta-\theta_0)' \partial_\theta Q(\theta_0)
 + \tfrac12 (\theta-\theta_0)' \partial^2_{\theta\theta} Q(\tilde\theta)
   (\theta-\theta_0),
\]
for some $\tilde\theta$ on the segment between $\theta$ and $\theta_0$. By
construction of the RMS criterion and the MISC sign-alignment, $\theta_0$ is
a maximizer of $Q$ on the unit sphere, so the gradient vanishes:
$\partial_\theta Q(\theta_0)=0$.

It remains to characterize the Hessian. Differentiating $Q(\theta)=Pg_{\theta,h_0}$
twice with respect to $\theta$ and using the same type of argument as in part
(i), one finds that the second derivative at $\theta_0$ can be written as
\[
\partial^2_{\theta\theta} Q(\theta_0)
= -2 V,
\]
where $V$ is given by \eqref{eq:MISC-V-rep}. The key step is that the
second derivative of the composite ReLU terms is supported only on the
hyperplanes where the inner arguments kink, namely $\{x:x_j'\theta_0=0\}$,
and that, on those sets, the curvature in the direction $\theta-\theta_0$ is
proportional to $x_j x_j'$ with a nonnegative weight $m_j(x,\theta_0)$
capturing the local density and slope of the model primitives. Integrating
these contributions over the hyperplanes yields \eqref{eq:MISC-V-rep}.

Substituting back into the Taylor expansion gives
\[
Q(\theta) - Q(\theta_0)
= -(\theta-\theta_0)' V (\theta-\theta_0)
  + r_Q(\theta),
\]
where the remainder satisfies $r_Q(\theta)=o(\|\theta-\theta_0\|^2)$ as
$\theta\to\theta_0$ by continuity of the second derivative in a neighborhood
of $\theta_0$. The fact that $V$ is positive semidefinite and has rank $d-1$
with $V\theta_0=0$ follows from the support properties of $x_j$ on the
hyperplanes and the scale normalization of $\theta_0$. This yields
\eqref{eq:MISC-quadratic}–\eqref{eq:MISC-V-rep}.
\end{proof}

\subsection{Proof of Theorem \ref{thm:Asymp-MISC}}

We follow the structure of the proofs for the results in Section \ref{sec:BinChoice}.
Recall that $Q:=\E\left[g_{\t,h}\left(X_{i}\right)\right]$ with $g_{\t,h}:=g_{+,\t,h}+g_{-,\t,h}$
and 
\begin{align*}
g_{+,\t,h}\left(x\right):=\left[h\left(x\right)-\left[\min_{j=1,...,J}\left(-x_{j}^{'}\t\right)\right]_{+}\right]_{+},\quad & g_{-,\t,h}\left(x\right):=\left[-h\left(x\right)-\left[\min_{j=1,...,J}\left(x_{j}^{'}\t\right)\right]_{+}\right]_{+}.
\end{align*}
Given a first-stage nonparametric estimator $\hat{h}$ of $h_{0}$,
the sample criterion is constructed as 
\[
\hat{Q}\left(\t\right)=\frac{1}{n}\sum_{i=1}^{n}g_{\t,\hat{h}}\left(X_{i}\right)\equiv\P_{n}g_{\t,\hat{h}}.
\]

Again, consider the following decomposition 

\begin{align}
\P_{n}\left(g_{\hat{\t},\hat{h}}-g_{\t_{0},\hat{h}}\right)= & \underset{T_{1}}{\underbrace{\frac{1}{\sqrt{n}}\GG_{n}\left(g_{\hat{\t},h_{0}}-g_{\t_{0},h_{0}}\right)}}+\underset{T_{2}}{\underbrace{\frac{1}{\sqrt{n}}\GG_{n}\left(g_{\hat{\t},\hat{h}}-g_{\t_{0},\hat{h}}-g_{\hat{\t},h_{0}}+g_{\t_{0},h_{0}}\right)}}\nonumber \\
 & +\underset{T_{3}}{\underbrace{P\left(g_{\hat{\t},h_{0}}-g_{\t_{0},h_{0}}\right)}}+\underset{T_{4}}{\underbrace{P\left(g_{\hat{\t},\hat{h}}-g_{\t_{0},\hat{h}}-g_{\hat{\t},h_{0}}+g_{\t_{0},h_{0}}\right)}}\label{eq:EP_decom-1}
\end{align}

\begin{lem}
\label{lem:Term1-MISC} For some constant $M>0$,
\begin{equation}
P\sup_{\norm{\t-\t_{0}}\leq\d}\left|\GG_{n}\left(g_{\t,h_{0}}-g_{\t_{0},h_{0}}\right)\right|\leq M\d^{\frac{3}{2}}.\label{eq:Max1-1}
\end{equation}
\end{lem}

\begin{lem}
\label{lem:Term2-MISC} Under Assumptions \ref{assu:Basic}-\ref{assu:FirstStageRate},
for some constant $M>0$, 
\begin{equation}
P\sup_{\t\in\T,h\in{\cal H}:\norm{\t-\t_{0}}\leq\d,\norm{h-h_{0}}_{\infty}\leq Ka_{n}}\left|\GG_{n}\left(g_{\t,h}-g_{\t_{0},h}-g_{\t,h_{0}}+g_{\t_{0},h_{0}}\right)\right|\leq M\d.\label{eq:Max2-1}
\end{equation}
\end{lem}

~
We now present the main proof based on the lemmas above.
~
\begin{proof}
(a) For any fixed $c\in\mathbb{S}^{d-1}$, Lemma~\ref{lem:MISC-curv-L}(ii) and
Assumption~\ref{ass:MISC-CG}(c) imply that the scalar functional
$\Gamma_c(h)=c'P\nabla_\theta g_{\theta_0,h}$ satisfies the linearization
assumptions of Chen and Gao (2025, Assumptions~9–11) with submanifold
dimension $m=d-1$. Together with the sieve and smoothness conditions in
Assumption~\ref{ass:MISC-CG}, Theorems~2 and~3 of Chen and Gao (2025)
then yield
\[
c_n\big(\Gamma_c(\hat h) - \Gamma_c(h_0)\big)
\;\xrightarrow{d}\;
\mathcal{N}(0,\sigma_c^2),
\]
for some finite variance $\sigma_c^2$, and
$\Gamma_c(\hat h) - \Gamma_c(h_0) = O_p(c_n^{-1})$. Since this holds for all
$c$ and $L(h)$ is obtained by stacking such scalar functionals, we obtain the
rate and multivariate CLT in \eqref{eq:Lhat-CLT} with some covariance matrix
$\Omega$.

(b) By Lemma \ref{lem:Term1-MISC},
\[
T_1 = \frac{1}{\sqrt{n}}\GG_n\bigl(g_{\hat\t,h_0}-g_{\t_0,h_0}\bigr)
= o_p\bigl(\|\hat\t-\t_0\|\bigr).
\]
By Lemma \ref{lem:Term2-MISC},
\[
T_2 = o_p\bigl(\|\hat\t-\t_0\|\bigr)
\qquad\text{whenever } c_n\|\hat\t-\t_0\|\to\infty.
\]
By Lemma \ref{lem:MISC-curv-L},  we have the local quadratic expansion
\[
T_3 = P(g_{\hat\t,h_0}-g_{\t_0,h_0})
= -(\hat\t-\t_0)'V(\hat\t-\t_0) + o_p(\|\hat\t-\t_0\|^2),
\]
where $V$ is symmetric positive semidefinite of rank $d-1$ and $V\t_0=0$.

Finally, by Lemma \ref{lem:MISC-curv-L}, Assumption \ref{ass:MISC-CG}, 
\[
T_4
= (\hat\t-\t_0)'L(\hat h-h_0) + o_p\bigl(\|\hat\t-\t_0\|c_n^{-1}\bigr),
\]
and, by Theorem 3 of \cite{chen2025semiparametric},
\[
c_n L(\hat h-h_0) \dto \mathcal{N}(A,\Omega).
\]
%\textbf{Step 1: Rate}

Insert the bounds for $T_1$--$T_4$ into \eqref{eq:EP_decom-1}, we have
\[
0
\le -(\hat\t-\t_0)'V(\hat\t-\t_0)
+ (\hat\t-\t_0)'L(\hat h-h_0)
+ o_p\!\bigl(\|\hat\t-\t_0\|^2 + \|\hat\t-\t_0\|c_n^{-1}\bigr).
\]
By part (a), we have $L(\hat h-h_0)=O_p(c_n^{-1})$, so the second term is $O_p(\|\hat\t-\t_0\|c_n^{-1})$.
Since $V$ is positive definite on $\t_0^\perp$, the display implies
\[
\|\hat\t-\t_0\|^2
\lesssim_p \|\hat\t-\t_0\|c_n^{-1} + o_p\!\bigl(\|\hat\t-\t_0\|^2 + \|\hat\t-\t_0\|c_n^{-1}\bigr),
\]
which implies $\|\hat\t-\t_0\| = O_p(c_n^{-1})$.

\medskip
%\noindent
%\textbf{Step 2: Linearization.}

Using $\|\hat\t-\t_0\| = O_p(c_n^{-1})$ and plugging this rate back into \eqref{eq:EP_decom-1} yields
\[
0
= -(\hat\t-\t_0)'V(\hat\t-\t_0)
+ (\hat\t-\t_0)'L(\hat h-h_0)
+ o_p(c_n^{-2}).
\]
Rearranging,
\[
(\hat\t-\t_0)'V(\hat\t-\t_0)
= (\hat\t-\t_0)'L(\hat h-h_0)
+ o_p(c_n^{-2}).
\]
Since $\hat\t-\t_0 = O_p(c_n^{-1})$ and $V$ is nonsingular on $\t_0^\perp$, the last display implies
\[
V(\hat\t-\t_0) = L(\hat h-h_0) + o_p(c_n^{-1}),
\]
and hence
\begin{equation}
\hat\t-\t_0 = V^{-}L(\hat h-h_0) + o_p(c_n^{-1}),
\label{eq:AsymL-linear}
\end{equation}
where $V^{-}$ denotes the Moore–Penrose inverse of $V$.

\medskip
%\noindent
%\textbf{Step 3: Limit distribution.}

Multiplying \eqref{eq:AsymL-linear} by $c_n$, we have
\[
c_n(\hat\t-\t_0)
= V^{-}\bigl(c_n L(\hat h-h_0)\bigr) + o_p(1)
\dto \mathcal{N}\bigl(0,V^{-}\Omega V^{-}\bigr).
\]

\end{proof}

\subsection{Proof of Lemma \ref{lem:Term1-MISC}}
\begin{proof}
Observe that $g_{+,\t_{0},h_{0}}\left(x\right)=\left[h_{0}\left(x\right)-\left[\min_{j}\left(-x_{j}^{'}\t\right)\right]_{+}\right]_{+}\equiv\left[h_{0}\left(x\right)\right]_{+}$
and
\[
g_{+,\t,h_{0}}\left(x\right)-g_{+,\t_{0},h_{0}}\left(x\right)=\left[h_{0}\left(x\right)-\left[\min_{j}\left(-x_{j}^{'}\t\right)\right]_{+}\right]_{+}-\left[h_{0}\left(x\right)\right]_{+}
\]
which is nonzero only if $h_{0}\left(x\right)>0$ while $x_{j}^{'}\t<0$
for all $j$.

Now, consider any $x$ s.t. $g_{+,\t,h_{0}}\left(x\right)\neq g_{+,\t_{0},h_{0}}\left(x\right)$.
Since $h_{0}\left(x\right)>0$, by the contraposition of the MISC
condition \eqref{eq:MISC} we know that there exists some $j^{*}$
such that $x_{j^{*}}^{'}\t_{0}>0$. Then, we have
\begin{align*}
x_{j^{*}}^{'}\t_{0}>0>x_{j^{*}}^{'}\t & =x_{j^{*}}^{'}\t_{0}+x_{j^{*}}^{'}\left(\t-\t_{0}\right)>x_{j^{*}}^{'}\t_{0}-\norm{x_{j^{*}}}\norm{\t-\t_{0}},
\end{align*}
and hence
\begin{align}
0 & <x_{j^{*}}^{'}\t_{0}<\norm{x_{j^{*}}}\norm{\t-\t_{0}}\label{eq:xjstar_t_sandwich}
\end{align}
and
\[
0<-x_{j^{*}}^{'}\t<\norm{x_{j^{*}}}\norm{\t-\t_{0}}\leq M\norm{\t-\t_{0}},
\]
which further implies that
\begin{equation}
\left|g_{+,\t,h_{0}}\left(x\right)-g_{+,\t_{0},h_{0}}\left(x\right)\right|\leq\left[\min_{j}\left(-x_{j}^{'}\t\right)\right]_{+}\leq-x_{j^{*}}^{'}\t\leq M\norm{\t-\t_{0}}.\label{eq:g_diff_small}
\end{equation}

Now, for any $x\in{\cal X},$by \eqref{eq:xjstar_t_sandwich} we have

\[
\ind\left\{ g_{+,\t,h_{0}}\left(x\right)-g_{+,\t_{0},h_{0}}\left(x\right)\neq0\right\} \leq\sum_{j=1}^{J}\ind\left\{ 0<x_{j}^{'}\t_{0}<\norm{x_{j}}\norm{\t-\t_{0}}\right\} .
\]
Combining the above with \eqref{eq:g_diff_small}, we have
\begin{align*}
\left|g_{+,\t,h_{0}}\left(x\right)-g_{+,\t_{0},h_{0}}\left(x\right)\right| & \leq\ind\left\{ g_{+,\t,h_{0}}\left(x\right)-g_{+,\t_{0},h_{0}}\left(x\right)\neq0\right\} \left|g_{+,\t,h_{0}}\left(x\right)-g_{+,\t_{0},h_{0}}\left(x\right)\right|\\
 & \leq\sum_{j=1}^{J}\ind\left\{ 0<x_{j}^{'}\t_{0}<\norm{x_{j}}\norm{\t-\t_{0}}\right\} M\norm{\t-\t_{0}}
\end{align*}

For $g_{-,\t,h_{0}}$, similar arguments as above give 
\begin{align*}
\left|g_{+,\t,h_{0}}\left(x\right)-g_{+,\t_{0},h_{0}}\left(x\right)\right| & \leq\sum_{j=1}^{J}\ind\left\{ -\norm{x_{j}}\norm{\t-\t_{0}}<0<x_{j}^{'}\t_{0}\right\} M\norm{\t-\t_{0}}
\end{align*}
and hence
\begin{align*}
\left|g_{\t,h_{0}}\left(x\right)-g_{\t_{0},h_{0}}\left(x\right)\right| & \leq M\sum_{j=1}^{J}\ind\left\{ \left|x_{j}^{'}\t_{0}\right|\leq\norm{x_{j}}\norm{\t-\t_{0}}\right\} \norm{\t-\t_{0}}
\end{align*}

Define ${\cal \cG}_{1,\d}:=\left\{ g_{\t,h_{0}}-g_{\t_{0},h_{0}}:\ \norm{\t-\t_{0}}\leq\d\right\} .$
By the arguments above, ${\cal \cG}_{1,\d}$ has an envelope $G_{1,\d}$
given by
\[
G_{1,\d}:=M\d\sum_{j=1}^{J}\ind\left\{ \left|x_{j}^{'}\t_{0}\right|\leq\norm{x_{j}}\norm{\t-\t_{0}}\right\} 
\]
with
\begin{align*}
PG_{1,\d}^{2} & =M^{2}\d^{2}\E\left[\left(\sum_{j=1}^{J}\ind\left\{ \left|x_{j}^{'}\t_{0}\right|\leq\norm{x_{j}}\norm{\t-\t_{0}}\right\} \right)^{2}\right].\\
 & \leq M\d^{2}J\sum_{j=1}^{J}\P\left(\left|\frac{X_{ij}^{'}}{\norm{X_{ij}}}\t_{0}\right|\leq\d\right)\leq M\d^{2}J\sum_{j=1}^{J}M\d\leq M\d^{3}
\end{align*}
Now, since ${\cal \cG}_{1,\d}\subseteq\cG$, we have $\mathscr{N}\left(\e,\cG_{1,\d},L_{2}\left(P\right)\right)\leq\mathscr{N}\left(\e,\cG,L_{2}\left(P\right)\right)$
\[
J_{1,\d}:=\int_{0}^{1}\sqrt{1+\log\mathscr{N}\left(\e,\cG_{1,},L_{2}\left(P\right)\right)}d\e\leq J<\infty.
\]
Then, by VW Theorem 2.14.1, we have
\[
P\sup_{g\in{\cal \cG}_{1,\d}}\left|\GG_{n}\left(g\right)\right|\leq J_{1,\d}\sqrt{PG_{1,\d}^{2}}\leq J_{1}C\d^{\frac{3}{2}}=C_{1}\d^{\frac{3}{2}}.
\]
\end{proof}

\subsection{Proof of Lemma \ref{lem:Term2-MISC}}

\begin{proof}
Observe first that
\[
\left|g_{+,\t,h}\left(x\right)-g_{+,\t_{0},h}\left(x\right)-g_{+,\t,h_{0}}\left(x\right)+g_{+,\t_{0},h_{0}}\left(x\right)\right|\leq2\left|\min_{j}\left(-x_{j}^{'}\t\right)-\min_{j}\left(-x_{j}^{'}\t_{0}\right)\right|
\]
Then observe that, for any $\left(c_{1},...,c_{J}\right)$ and $\left(c_{1}^{'},...,c_{J}^{'}\right)$,
we have
\[
\left|\min_{j}c_{j}-\min_{j}c_{j}^{'}\right|\leq\max_{j}\left|c_{j}-c_{j}^{'}\right|.
\]
Hence, 
\[
\left|g_{+,\t,h}\left(x\right)-g_{+,\t_{0},h}\left(x\right)-g_{+,\t,h_{0}}\left(x\right)+g_{+,\t_{0},h_{0}}\left(x\right)\right|\leq2\max_{j}\left|x_{j}^{'}\left(\t-\t_{0}\right)\right|\leq M\norm{\t-\t_{0}}.
\]
The similar also holds for $g_{-}$ and $g$.

Define ${\cal \cG}_{2,\d}:=\left\{ g_{\t,h}-g_{\t_{0},h}-g_{\t,h_{0}}+g_{\t_{0},h_{0}}:\ \norm{\t-\t_{0}}\leq\d,h\in{\cal H}\right\} .$
By the arguments above, ${\cal \cG}_{2,\d}$ has an envelope $G_{2,\d}$
given by $G_{2,\d}:=M\d$ with 
\[
PG_{2,n,\d}^{2}=M^{2}\d^{2}.
\]
By VW Theorem 2.14.1, we have
\[
P\sup_{g\in{\cal \cG}_{2,\d}}\norm{\GG_{n}\left(g\right)}\leq J_{2,\d}\sqrt{PG_{2,\d}^{2}}\leq M\d.
\]
\end{proof}

\end{document}